\documentclass[11pt]{article}

\usepackage[margin=1in,headheight=14pt,footskip=30pt]{geometry}
\usepackage{amsmath,amssymb,amsthm}
\usepackage{array,graphicx}
\usepackage{tikz}
\usepackage{pdflscape}
\usetikzlibrary{arrows.meta,positioning,calc,fit,backgrounds}
\providecommand{\toprule}{\hline}\providecommand{\midrule}{\hline}\providecommand{\bottomrule}{\hline}
\usepackage[hidelinks]{hyperref}
\hypersetup{
  pdftitle={Step Recursion: A Three-Parameter Refinement of the Grzegorczyk Hierarchy},
  pdfauthor={Kirill Osipov},
  pdfsubject={Step recursion and refined Grzegorczyk classes},
  pdfkeywords={bounded recursion, step recursion, Grzegorczyk hierarchy, generalized inverses, trace sparsity}
}

\newcommand{\Hclass}[3]{\mathcal H^{#1}_{#2,#3}}

\newcommand{\Eclass}[1]{\mathcal E^{#1}}

\newcommand{\Nat}{\mathbb N}
\newcommand{\NZ}{\operatorname{NZ}}
\newcommand{\ISZ}{\operatorname{ISZ}}
\newcommand{\dotminus}{\mathbin{\dot{-}}}
\newcommand{\sem}[1]{\mathopen{\lbrack\!\lbrack}#1\mathclose{\rbrack\!\rbrack}}
\newtheorem{theorem}{Theorem}
\newtheorem{corollary}[theorem]{Corollary}
\newtheorem{lemma}[theorem]{Lemma}
\newtheorem{proposition}[theorem]{Proposition}
\newtheorem{remark}[theorem]{Remark}
\newtheorem{definition}[theorem]{Definition}

\title{Step Recursion:\\A Three-Parameter Refinement of the Grzegorczyk Hierarchy}
\author{Kirill Osipov\\{\normalsize Independent researcher}}
\date{August $2026$\\[0.5em]
}

\begin{document}
\maketitle
\begin{abstract}
We ask whether asymptotic recursion depth determines the expressive strength
of a bounded recursive algebra, and prove that it does not.  We replace ordinary predecessor
recursion by generalized-inverse descent along a fixed iterate $g_n^{[l]}$ and
obtain classes $\Hclass{m}{n}{l}$, where $m$ measures initial-function strength,
$n$ the growth row, and $l$ the traversal stride.

For all rows $n,n'\ge2$ we prove an exact inclusion criterion.  At a fixed row $n\ge2$ three
regimes occur: below the critical basis ($m<n$), equal-row inclusion is exactly
reverse divisibility $l'\mid l$; at $m=n$ every stride collapses to one class;
and from $m=n+1$ this class is ordinary bounded recursion $\Eclass{m}$.  Hence
pairwise $\Theta$-equivalent descent depths can induce infinite descending
chains, infinite antichains, and copies of every finite partial order.  The
separation is therefore controlled by traversal alignment rather than by
growth rate or recursion depth alone.

The proof combines exact-depth simulation, trace sparsity, and selected
dependency chains.  The exceptional doubling row has the same reverse-
divisibility order at basis zero, but all strides collapse from basis one
onward; from basis three it equals ordinary bounded recursion, while at basis
two $\Hclass{2}{1}{l}\subsetneq\mathsf{FP}$.  At basis zero the doubling-row
classes are proper subclasses of deterministic functional logspace, so the same
dual-divisibility order already occurs inside $\mathsf{FL}$.
\end{abstract}


\section{Introduction}

\subsection{The operation}

Ordinary bounded recursion constructs
\[
 f(\bar x,0)=g(\bar x),\qquad
 f(\bar x,y+1)=h(\bar x,y,f(\bar x,y)),
\]
subject to an earlier pointwise bound, so recursion on $y$ permits one update
at each predecessor stage.  Step recursion changes only this schedule.  For a
strictly increasing $\varphi:\Nat\to\Nat$ with $\varphi(x)\ge x+1$, let
\[
 \rho_\varphi(0)=0,\qquad
 \rho_\varphi(y)=\min\{z:\varphi(z)\ge y\}\quad(y>0).
\]
The recursion clause becomes
\[
 f(\bar x,0)=g(\bar x),\qquad
 f(\bar x,y)=h\bigl(\bar x,\rho_\varphi(y),f(\bar x,\rho_\varphi(y))\bigr),
\]
and follows the descent
$y,\rho_\varphi(y),\rho_\varphi^{[2]}(y),\ldots,0$.
For example, $\varphi(x)=x^2+2$ sends the recursion starting at $38$ through
$38,6,2,0$, so the relevant resource is the descent depth
\[
 D_{\rho_\varphi}(y)=\min\{t:\rho_\varphi^{[t]}(y)=0\}.
\]

We fix Grzegorczyk bases $B_m$ and
\[
 g_0(x)=x+1,\qquad g_1(x)=2x+1,\qquad g_n=e_n\quad(n\ge2).
\]
For $l\ge1$, $\Hclass{m}{n}{l}$ is the closure of $B_m$ under composition
and bounded step recursion with step $g_n^{[l]}$.  Thus $m$ measures the
initial-function strength, $n$ the growth scale, and $l$ the stride through
its canonical layers.  Below the horizontal collapse threshold, a six-layer
move can be simulated by two three-layer moves, while strides $2$ and $3$ are
incomparable: the order is divisibility rather than numerical size.

\subsection{Context and significance}

Classical subrecursive hierarchies separate functions by changing the initial
basis, the permitted recursion scheme, or a syntactic measure of recursion.
The Grzegorczyk hierarchy and its function-algebraic presentations retain
ordinary predecessor recursion while calibrating generator strength
\cite{Grzegorczyk1953,Rose1984,Clote1999}.  Recursion-number and ranking
approaches measure the structural complexity of primitive-recursive
definitions \cite{Parsons1968,Schwichtenberg1969,BellantoniNiggl1999}, while
loop hierarchies measure program or nesting structure
\cite{MeyerRitchie1967,Cleave1963,GoetzeNehrlich1978,GoetzeNehrlich1980}.
Axt's subrecursive hierarchy and later iteration results, together with the
classical elimination literature, vary the hierarchy construction or modify and
iterate recursion schemes themselves
\cite{Axt1959,Axt1965,Axt1966,Rodding1964,Gladstone1967,Gladstone1971,Georgieva1977}.
Other nearby lines characterize complexity by machine bounds, ordinal descent,
safe or ramified recursion, path orders, discrete equations, or cyclic proof
systems
\cite{Ritchie1963,Cobham1965,LobWainer1970,Wainer1972,BellantoniCook1992,
DalLagoMartiniZorzi2010,AvanziniDalLago2018,DasOitavem2018,
AvanziniMoser2013,BournezDurand2019,CurziDas2026,
TabatabaiGreatiRamanayake2025}.

The parameter isolated here is different.  We keep both the bounded-recursion
format and the growth row fixed and vary only the generalized-inverse
\emph{traversal} through that row.  In particular, varying $l$ changes neither
recursion rank nor recursion-nesting depth.  Thus $l$ is neither a stronger
generator nor a change of recursion scheme: it specifies which canonical
layers are visited by the same bounded recursive mechanism.
The main theorem shows that asymptotic descent depth is not a complete
invariant of the resulting algebra.  Below the critical basis, all fixed
strides have pairwise $\Theta$-equivalent depths, yet their inclusion order is
exactly dual divisibility.  Consequently a single fixed Grzegorczyk sector
contains infinite descending chains, infinite antichains, and copies of every
finite partial order.  Rich partial orders occur in broader subrecursive
settings \cite{Machtey1971}; here the order is forced canonically by one
arithmetic traversal parameter.

Equally important is the disappearance of this information.  At the critical
basis the whole stride order collapses, and one basis level later the common
class is ordinary bounded recursion.  The exceptional doubling row shows that
the phenomenon is not tied to the rapidly growing rows: reverse divisibility
reappears at basis zero on dyadic layers and vanishes as soon as the initial
basis is strengthened.  The classification therefore locates both a
traversal-sensitive region and the exact threshold at which initial-function
strength erases traversal information.

\subsection{Results}

The principal theorem gives a complete cross-row classification for all
$n,n'\ge2$ and, in particular, shows that $\Theta$-equivalence of recursion
depth does not determine class equality; row zero is integrated separately.  In the exceptional doubling row $n=1$, the only
stride-sensitive basis is $m=0$, where equal-row inclusion is again exactly
reverse divisibility; all strides collapse for $m\ge1$.  The main phenomenon
is already visible at one fixed row $n\ge2$:
\[
 \boxed{m<n:\ \text{reverse divisibility}}
 \quad\longrightarrow\quad
 \boxed{m=n:\ \text{stride collapse}}
 \quad\longrightarrow\quad
 \boxed{m\ge n+1:\ \Hclass{m}{n}{l}=\Eclass{m}}.
\]
Thus traversal remains visible below basis level $n$, disappears exactly at
$m=n$, and one level later step recursion recovers ordinary bounded recursion.
The complete cross-row criterion is
\[
\Hclass{a}{n}{l}\subseteq\Hclass{b}{n'}{l'}
\quad\Longleftrightarrow\quad
 a\le b\ \text{ and }\
\begin{cases}
\text{no further condition}, & n>n',\\
 b\ge n\text{ or }l'\mid l, & n=n',\\
 b\ge n'+1, & n<n'.
\end{cases}
\]
The doubling row has its own sharp transition.  At basis zero,
\[
 \Hclass{0}{1}{p}\subseteq\Hclass{0}{1}{q}
 \quad\Longleftrightarrow\quad q\mid p,
\]
and every class in this basis-zero stride sector is a proper subclass of
$\mathsf{FL}$.  Thus the dual-divisibility order already occurs inside
functional logspace.  All strides collapse from basis one onward.  Moreover,
$\Hclass{m}{1}{l}=\Eclass{m}$ for $m\ge3$; at basis $2$ we construct the
base-two vertical bridge and prove $\Hclass{2}{1}{l}\subsetneq\mathsf{FP}$,
while equality with $\Eclass{2}$ would imply $\mathsf P=\mathsf{NP}$.

The proof architecture is modular; the three main mechanisms can be read as follows.
\begin{center}
\small
\renewcommand{\arraystretch}{1.08}
\begin{tabular}{@{}>{\raggedright\arraybackslash}p{0.37\linewidth}c>{\raggedright\arraybackslash}p{0.47\linewidth}@{}}
\toprule
\textbf{Mechanism} & & \textbf{Main role} \\
\midrule
Exact-depth transfer & $\longrightarrow$ & positive inclusions \\
Trace sparsity & $\longrightarrow$ & saturation and vertical separations \\
Canonical zones and selected dependency chains & $\longrightarrow$ & equal-row divisibility \\
\bottomrule
\end{tabular}
\end{center}
Vertical depth domination handles the remaining changes of row.  Formal
definitions begin in Section~\ref{sec:framework}.

\section{Framework}
\label{sec:framework}

Throughout the paper,
\[
 \Nat=\{0,1,2,\ldots\},\qquad
 \Nat_{>0}:=\Nat\setminus\{0\}.
\]
For $x,y\in\Nat$, write
\[
 x\dotminus y:=\max\{x-y,0\}
\]
for truncated subtraction.  All functions are total maps from a positive finite Cartesian power of
$\Nat$ to $\Nat$.  All indices that specify a generating function or an
iteration count are fixed parameters, not numerical inputs.  For a unary function $u$, write
\[
 u^{[0]}(x)=x,\qquad u^{[t+1]}(x)=u(u^{[t]}(x)).
\]
Thus $u^{[t]}$ is the $t$-fold iterate of $u$.  For functions of the same
arity, $f\le g$ means $f(\bar x)\le g(\bar x)$ for every input.  A function
$g$ satisfying $f\le g$ is a \emph{majorant} of $f$.  For a statement $\mathsf A$,
$[\mathsf A]$ is $1$ if $\mathsf A$ is true and $0$ otherwise; a \emph{predicate} is a
$\{0,1\}$-valued function.  A function is \emph{internal to} a class if it
belongs to that class.

When $G$ is unary and $c$ is fixed, saying that a quantity depending on $N$
is bounded by $G^{[c]}(N+c)$ means that it is at most
$G^{[c]}(N+c)$ for every $N$.  The dependence of $c$ is stated where it is used; $c$ never depends on
$N$ or on the input values.

The initial functions are the unary zero function
\[
 Z(x)=0,
\]
the unary successor function
\[
 S(x)=x+1,
\]
and, for every $k\ge1$ and $1\le i\le k$, the projection
\[
 P_i^k(x_1,\ldots,x_k)=x_i.
\]
Zero functions of higher arity are obtained as $Z\circ P_1^k$; projections
also provide dummy variables and permutations of arguments.

We use Rose's level numbering for the Grzegorczyk hierarchy
\cite[Chapter~2, especially pp.~31--36]{Rose1984}, but index its
generating functions from $1$: for $i\ge1$, our $e_i$ is Rose's
$e_{i-1}$.  Thus
\[
 e_1(x,y)=x+y,\qquad e_2(x)=x^2+2,
\]
and, for $r\in\Nat$,
\[
 e_{r+3}(0)=2,\qquad
 e_{r+3}(x+1)=e_{r+2}(e_{r+3}(x)).
\]
Hence $e_1$ is binary and $e_k$ is unary for $k\ge2$.  Equivalently,
$e_{r+3}(x)=e_{r+2}^{[x]}(2)$.  For $m\in\Nat$, let
\[
 B_m=\{Z,S\}\cup\{P_i^k:k\ge1,\ 1\le i\le k\}
      \cup\{e_k:1\le k\le m\}.
\]
This is Rose's level-$m$ initial basis; only the symbols for its generators
have been shifted by one.

If $f$ is $r$-ary and $h_1,\ldots,h_r$ are $k$-ary, their composition is
\[
 f(h_1,\ldots,h_r)(\bar x)
   =f(h_1(\bar x),\ldots,h_r(\bar x)).
\]

\begin{definition}[finite-stage closure]
\label{def:finite-stage-closure}
Let $C$ be a family of functions and let $\mathcal O$ be a collection of
operations on functions.  Put $C^{(0)}=C$, and let $C^{(s+1)}$ contain
$C^{(s)}$ together with every function obtained by one application of an
operation in $\mathcal O$ to functions in $C^{(s)}$.  We write
\[
 \operatorname{Cl}_{\mathcal O}(C)=\bigcup_{s\in\Nat}C^{(s)}.
\]
Thus $\operatorname{Cl}_{\circ}(C)$ denotes closure under composition.
A function used at stage $s+1$ is called
\emph{earlier} if it belongs to $C^{(s)}$.
\end{definition}

\begin{definition}[bounded recursion]
\label{def:bounded-recursion}
Fix $k\ge1$.  Let $g:\Nat^k\to\Nat$,
$h:\Nat^{k+2}\to\Nat$, and $b:\Nat^{k+1}\to\Nat$ be earlier functions.
Bounded recursion forms the function $f:\Nat^{k+1}\to\Nat$ defined by
\[
 f(\bar x,0)=g(\bar x),\qquad
 f(\bar x,y+1)=h(\bar x,y,f(\bar x,y)),
\]
provided
\[
 f(\bar x,y)\le b(\bar x,y)
\]
for all arguments.  We denote this operation by $\mathrm{BR}$, write
$\operatorname{Cl}_{\circ,\mathrm{BR}}(C)$ for closure under composition
and bounded recursion, and define
\[
 \Eclass{m}=\operatorname{Cl}_{\circ,\mathrm{BR}}(B_m).
\]
These are the bounded-recursion classes used in the article.
\end{definition}

\begin{definition}[bounded step recursion]
\label{def:descent-closure}
A \emph{descent function} is a nondecreasing unary function $\rho$ satisfying
\[
 \rho(0)=0,\qquad \rho(y)<y\quad(y>0).
\]
Its \emph{descent depth} is the total function
\[
 D_\rho(y):=\min\{t\in\Nat:\rho^{[t]}(y)=0\}.
\]
The minimum exists because every positive iterate strictly decreases until it
reaches $0$.  Moreover, $D_\rho$ is nondecreasing.  Indeed, if $x\le y$, then
monotonicity of $\rho$ gives $\rho^{[t]}(x)\le\rho^{[t]}(y)$ for every $t$;
therefore $\rho^{[t]}(y)=0$ implies $\rho^{[t]}(x)=0$, and hence
$D_\rho(x)\le D_\rho(y)$.  Fix $k\ge1$.  Let $g:\Nat^k\to\Nat$,
$h:\Nat^{k+2}\to\Nat$, and $b:\Nat^{k+1}\to\Nat$ be earlier functions.
Bounded step recursion along $\rho$ forms the function
$f:\Nat^{k+1}\to\Nat$ defined by
\[
 f(\bar x,0)=g(\bar x),\qquad
 f(\bar x,y)=h\bigl(\bar x,\rho(y),f(\bar x,\rho(y))\bigr)
 \quad(y>0),
\]
provided $f(\bar x,y)\le b(\bar x,y)$ for all arguments.  We denote by
$\mathrm{SR}_\rho$ the operation that maps the admissible data $(g,h,b)$ to
$f$, and write
\[
 A_\rho(C)=\operatorname{Cl}_{\circ,\mathrm{SR}_\rho}(C).
\]

If $\varphi:\Nat\to\Nat$ is strictly increasing and
$\varphi(x)\ge x+1$, define its generalized inverse by
\[
 \rho_\varphi(0)=0,\qquad
 \rho_\varphi(y)=\min\{z\in\Nat:\varphi(z)\ge y\}\quad(y>0).
\]
This is a descent function.  We call $\mathrm{SR}_{\rho_\varphi}$
\emph{step recursion with step $\varphi$}.
\end{definition}

Define a sequence of increasing functions by
\[
 g_0(x)=x+1,\qquad
 g_1(x)=S(e_1(x,x))=2x+1,\qquad
 g_n=e_n\quad(n\ge2).
\]
For $n\in\Nat$ and $l\ge1$, set
\[
 \varphi_{n,l}=g_n^{[l]},
\]
so that
\[
 \varphi_{0,l}(x)=x+l,\qquad
 \varphi_{1,l}(x)=2^l(x+1)-1,\qquad
 \varphi_{n,l}=e_n^{[l]}\quad(n\ge2).
\]
Let $\rho_{n,l}=\rho_{\varphi_{n,l}}$ and put
\[
 D_{n,l}:=D_{\rho_{n,l}}.
\]
Finally, put
\[
 \Hclass{m}{n}{l}=A_{\rho_{n,l}}(B_m).
\]
The family
\[
 \bigl\{\Hclass{m}{n}{l}:m,n\in\Nat,\ l\ge1\bigr\}
\]
is the three-parameter \emph{step-recursion hierarchy} introduced in this
paper.  It refines the Grzegorczyk hierarchy by separating the strength
of the initial basis from the growth scale and stride of the recursion step.
For $n\ge2$, row $n$ uses the Grzegorczyk generator $e_n$ directly.
Row $1$ instead uses the strict unary diagonal majorant
$g_1(x)=S(e_1(x,x))=2x+1$.  The parameter $m$ determines the initial functions,
$n$ selects $g_n$, and $l$ is the number of iterations used to form the step.  The family with fixed
$n$ is the $n$-th row.  Since $B_a\subseteq B_b$ for $a\le b$,
\begin{equation}
\label{eq:basis-monotonicity}
 \Eclass{a}\subseteq\Eclass{b},\qquad
 \Hclass{a}{n}{l}\subseteq\Hclass{b}{n}{l}.
\end{equation}

\begin{remark}[index convention]
\label{rem:convention-note}
The lowest choices are $\varphi_{0,1}(x)=x+1$ and
$\varphi_{1,1}(x)=2x+1$.  For $n\ge2$ the step is a fixed iterate of the
Grzegorczyk generating function $e_n$.  The classes with $n=1$ are not
completely classified at the low initial bases $m\le2$; Section~\ref{sec:doubling-row}
proves the base-two bridge and saturation from initial basis $3$ onward.
\end{remark}

For example, step recursion with step $\varphi_{0,l}(x)=x+l$ uses
\[
 \rho_{0,l}(y)=y\dotminus l,
 \qquad
 D_{0,l}(y)=\left\lceil\frac yl\right\rceil.
\]
For example, when $l=2$, the successive recursion arguments starting from
$7$ are
\[
 7,\ 5,\ 3,\ 1,\ 0,
\]
so the descent depth is $4=\lceil 7/2\rceil$.
For $\varphi_{1,l}(x)=2^l(x+1)-1$,
\[
 \rho_{1,l}(y)=\left\lfloor\frac y{2^l}\right\rfloor,
 \qquad
 D_{1,l}(y)=\left\lceil\frac{\log_2(y+1)}l\right\rceil.
\]
For $\varphi_{2,1}(x)=x^2+2$, repeated applications of the generalized
inverse reduce the argument approximately by successive square roots, and
$D_{2,1}(y)$ has order $\log\log(y+4)$.

Bounded recursion is precisely step recursion with step $x+1$.
Indeed, for $\varphi_{0,1}(y)=y+1$,
\[
 \rho_{0,1}(0)=0,\qquad \rho_{0,1}(y+1)=y,
\]
so the defining equations and admissible bounds of the two operations are
identical.  Hence
\begin{equation}
\label{eq:ordinary-boundary}
 \boxed{\Eclass{m}=\Hclass{m}{0}{1}.}
\end{equation}
The value $l=0$ is excluded because it makes $\varphi_{n,0}$ the identity,
whose generalized inverse is not a strict descent.

\begin{lemma}[generator identities]
\label{lem:generator-identities}
For every $r\ge2$,
\[
 e_{r+1}(x)=e_r^{[x]}(2)
 \qquad\text{and}\qquad
 e_r(x)\le e_{r+1}(x).
\]
In particular, $e_r(x)\ge x+2$ and every unary generating function $e_r$,
$r\ge2$, is strictly increasing.
\end{lemma}

\begin{proof}
The iteration identity is induction on $x$ from the defining equations.
The final two assertions hold for $e_2(x)=x^2+2$.  For the first index
comparison, induction on $x$ gives
\[
 e_3(x+1)=e_2(e_3(x))\ge e_2(e_2(x))\ge e_2(x+1).
\]
If $e_r\ge e_{r-1}$ for some $r\ge3$, the same induction gives
\[
 e_{r+1}(x+1)=e_r(e_{r+1}(x))
 \ge e_r(e_r(x))\ge e_{r-1}(e_r(x))=e_r(x+1).
\]
All unary generators take the value $2$ at $0$.  Finally, if
$e_r(x)\ge x+2$, then
$e_{r+1}(x+1)=e_r(e_{r+1}(x))\ge e_{r+1}(x)+2$; induction on $x$ gives
$e_{r+1}(x)\ge2x+2\ge x+2$ and strict increase.
\end{proof}

\begin{lemma}[stage-preserving composition majorants]
\label{lem:majorants}
Let $\rho$ be any descent function.  Consider either finite-stage construction
\[
 \operatorname{Cl}_{\circ,\mathrm{BR}}(B_m)
 \qquad\text{or}\qquad
 \operatorname{Cl}_{\circ,\mathrm{SR}_{\rho}}(B_m),
\]
and write $C^{(s)}$ for its stages from
Definition~\ref{def:finite-stage-closure}.  For every $s$ and every
$k$-ary $f\in C^{(s)}$, there is a coordinatewise nondecreasing
\[
 T_f\in \operatorname{Cl}_\circ(B_m)\cap C^{(s)}
\]
such that $f(\bar x)\le T_f(\bar x)$ for all $\bar x$.

If $m=0$, $k$-ary composition terms are constants or $x_i+c$; if $m=1$,
they are affine functions with nonnegative integer coefficients.  If
$m\ge2$ and $F=e_m$, every $k$-ary composition term $T$ satisfies
\[
 T(\bar x)\le F^{[c]}(M+c),\qquad M=\max_i x_i,
\]
for some constant $c$ depending on $T$.
\end{lemma}

\begin{proof}
For the stage-preserving assertion, induct on $s$.  Every member of
$C^{(0)}=B_m$ is coordinatewise nondecreasing and majorizes itself.  Suppose
the assertion holds at stage $s$.  A function carried from $C^{(s)}$ keeps
its previous majorant.  If
\[
 f=F(G_1,\ldots,G_r)
\]
is added by composition from functions in $C^{(s)}$, compose the induction
majorants of $F,G_1,\ldots,G_r$.  Monotonicity of the majorant of $F$ shows
that the resulting composition majorizes $f$, and one application of
composition places it in $C^{(s+1)}$.  If $f$ is added by bounded recursion
or bounded step recursion with declared bound $b\in C^{(s)}$, the induction
majorant of $b$ already lies in $C^{(s)}$ and majorizes $f$.  Thus the
majorant never occurs later than the function it bounds.

At $m=0$, substitution among zero, successor, and projections produces
exactly constants and shifted projections.  Adding $e_1(x,y)=x+y$ at $m=1$
produces exactly nonnegative affine forms, by structural induction.

For $m\ge2$, Lemma~\ref{lem:generator-identities} shows that every unary
generating function in $B_m$ is bounded by $F=e_m$ and that
\[
 F(u)\ge e_2(u)=u^2+2\ge2u.
\]
Induct on a composition term.  Zero, projections, and successor satisfy the
claim after increasing $c$.  A unary generator costs at most one additional
outer application of $F$.  If two arguments of addition are at most $A$,
then $a+b\le2A\le F(A)$.  Taking the maximum of the finitely many constants
at a composition step completes the induction.
\end{proof}

\begin{corollary}[arbitrary and monotone bound presentations]
\label{cor:bound-presentation-equivalence}
The classes $\Eclass{m}$ and $\Hclass{m}{n}{l}$ are unchanged if their
recursion rules are required to use a coordinatewise nondecreasing declared
bound.  Equivalently, an arbitrary earlier pointwise bound may always be
replaced by an earlier monotone composition majorant.
\end{corollary}

\begin{proof}
A monotone bound is a special case of an arbitrary pointwise bound.  In the
other direction, let $C^{(s)}$ be the stage sequence of the relevant closure
and suppose a recursion is introduced at stage $s+1$ with an arbitrary
declared bound $b\in C^{(s)}$.  Lemma~\ref{lem:majorants} supplies
a coordinatewise nondecreasing $B\in C^{(s)}$ with $b\le B$.  The same base
and transition equations define the same function and satisfy
$f\le b\le B$, so the recursion is admissible with the earlier monotone
bound $B$.
\end{proof}

Corollary~\ref{cor:bound-presentation-equivalence} shows that allowing an
arbitrary earlier pointwise bound gives the same classes as Rose's
monotone-bound presentation.

\begin{lemma}[generalized-inverse and descent-depth identities]
\label{lem:step-function-arithmetic}
Let $\varphi:\Nat\to\Nat$ be strictly increasing, unbounded, and satisfy
$\varphi(x)\ge x+1$.  Let $\rho$ be its generalized inverse.  Then the
following \emph{threshold law} holds:
\begin{equation}
\label{eq:threshold-law}
 D_\rho(y)\le d\quad\Longleftrightarrow\quad y\le\varphi^{[d]}(0).
\end{equation}
Thus $D_\rho(y)$ is the least $d$ such that
$y\le\varphi^{[d]}(0)$.  Moreover, for every $k\ge1$,
the generalized inverse of $\varphi^{[k]}$ is $\rho^{[k]}$, and its
descent depth is
\[
 D_{\rho^{[k]}}(y)=
 \left\lceil\frac{D_\rho(y)}k\right\rceil.
\]
In particular,
\[
 \rho_{n,kl}=\rho_{n,l}^{[k]},\qquad
 D_{n,kl}(y)=\left\lceil\frac{D_{n,l}(y)}{k}\right\rceil.
\]
\end{lemma}

\begin{proof}
The first assertion is induction on $d$.  At $d=0$ both sides say $y=0$.
For the induction step, the defining generalized-inverse relation gives
\[
 \rho(y)\le\varphi^{[d]}(0)
 \quad\Longleftrightarrow\quad
 y\le\varphi^{[d+1]}(0).
\]
Strict increase shows directly that the generalized inverse of
$\varphi^{[k]}$ is the $k$-fold inverse $\rho^{[k]}$.  It reaches $0$ after
the least $d$ with $kd\ge D_\rho(y)$, which is the displayed ceiling.
Taking $\varphi=\varphi_{n,l}$ gives the final formulas.
\end{proof}

\begin{remark}[depth versus iterated inverse]
The symbol $D_{\rho^{[k]}}$ denotes the descent depth of the $k$-fold inverse
$\rho^{[k]}$.  It is not an iterate of the numerical function $D_\rho$.
Lemma~\ref{lem:step-function-arithmetic} records the relation
$D_{\rho^{[k]}}(y)=\lceil D_\rho(y)/k\rceil$.
\end{remark}

\section{Step functions, generalized inverses, and ordinary calibration}

Whenever a displayed bounded recursion or bounded step recursion below has no
genuine parameter, we introduce one dummy parameter and remove it afterward
by composition with a projection.  Thus every such construction is a literal
instance of Definition~\ref{def:bounded-recursion} or
Definition~\ref{def:descent-closure}, respectively.  Every fixed numeral used below is
obtained by a fixed finite iterate of successor applied to zero, with any
dummy variables supplied by projections.

\begin{lemma}[internal zero-or-value selector]
\label{lem:internal-zero-value-selector}
Let $\rho$ be any descent function, let $C$ contain zero and all projections,
and put $\mathcal A=A_\rho(C)$.  Then the function
\[
 \operatorname{ZV}(u,0)=0,\qquad
 \operatorname{ZV}(u,c)=u\quad(c>0)
\]
belongs to $\mathcal A$.  In particular,
$\operatorname{ZV}\in\Hclass{m}{n}{l}$ for every $m,n\in\Nat$ and $l\ge1$.
\end{lemma}

\begin{proof}
Use one bounded step recursion along $\rho$, with $u$ as a parameter and $c$
as the recursion argument.  Its base is the zero function, its transition is
the projection to $u$ and ignores both the lower address and the previous
value, and its declared bound is the projection $u$.  Thus the recursion
returns $0$ at $c=0$ and $u$ at every positive $c$.
\end{proof}

\begin{lemma}[bounded controls and numerical selection]
\label{lem:ezero-selector}
The functions predecessor, truncated subtraction, minimum,
\[
 \NZ(0)=0,\quad \NZ(t+1)=1,\qquad
 \ISZ(0)=1,\quad \ISZ(t+1)=0,
\]
and the capped selector
\[
 \operatorname{Sel}(c,u,v,0)=\min(u,c),\qquad
 \operatorname{Sel}(c,u,v,t+1)=\min(v,c)
\]
belong to $\Eclass{0}$.  If $u,v\le c$, then
\[
 \operatorname{Sel}(c,u,v,t)=
 \begin{cases}
  u,&t=0,\\
  v,&t>0.
 \end{cases}
\]
Equality of two numbers is also decidable in $\Eclass{0}$.
\end{lemma}

\begin{proof}
Use the bounded recursions
\[
 \operatorname{pred}(0)=0,\qquad \operatorname{pred}(t+1)=t,
\]
\[
 x\dotminus0=x,\qquad
 x\dotminus(t+1)=\operatorname{pred}(x\dotminus t).
\]
They have projection bounds, and
$\min(x,c)=x\dotminus(x\dotminus c)$.  The displayed definitions of
$\NZ$, $\ISZ$, and $\operatorname{Sel}$ are bounded recursions with bounds
$1$, $1$, and $c$, respectively.  Finally, put
\[
 \operatorname{Eq}(x,y)
 :=\operatorname{Sel}
 \bigl(1,0,\ISZ(y\dotminus x),\ISZ(x\dotminus y)\bigr).
\]
If $x>y$, the last selector returns $0$; if $x\le y$, it returns $1$ exactly
when $y\le x$.  Hence $\operatorname{Eq}(x,y)=1$ exactly when $x=y$.
\end{proof}

\begin{lemma}[fixed step functions in the ordinary hierarchy]
\label{lem:step-function-ezero}
For every fixed $n\in\Nat$ and $l\ge1$, the functions
\[
 \rho_{n,l}(y),\qquad
 R_{n,l}(y,t)=\rho_{n,l}^{[t]}(y),\qquad
 D_{n,l}(y)
\]
belong to $\Eclass{0}$.
\end{lemma}

\begin{proof}
By Lemma~\ref{lem:ezero-selector}, the elementary bounded controls used below
belong to $\Eclass{0}$.  First evaluate the fixed generator under a cap.
Capped addition and multiplication are bounded recursions computing
$\min(x+t,c)$ and $\min(xt,c)$.  Hence
$\widehat g_i(x,c)=\min(g_i(x),c)$ is available for $i=0,1,2$.  If
$\widehat e_r(x,c)=\min(e_r(x),c)$ is available for fixed $r\ge2$, the bounded
recursion
\[
 T_r(c,0)=\min(2,c),\qquad
 T_r(c,x+1)=\widehat e_r(T_r(c,x),c)
\]
computes $\min(e_{r+1}(x),c)$: the computation is exact before reaching the
cap and then remains there.  External induction on the fixed index, followed
by $l$ fixed compositions, gives
\[
 \widehat\varphi_{n,l}(z,c)=\min(\varphi_{n,l}(z),c)\in\Eclass{0}.
\]

We use one bounded least-witness construction twice.  For any earlier
predicate $P(\bar a,z)$ put
\[
 Q_P(\bar a,z)=\operatorname{Sel}(z+1,0,z+1,P(\bar a,z)),
\]
\[
 L_P(\bar a,0)=0,\qquad
 L_P(\bar a,z+1)=\operatorname{Sel}
 \bigl(z+1,Q_P(\bar a,z),L_P(\bar a,z),\NZ(L_P(\bar a,z))\bigr).
\]
At stage $z+1$ the stored positive code is at most $z+1$, so this is bounded
recursion with the earlier bound $z+1$.  If a witness occurs below $U$, then
$\operatorname{pred}(L_P(\bar a,U))$ is its least index.

Apply this to
\[
 P(y,z)=\operatorname{Eq}(\widehat\varphi_{n,l}(z,y),y).
\]
It holds exactly when $\varphi_{n,l}(z)\ge y$, and $z=y-1$ is a witness for
$y>0$; hence $\rho_{n,l}(y)=\operatorname{pred}(L_P(y,y))\in\Eclass{0}$.
Next
\[
 R_{n,l}(y,0)=y,\qquad
 R_{n,l}(y,t+1)=\rho_{n,l}(R_{n,l}(y,t))
\]
is bounded by $y$.  Strict descent reaches $0$ by time $y$, so applying the
same construction to $\ISZ(R_{n,l}(y,t))$ with search length $y+1$ gives
$D_{n,l}(y)$.  Thus all three fixed controls belong to $\Eclass{0}$.
\end{proof}

\begin{lemma}[ordinary upper calibration]
\label{lem:calibration}
For all $m,n\in\Nat$ and $l\ge1$,
\[
 \Hclass{m}{n}{l}=A_{\rho_{n,l}}(B_m)\subseteq\Eclass{m}.
\]
\end{lemma}

\begin{proof}
Induct on the finite-stage construction of $A_{\rho_{n,l}}(B_m)$.  Only the
step-recursion case is nontrivial.  Suppose
\[
 f(\bar x,0)=g(\bar x),\qquad
 f(\bar x,y)=h(\bar x,\rho(y),f(\bar x,\rho(y)))\quad(y>0),
\]
where $\rho=\rho_{n,l}$ and $b$ is an earlier declared bound.  By the induction
hypothesis $g,h\in\Eclass{m}$, and Lemma~\ref{lem:majorants} supplies an
earlier coordinatewise nondecreasing composition majorant $B\ge b$.
Lemma~\ref{lem:step-function-ezero} puts $\rho$, its iterator $R$, and its
depth $D$ in $\Eclass{0}\subseteq\Eclass{m}$.

Write $d=D(y)$ and $y_i=\rho^{[i]}(y)$, so $y_d=0$.  Ordinary bounded
recursion, with bound $B(\bar x,y)$, folds forward along this chain:
\[
 V(\bar x,y,0)=g(\bar x),
\]
\[
 V(\bar x,y,t+1)=\operatorname{Sel}\!\left(
 B(\bar x,y),V(\bar x,y,t),
 h\bigl(\bar x,R(y,D(y)\dotminus t),V(\bar x,y,t)\bigr),
 \NZ(D(y)\dotminus t)\right).
\]
For $t\le d$, induction gives
$V(\bar x,y,t)=f(\bar x,y_{d-t})$.  Every transition value is at most
$b(\bar x,y_{d-t-1})\le B(\bar x,y)$ by monotonicity of $B$; after $t=d$
the selector holds the value fixed.  Hence
$V(\bar x,y,D(y))=f(\bar x,y)$, proving that the step-recursion clause belongs
to $\Eclass{m}$ and completing the construction induction.
\end{proof}

\begin{lemma}[internal controls for an arbitrary descent]
\label{lem:internal-controls}
Let $\rho$ be any descent in the sense of
Definition~\ref{def:descent-closure}, let $C$ contain zero, successor, and
projections, and put $\mathcal A=A_\rho(C)$.  Then $\mathcal A$ contains
the functions
\[
 \rho,\qquad D_\rho(y)=\min\{t:\rho^{[t]}(y)=0\},
 \qquad\NZ,\qquad\ISZ,
\]
and the depth bit $D_\rho(y)\bmod2$.  In particular, every
$\Hclass{m}{n}{l}$ contains $\rho_{n,l}$, $D_{n,l}$, these predicates, and
\[
 \operatorname{bit}_{n,l}(y)=D_{n,l}(y)\bmod2.
\]
\end{lemma}

\begin{proof}
With a dummy parameter $x$, define
\[
 F(x,0)=0,\qquad F(x,y)=\rho(y)\quad(y>0).
\]
This is a step recursion whose transition returns $\rho(y)$ and whose
bound is $y$.  Diagonal composition gives $\rho$.  Similarly,
\[
 G(x,0)=0,\qquad
 G(x,y)=G(x,\rho(y))+1\quad(y>0)
\]
has bound $y$ and computes $D_\rho(y)$, since strict descent gives
$D_\rho(y)\le y$.

The recursions with base $0$ and constant transition $1$, and with base $1$
and constant transition $0$, have bound $1$ and give $\NZ$ and
$\ISZ$.  Finally,
\[
 B(x,0)=0,\qquad
 B(x,y)=\ISZ(B(x,\rho(y)))\quad(y>0)
\]
is bounded by $1$ and toggles once at every recursion step.  Induction on depth
gives $B(x,y)=D_\rho(y)\bmod2$.  Substitution of
$\rho=\rho_{n,l}$ gives the final assertion.
\end{proof}

\section{Direct trace bounds}
\label{sec:direct-traces}

\begin{definition}[trace band]
\label{def:trace-band}
A pair of coordinates $(m,n)$ lies in the \emph{trace band} when
\[
 n\ge2\qquad\text{and}\qquad m\le n.
\]
\end{definition}

\begin{lemma}[trace-band depth overhead]
\label{lem:depth-overhead}
Suppose $(m,n)$ lies in the trace band and
$f\in\Hclass{m}{n}{l}$ is $k$-ary.  A constant $c\ge1$ exists such that,
with $M=\max_i x_i$,
\[
 D_{n,l}(f(\bar x))
 \le c\bigl(1+D_{n,l}(M+c)\bigr).
\]
Moreover, for the fixed parameters $n$ and $l$ under consideration,
\[
 D_{n,l}(z)=O_{n,l}(\log\log(z+4)).
\]
\end{lemma}

\begin{proof}
By Lemma~\ref{lem:majorants}, $f$ has a monotone composition majorant $T$.
Put $G=g_n=e_n$ and $\varphi=G^{[l]}$.  There are constants $a,c_0$
with
\[
 T(\bar x)\le G^{[a]}(M+c_0).
\]
For $m=0,1$, one iterate of $G\ge e_2$ after a fixed shift majorizes every
constant, shifted projection, or affine form.  For $m\ge2$, the last part of
Lemma~\ref{lem:majorants} applies and $e_m\le G$.  If
$d=D_{n,l}(M+c_0)$, the generalized-inverse identities give
\[
 M+c_0\le\varphi^{[d]}(0)=G^{[ld]}(0),
\]
whence
\[
 f(\bar x)\le G^{[a+ld]}(0)
 \le\varphi^{[d+\lceil a/l\rceil]}(0).
\]
The threshold law gives the claimed depth estimate.  Since $G\ge e_2$ and
$e_2^{[t]}(0)\ge2^{2^{t-1}}$ for $t\ge1$, the numerical bound is
$D_{n,l}(z)=O_{n,l}(\log\log(z+4))$.
\end{proof}

\begin{definition}[derivations]
\label{def:finite-derivations}
Fix $B_m$ and a descent function $\rho$.  A \emph{derivation} is a finite
tree.  Each node $\delta$ denotes a function, written $\sem{\delta}$.
The nodes are:
\begin{enumerate}
\item For $R\in B_m$, $\operatorname{Init}(R)$ is a leaf and
\[
 \sem{\operatorname{Init}(R)}=R.
\]
\item If $\eta$ denotes an $r$-ary function and
$\delta_1,\ldots,\delta_r$ denote functions of one common arity, then
$\operatorname{Comp}(\eta;\delta_1,\ldots,\delta_r)$ denotes their
composition:
\[
 \sem{\operatorname{Comp}(\eta;\delta_1,\ldots,\delta_r)}(\bar x)
 =\sem{\eta}\bigl(\sem{\delta_1}(\bar x),\ldots,
                 \sem{\delta_r}(\bar x)\bigr).
\]
\item Suppose $\gamma,\eta,\beta$ denote the base
$g:\Nat^k\to\Nat$, transition $h:\Nat^{k+2}\to\Nat$, and bound
$b:\Nat^{k+1}\to\Nat$.  If the step recursion from
Definition~\ref{def:descent-closure} satisfies its bound $f\le b$, then
$\operatorname{SR}_{\rho}(\gamma,\eta,\beta)$ denotes $f$.
\end{enumerate}
Thus $\sem{\delta}(\bar x)$ is the value denoted by $\delta$ at
$\bar x$.  The \emph{construction rank} is
\[
\begin{aligned}
 \operatorname{rk}(\operatorname{Init}(R))&=0,\\
 \operatorname{rk}(\operatorname{Comp}(\eta;\bar\delta))
 &=1+\max\bigl(\operatorname{rk}(\eta),\max_i\operatorname{rk}(\delta_i)\bigr),\\
 \operatorname{rk}(\operatorname{SR}_{\rho}(\gamma,\eta,\beta))
 &=1+\max\bigl(\operatorname{rk}(\gamma),
               \operatorname{rk}(\eta),
               \operatorname{rk}(\beta)\bigr).
\end{aligned}
\]
\end{definition}

\begin{definition}[evaluation values and addresses]
\label{def:derivation-evaluation-values}
When a derivation node of the form
$\operatorname{SR}_{\rho}(\gamma,\eta,\beta)$ is invoked at
$(\bar p,z)$, its \emph{address} is the recursion argument $z$.  For a
derivation $\delta$ and input $\bar x$, let
$\operatorname{Addr}_{\delta}(\bar x)$ be the finite multiset of the
addresses of all such node invocations while computing
$\sem{\delta}(\bar x)$.  It includes invocations inside base and transition
evaluations.  Let $\operatorname{Val}_{\delta}(\bar x)$ be the finite set of
all numerical values created in the same evaluation.

For $\operatorname{Init}(R)$, the value set contains the inputs and
$R(\bar x)$, and the address multiset is empty.  For a composition, evaluate the argument derivations and then $\eta$ at
their outputs, and take the corresponding unions.

For
\[
 \delta=\operatorname{SR}_{\rho}(\gamma,\eta,\beta)
\]
evaluated at $(\bar p,z)$, put $d=D_\rho(z)$,
\[
 z_j=\rho^{[d-j]}(z)\quad(0\le j\le d).
\]
These $z_j$ are values, not additional addresses of the outer node.  Put
\[
 s_0=\sem{\gamma}(\bar p),\qquad
 s_{j+1}=\sem{\eta}(\bar p,z_j,s_j)\quad(0\le j<d).
\]
Then $z_0=0$, $z_d=z$, and
$s_j=\sem{\delta}(\bar p,z_j)$.  The value set contains $\bar p,z$, all
$z_j,s_j$, and the values created in the displayed base and transition
evaluations.  The address multiset contains the outer address $z$ and all
addresses from those evaluations.  Thus one invocation of the outer
step-recursion node contributes its invocation argument $z$ as one recorded
address; the points $z_0,\ldots,z_d$ form the internal schedule of that single
invocation.  Recording those points as addresses as well would not change any
later estimate, because $D_\rho(z_j)=j\le D_\rho(z)$.  The bound derivation
$\beta$ is a static admissibility certificate: it certifies
$s_j\le\sem{\beta}(\bar p,z_j)$, but its evaluation is not part of the
operational value or address sets.
\end{definition}

\begin{lemma}[derivation-tree representation]
\label{lem:derivation-tree-representation}
Every function in $A_\rho(B_m)$ has a derivation.  Replacing a shared
use of an earlier function by a separate copy preserves denotation.
\end{lemma}

\begin{proof}
Induct on the finite-stage construction of
Definition~\ref{def:finite-stage-closure}.  Initial functions give initial
derivations.  A composition or bounded step recursion applied at the next
stage becomes the corresponding tree node with copies of the finitely many
earlier derivations.  Copying preserves the denoted function and the bound condition.
\end{proof}

\begin{lemma}[extensional evaluation invariance]
\label{lem:extensional-evaluation}
Fix a derivation $\delta$.  Its denoted output, value set, and address multiset
at an input depend only on the numerical input tuple.  In particular, if
unary external substitutions $T_1,\ldots,T_k$ and
$U_1,\ldots,U_k$ satisfy
\[
 T_i(y)=U_i(y)\qquad(1\le i\le k),
\]
then the two evaluations of $\delta$ at those substituted tuples have the
same output, the same internal values, and the same recorded addresses at
that $y$.
\end{lemma}

\begin{proof}
Induct on the derivation tree.  The assertion is immediate for an initial
node.  At a composition node, equality of the numerical argument outputs
makes the head derivation receive the same numerical tuple, and the induction
hypothesis applies to every subevaluation.  At a step-recursion node, equality
of the input tuple gives the same outer address, the same descent schedule,
and, inductively, the same base value and the same transition value at every
schedule point.  The unions defining the value and address data are therefore
identical.
\end{proof}

\begin{lemma}[absorption of fixed growth envelopes]
\label{lem:envelope-absorption}
Let $G:\Nat\to\Nat$ be nondecreasing and satisfy $G(x)\ge x+2$.  Every
finite expression obtained from a variable $N$, fixed numerals, and maps of
the form
\[
 x\longmapsto G^{[a]}(x+b)
 \qquad(a,b\in\Nat\text{ fixed})
\]
is at most
\[
 G^{[c]}(N+c)
\]
for all $N$, for some $c\in\Nat$ independent of $N$.  A finite maximum of
such expressions has an upper bound of the same form.
\end{lemma}

\begin{proof}
For each fixed $b$, choose $s$ with $2s\ge b$.  Then
$G^{[s]}(x)\ge x+b$, and hence
\[
 G^{[a]}(x+b)\le G^{[a+s]}(x).
\]
Induction over the finite expression absorbs one fixed shift and one fixed
iterate count at each node.  Increasing the final constant bounds fixed
numerals and all members of a finite family simultaneously.
\end{proof}

\begin{lemma}[uniform internal-value envelope]
\label{lem:derivation-internal-envelope}
Suppose $(m,n)$ lies in the trace band, put $G=g_n=e_n$, and fix a finite
derivation $\delta$ over $B_m$ using $\rho_{n,l}$.  There is a constant
$c_\delta\ge1$ such that, whenever all input coordinates are at most $N$,
\[
 v\le G^{[c_\delta]}(N+c_\delta)
 \qquad
 \bigl(v\in\operatorname{Val}_{\delta}(\bar x)\bigr).
\]
Consequently every recursion address in the evaluation satisfies
\begin{equation}
\label{eq:derivation-depth-envelope}
 D_{n,l}(a)
 \le C_\delta\bigl(1+D_{n,l}(N+C_\delta)\bigr)
 \qquad
 \bigl(a\in\operatorname{Addr}_{\delta}(\bar x)\bigr)
\end{equation}
for a derivation-dependent constant $C_\delta$.  At the inputs
$Y_t=G^{[t]}(0)$ the right-hand side is $O_\delta(t+1)$.
\end{lemma}

\begin{proof}
We induct on construction rank and use
Lemma~\ref{lem:envelope-absorption} whenever finitely many fixed shifts,
iterates, or maxima have to be combined.

For an initial derivation the assertion follows from
Lemma~\ref{lem:majorants}, since every generator in $B_m$ is bounded, on
inputs at most $N$, by a fixed iterate of $G$ after a fixed shift.  The input
coordinates themselves satisfy the same estimate.

For a composition, the induction hypotheses uniformly bound every value in
the argument evaluations.  Their outputs are therefore bounded by one common
quantity $A=G^{[a]}(N+a)$.  Apply the induction hypothesis for the head
derivation with input cap $A$ and absorb the resulting fixed composition of
$G$-iterates into one envelope of the required form.

Consider a step-recursion derivation at input $(\bar p,z)$ with coordinates
at most $N$.  Every $z_j$ is at most $z\le N$.  By the
induction hypothesis for the bound derivation, its output on $(\bar p,z_j)$
is at most $B=G^{[b]}(N+b)$, uniformly in $j$.  Since
$s_j=\sem{\delta}(\bar p,z_j)$, admissibility gives
$s_j\le\sem{\beta}(\bar p,z_j)\le B$.  The evaluation of $\gamma$ has inputs at most $N$.  Every evaluation of
$\eta$ has all parameters and its address at
most $N$, and its previous-state input at most $B$.  Applying the induction
hypotheses to $\gamma$ and $\eta$ with input caps $N$ and $\max(N,B)$
therefore bounds not only their outputs but every internal value and every
recursion address created in every evaluation of $\gamma$ or $\eta$.
Lemma~\ref{lem:envelope-absorption} combines these finitely many envelope
forms into the required bound, independently of the outer recursion depth
$d$.

Every address belongs to the value set, so monotonicity of $D_{n,l}$ and the
threshold law give the first depth bound.  More explicitly, if
$D_{n,l}(N+c)=d$, then
$N+c\le G^{[ld]}(0)$, and a fixed additional number of $G$-iterations adds
only a fixed amount to the $\rho_{n,l}$-depth.  Enlarging constants yields
\eqref{eq:derivation-depth-envelope}.  Finally,
$D_{n,l}(Y_t)=\lceil t/l\rceil$, up to the harmless initial value $Y_0=0$,
so the canonical-input estimate follows.
\end{proof}

\begin{definition}[descent-composition functions]
\label{def:descent-composition}
Fix $m\in\Nat$ and a descent function $\rho$.  Put
\[
 \mathcal C_m^\rho
 :=\operatorname{Cl}_\circ\bigl(B_m\cup\{\rho\}\bigr)
\]
for the all-arity composition closure, and let $\mathcal M_m^\rho$ be its
unary part.  Every function in $\mathcal C_m^\rho$ is coordinatewise
nondecreasing, and every function in $\mathcal M_m^\rho$ is therefore
nondecreasing.
\end{definition}

\begin{theorem}[direct derivation trace theorem]
\label{thm:structural-trace}
Fix a $k$-ary derivation $\delta$ over $B_m$ using a descent $\rho$.
There is a nondecreasing function
$Q_\delta:\Nat\to\Nat_{>0}$, depending only on $\delta$, with the following
property.  Let $J\subseteq\Nat$ be a nonempty finite consecutive interval and
let $T_1,\ldots,T_k\in\mathcal M_m^\rho$.  Suppose that for every $y\in J$
and every address
\[
 a\in\operatorname{Addr}_{\delta}
       \bigl(T_1(y),\ldots,T_k(y)\bigr)
\]
one has $D_\rho(a)\le K$.  Then $J$ has a partition into at most
$Q_\delta(K)$ consecutive subintervals on each of which
\[
 y\longmapsto
 \sem{\delta}\bigl(T_1(y),\ldots,T_k(y)\bigr)
\]
agrees with a member of $\mathcal M_m^\rho$.

Constants $A_\delta,r_\delta\ge1$ exist such that
\begin{equation}
\label{eq:direct-trace-growth}
 Q_\delta(K)
 \le\exp\!\bigl(A_\delta(K+1)^{r_\delta}\bigr).
\end{equation}
The same bound applies after restricting $J$ to any nonempty consecutive
subinterval.
\end{theorem}

\begin{proof}
We induct on construction rank and define $Q_\delta$ recursively.

For $\delta=\operatorname{Init}(R)$, the composed unary function
$R(T_1,\ldots,T_k)$ belongs to $\mathcal M_m^\rho$, so put
$Q_\delta(K)=1$.

Suppose
\[
 \delta=\operatorname{Comp}(\eta;\delta_1,\ldots,\delta_r).
\]
Apply the induction hypothesis successively to the argument derivations and
take the common refinement.  It has at most
$\prod_iQ_{\delta_i}(K)$ pieces.  On each piece the argument outputs agree
with unary functions $U_i\in\mathcal M_m^\rho$.  On the current piece one has
$U_i(y)=\sem{\delta_i}(T_1(y),\ldots,T_k(y))$ for every $i$.  Hence the head
derivation $\eta$ receives exactly the same numerical input tuple as in the
original composition evaluation.  By
Lemma~\ref{lem:extensional-evaluation}, its internal value and address data
are exactly those of the corresponding original subevaluation.  The
evaluations used to compute the external substitution functions $U_i$ are
not part of $\operatorname{Addr}_\eta$, so every recorded address still
satisfies the same bound.  A further application of the induction hypothesis
for $\eta$ gives
\[
 Q_\delta(K)
 =Q_\eta(K)\prod_{i=1}^rQ_{\delta_i}(K).
\]

Now suppose
\[
 \delta=\operatorname{SR}_\rho(\gamma,\eta,\beta).
\]
After substituting the unary functions $T_1,\ldots,T_k$, write the resulting
parameter functions as $\bar P$ and the final recursion argument as $A$.  All belong to
$\mathcal M_m^\rho$.  The monotonicity of $D_\rho$ was proved after
Definition~\ref{def:descent-closure}; hence $A$ and $D_\rho\circ A$ are
nondecreasing.  Its values lie in
$\{0,\ldots,K\}$, so its nonempty level sets partition $J$ into at most
$K+1$ consecutive intervals.

Fix one such interval and let the common depth be $d$.  Put
\[
 A_j=\rho^{[d-j]}\circ A\qquad(0\le j\le d).
\]
On the current interval,
$A_0(y)=0$, $A_d(y)=A(y)$, and
$\rho(A_{j+1}(y))=A_j(y)$ for every $y$; moreover every $A_j$ belongs to
$\mathcal M_m^\rho$.  Define unary functions
\[
 W_0(y)=\sem{\gamma}(\bar P(y)),\qquad
 W_{j+1}(y)=\sem{\eta}(\bar P(y),A_j(y),W_j(y))
 \quad(0\le j<d).
\]
Then $W_d(y)=\sem{\delta}(\bar P(y),A(y))$.

The induction hypothesis for $\gamma$ partitions the interval into at most
$Q_\gamma(K)$ pieces on which $W_0$ agrees with a member of
$\mathcal M_m^\rho$.  Suppose after a refinement that $W_j$ agrees on each
current piece with such a unary function $U$.  On that piece apply the induction hypothesis
for $\eta$ to $(\bar P,A_j,U)$.  Since $U(y)=W_j(y)$ there, the transition
derivation is evaluated at the same numerical tuple
$(\bar P(y),A_j(y),W_j(y))$ as in the original recursion evaluation.
Lemma~\ref{lem:extensional-evaluation} identifies its internal value and
address data with the corresponding original transition subevaluation.
Evaluations of the external substitution functions are not included in
$\operatorname{Addr}_\eta$, so every recorded address is bounded by $K$.  This multiplies the number of pieces by at most
$Q_\eta(K)$ and makes $W_{j+1}$ piecewise represented by
$\mathcal M_m^\rho$.  Repeating at most $d\le K$ times gives
\[
 Q_\delta(K)
 =(K+1)Q_\gamma(K)Q_\eta(K)^K.
\]
The bound derivation supplies the value envelope and certifies admissibility;
it is not evaluated in this recurrence.

These recursive definitions are nondecreasing in $K$.  Finally, prove
\eqref{eq:direct-trace-growth} simultaneously by induction.  Products add
logarithms in the composition case.  In the recursion case,
\[
 \log Q_\delta(K)
 \le\log(K+1)+\log Q_\gamma(K)+K\log Q_\eta(K),
\]
so multiplication by $K$ raises the degree of a fixed polynomial bound by at
most one.  Restricting the original interval preserves every hypothesis and
uses the same recurrence.
\end{proof}

\section{Finite-range trace consequence}
\label{sec:finite-range-trace-consequence}

\begin{theorem}[finite-range trace bound]
\label{thm:trace}
Suppose $(m,n)$ lies in the trace band.  If unary
$f\in\Hclass{m}{n}{l}$ has finite range and
\[
 \Delta_f(N)=|\{y<N:f(y)\ne f(y+1)\}|,
\]
then constants $C,r\ge1$, depending only on one fixed derivation of
$f$, satisfy
\begin{equation}
\label{eq:weak-finite-range-trace}
 \Delta_f(N)
 \le
 \exp\!\left(
   C\bigl(1+D_{n,l}(N+C)\bigr)^r
 \right).
\end{equation}
For $n\ge2$ this implies
\[
 \Delta_f(N)=N^{o(1)}.
\]
The constants are derivation-dependent; no uniform rate over all derivations
of members of the class is asserted.
\end{theorem}

\begin{proof}
Let the range of $f$ lie in $\{0,\ldots,M\}$ and choose a unary
derivation $\delta$ of $f$ by Lemma~\ref{lem:derivation-tree-representation}.
Apply Theorem~\ref{thm:structural-trace} to $T(y)=y$ on
$\{0,\ldots,N\}$.  Lemma~\ref{lem:derivation-internal-envelope} supplies a
common address-depth bound
\[
 K\le C_0\bigl(1+D_{n,l}(N+C_0)\bigr).
\]
Thus the interval has a partition into at most
$\exp(C_1(1+D_{n,l}(N+C_1))^r)$ pieces on each of which $f$ agrees with a
nondecreasing function.  On one piece a finite-range nondecreasing function
changes value at most $M$ times.  Including the boundaries gives
\[
 \Delta_f(N)\le(M+1)q,
\]
where $q$ is the number of pieces, and constants absorb the factor $M+1$.
For $n\ge2$, Lemma~\ref{lem:depth-overhead} gives
$D_{n,l}(N+C)=O(\log\log(N+4))$.  Every fixed power of $\log\log N$ is
$o(\log N)$, so the right-hand side is $N^{o(1)}$.
\end{proof}

\begin{corollary}[class-level finite-range sparsity]
\label{cor:finite-range-sparsity}
Let $n\ge2$, $m\le n$, and $l\ge1$.  Every unary finite-range
$f\in\Hclass{m}{n}{l}$ satisfies
\[
 \Delta_f(N)=N^{o(1)}.
\]
\end{corollary}

\begin{proof}
Choose one derivation witnessing membership and apply
Theorem~\ref{thm:trace}.
\end{proof}

\section{The ordinary and horizontal thresholds}

\begin{theorem}[ordinary saturation and strictness]
\label{thm:saturation}
For every $n\ge2$ and $l\ge1$,
\[
 \boxed{
 \Hclass{m}{n}{l}
 =\Eclass{m}
 \quad\Longleftrightarrow\quad
 m\ge n+1.}
\]
Equivalently,
\[
 \Hclass{m}{n}{l}\subsetneq\Eclass{m}
 \qquad(m\le n).
\]
\end{theorem}

\begin{proof}

The upper inclusion is Lemma~\ref{lem:calibration}.

Suppose $m\ge n+1$.  The initial basis then contains $e_{n+1}$.
Put $G=g_n=e_n$ and $\rho=\rho_{n,l}$.  Since addition and $e_{n+1}$ are in
the initial basis, the fixed composition $M_l(y)=ly$ belongs to the class, and
hence, by the generator identity $e_{n+1}(x)=e_n^{[x]}(2)$ of
Lemma~\ref{lem:generator-identities}, so does
\[
 E(y)=e_{n+1}(M_l(y))=G^{[ly]}(2)=\varphi_{n,l}^{[y]}(2).
\]
Strict increase gives $\rho(\varphi_{n,l}(x))=x$, while
$\rho(2)=0$.  Therefore
\[
 \rho(E(0))=0,\qquad \rho(E(y+1))=E(y),\qquad D_{n,l}(E(y))=y+1.
\]
Addition also supplies the bounded selector
\[
 \operatorname{Sel}^{+}(u,v,0)=u,\qquad
 \operatorname{Sel}^{+}(u,v,c)=v\quad(c>0),
\]
with bound $u+v$: it is a step recursion along the descent $\rho$ on $c$ whose transition
ignores the recursive value.  Lemma~\ref{lem:internal-controls} supplies
$D_{n,l}$.  For a bounded recursion $f(\bar x,0)=g(\bar x)$,
\[
 f(\bar x,t+1)=h(\bar x,t,f(\bar x,t)),\qquad
 f(\bar x,t)\le b(\bar x,t),
\]
apply Lemma~\ref{lem:majorants} to $b(\bar x,t)$ and choose an earlier
coordinatewise nondecreasing composition function $B(\bar x,t)$ such that
$b(\bar x,t)\le B(\bar x,t)$ for all $\bar x,t$.  Define
\[
 T(\bar x,u,v)=\operatorname{Sel}^{+}
 \bigl(v,h(\bar x,D_{n,l}(\rho(u)),v),u\bigr)
\]
and the step recursion
\[
 W(\bar x,0)=g(\bar x),\qquad
 W(\bar x,z)=T(\bar x,\rho(z),W(\bar x,\rho(z)))\quad(z>0).
\]
Induction on $d=D_{n,l}(z)$ gives
\[
 W(\bar x,z)=f(\bar x,\max(d-1,0)).
\]
For $d\le1$ the selector takes the first branch.  For $d\ge2$, putting
$u=\rho(z)$ gives $D_{n,l}(\rho(u))=d-2$ and the second branch gives exactly
$f(\bar x,d-1)$.  The depth is nondecreasing, so
\[
 \widehat b(\bar x,z)=B(\bar x,D_{n,l}(z))
\]
belongs to the class and bounds $W$.  Indeed, $\rho$ is nondecreasing, so
$D_{n,l}$ is nondecreasing, and
\[
 W(\bar x,z)\le
 b(\bar x,\max(D_{n,l}(z)-1,0))
 \le B(\bar x,\max(D_{n,l}(z)-1,0))
 \le\widehat b(\bar x,z).
\]
Finally
\[
 W(\bar x,E(y))=f(\bar x,y).
\]
Thus every bounded recursion over $\Hclass{m}{n}{l}$ is available,
and induction over an $\Eclass{m}$ construction gives
\[
 \Eclass{m}\subseteq\Hclass{m}{n}{l}.
\]

For $m\le n$, define parity by the bounded recursion
\[
 p(0)=0,\qquad p(t+1)=\ISZ(p(t)),
\]
with declared bound $1$.  Hence $p\in\Eclass{0}\subseteq\Eclass{m}$,
$p(t)=t\bmod2$, and $\Delta_p(N)=N$.  The pair $(m,n)$ lies in the trace band.
Corollary~\ref{cor:finite-range-sparsity} gives
$\Delta_f(N)=N^{o(1)}$ for every finite-range
$f\in\Hclass{m}{n}{l}$.  Since $p$ changes at every adjacent pair,
$p\notin\Hclass{m}{n}{l}$.  Together with the ordinary upper calibration this
proves strictness.
\end{proof}

\section{Exact-depth alignment and vertical padding}
\label{sec:transfer}

In construction-order displays, $A\prec B$ means that $A$ is available at an
earlier finite stage.  At basis zero, padding a target descent with idle stages
would require retaining an arbitrary accumulator under a bound such as $v+C$,
which is unavailable.  We instead construct from any sufficiently deep target
point $r(y)$ a point $\widehat r(y)$ of exact depth $D_\sigma(y)$.  Matching
source and target depths makes every target stage perform one genuine source
update and removes the need for an unbounded hold/update selector.

\begin{lemma}[depth retraction and successor]
\label{lem:depth-retraction}
Let $\rho$ be the defining descent of $\mathcal A=A_\rho(B_m)$.  Then the
following maps belong to $\mathcal A$:
\[
\begin{aligned}
 I_\rho(a,0)&=a,\\
 I_\rho(a,z)&=\rho\bigl(I_\rho(a,\rho(z))\bigr)
 &&(z>0),\\
 \operatorname{Rev}_a(v)&=I_\rho(a,v),\\
 \operatorname{Ret}_a(v)&=\operatorname{Rev}_a(\operatorname{Rev}_a(v)),\\
 \operatorname{Up}_a(v)&=\operatorname{Rev}_a\bigl(
 \rho(\operatorname{Rev}_a(\operatorname{Ret}_a(v)))\bigr),\\
 \operatorname{Choose}(a,v,0)&=\operatorname{Ret}_a(v),\\
 \operatorname{Choose}(a,v,c)&=\operatorname{Up}_a(v)
 &&(c>0).
\end{aligned}
\]
Moreover
\[
 I_\rho(a,z)=\rho^{[D_\rho(z)]}(a),
 \qquad
 D_\rho(\operatorname{Ret}_a(v))=\min\{D_\rho(v),D_\rho(a)\},
\]
and on its nonzero branch $\operatorname{Up}_a$ increases that minimum by one,
stopping at $D_\rho(a)$.  Every value of $\operatorname{Ret}_a$ and
$\operatorname{Up}_a$ lies in the descent sequence from $a$ and is at most $a$.
\end{lemma}

\begin{proof}
The recursion for $I_\rho$ is admissible with bound the projection $a$.
Induction on $D_\rho(z)$ gives
\[
 I_\rho(a,z)=\rho^{[D_\rho(z)]}(a).
\]
We use the following elementary identity, valid for every descent and every
$a,k$:
\begin{equation}
\label{eq:depth-shift}
 D_\rho(\rho^{[k]}(a))=\max\{D_\rho(a)-k,0\}.
\end{equation}
Indeed, if $d=D_\rho(a)$ and $k\le d$, then $d-k$ further iterates reach
zero.  Fewer than $d-k$ further iterates would make $a$ reach zero in fewer
than $d$ steps, contradicting the definition of $d$.  If $k>d$, the iterate
is already zero.

Put $A=D_\rho(a)$ and $V=D_\rho(v)$.  The iterate identity and
\eqref{eq:depth-shift} give
\[
 D_\rho(\operatorname{Rev}_a(v))=\max\{A-V,0\}.
\]
Applying this formula a second time yields
\[
 \begin{aligned}
 D_\rho(\operatorname{Ret}_a(v))
 &=\max\{A-\max(A-V,0),0\}\\
 &=\min\{V,A\}.
 \end{aligned}
\]
Let $R=\min\{V,A\}$.  Then
$D_\rho(\operatorname{Rev}_a(\operatorname{Ret}_a(v)))=A-R$.
After one application of $\rho$ this becomes
$\max\{A-R-1,0\}$, and the final reflection gives
\[
 D_\rho(\operatorname{Up}_a(v))=\min\{R+1,A\}.
\]
Thus $\operatorname{Up}_a$ advances the retained depth by exactly one unless
it has already reached $A$.  All displayed maps take values in the descent sequence generated by $\rho$
from $a$, hence their values are at most $a$.

Finally $\operatorname{Choose}$ is the step recursion along the descent $\rho$ on $c$ with
base $\operatorname{Ret}_a(v)$ and transition constantly
$\operatorname{Up}_a(v)$, ignoring the previous state.  The projection $a$
is a monotone bound.  The finite dependency order
\[
 I_\rho\prec\operatorname{Rev},\operatorname{Ret}\prec\operatorname{Up}
 \prec\operatorname{Choose}
\]
therefore defines a valid construction in the closure.
\end{proof}

\begin{lemma}[exact-depth realization]
\label{lem:exact-depth-point}
Let $\rho$ be the defining descent of $\mathcal A=A_\rho(B_m)$, and let
$\sigma\in\mathcal A$ be any source descent.  If $r\in\mathcal A$ satisfies
$D_\rho(r(y))\ge D_\sigma(y)$, then there is $\widehat r\in\mathcal A$ with
\[
 D_\rho(\widehat r(y))=D_\sigma(y)
\]
for every $y$.  Explicitly, with the maps of
Lemma~\ref{lem:depth-retraction},
\[
\begin{aligned}
 I_\sigma(y,0)&=y,\\
 I_\sigma(y,z)&=\sigma\bigl(I_\sigma(y,\rho(z))\bigr)
 &&(z>0),\\
 Q(y,a,0)&=0,\\
 Q(y,a,z)&=\operatorname{Choose}\Bigl(
 a,Q(y,a,\rho(z)),
 \NZ\bigl(I_\sigma(y,\operatorname{Ret}_a(Q(y,a,\rho(z))))\bigr)
 \Bigr)
 &&(z>0),
\end{aligned}
\]
and $\widehat r(y)=Q(y,r(y),r(y))$ works.
\end{lemma}

\begin{proof}
The iterator $I_\sigma$ is a step recursion along the target descent $\rho$,
bounded by the projection $y$, and $I_\sigma(y,z)=\sigma^{[D_\rho(z)]}(y)$.  The recursion $Q$ is
admissible with bound $a$ because its state always lies in the descent sequence generated by $\rho$ from $a$.  To verify the state invariant, let $t=D_\rho(z)$ and suppose the
preceding state has target depth
\[
 k=\min\{t-1,D_\sigma(y),D_\rho(a)\}.
\]
The retraction keeps that depth, while
\[
 I_\sigma\bigl(y,\operatorname{Ret}_a(Q(y,a,\rho(z)))\bigr)>0
 \quad\Longleftrightarrow\quad k<D_\sigma(y).
\]
Thus $\operatorname{Choose}$ applies $\operatorname{Up}_a$ exactly when one
more source level remains, and $\operatorname{Up}_a$ itself stops at
$D_\rho(a)$.  Starting from depth $0$ at $z=0$, induction on $t$ gives
\[
 D_\rho(Q(y,a,z))=
 \min\{D_\rho(z),D_\sigma(y),D_\rho(a)\}.
\]
Substituting $a=r(y)$ and using $D_\rho(r(y))\ge D_\sigma(y)$ yields
\begin{equation}
\label{eq:exact-depth-alignment}
 D_\rho(\widehat r(y))=D_\sigma(y).
\end{equation}
The finite dependency order $I_\sigma\prec Q\prec\widehat r$, after the maps of
Lemma~\ref{lem:depth-retraction}, defines a valid construction in the closure.
\end{proof}

\begin{lemma}[aligned source simulation]
\label{lem:aligned-simulation}
Under the hypotheses of Lemma~\ref{lem:exact-depth-point}, every admissible
step recursion along the source descent $\sigma$ whose base, transition, and bound belong to
$\mathcal A$ also belongs to $\mathcal A$.
\end{lemma}

\begin{proof}
Let $\widehat r$ be the exact-depth map of Lemma~\ref{lem:exact-depth-point}.  Define
\[
 J(v,0)=v,\qquad J(v,c)=\sigma(v)\quad(c>0),
\]
a step recursion along the target descent $\rho$ on $c$, bounded by $v$, and
\[
\begin{aligned}
 T(y,z,0)&=y,\\
 T(y,z,w)&=J\Bigl(
 T(y,z,\rho(w)),
 \NZ\bigl(I_\sigma(T(y,z,\rho(w)),z)\bigr)\Bigr)
 \qquad(w>0).
\end{aligned}
\]
Write $d=D_\sigma(y)$ and $e=D_\rho(z)$.  If the current state is
$\sigma^{[r]}(y)$, then its test is nonzero exactly when $r+e<d$.
Hence each target stage applies one $\sigma$ until $r=d-e$, after which the
state is retained.  Since $\widehat r(y)$ provides exactly $d$ stages, induction
gives
\[
 \operatorname{Tail}(y,z):=T(y,z,\widehat r(y))=
 \sigma^{[\max\{D_\sigma(y)-D_\rho(z),0\}]}(y).
\]
The state of $T$ only descends from $y$, so its bound is the projection $y$.
Suppose the source recursion has base $f(\bar x,0)=g(\bar x)$, transition
\[
 f(\bar x,y)=h(\bar x,\sigma(y),f(\bar x,\sigma(y)))\quad(y>0),
\]
and earlier pointwise bound $b(\bar x,y)$.  Apply
Lemma~\ref{lem:majorants} to $b(\bar x,y)$ and choose an earlier
coordinatewise nondecreasing composition function $B(\bar x,y)$ satisfying
$b(\bar x,y)\le B(\bar x,y)$ for all $\bar x,y$.  Use the
zero-or-value selector from Lemma~\ref{lem:internal-zero-value-selector} and
put $E=\operatorname{ZV}$.  Now define a step recursion along the target
descent $\rho$ by
\[
 W(\bar x,y,0)=g(\bar x),
\]
and, with $z=\rho(w)$, put
\[
 W(\bar x,y,w)=E\!\left(
 h(\bar x,\operatorname{Tail}(y,z),W(\bar x,y,z)),
 \NZ(I_\sigma(y,z))\right)
 \qquad(w>0).
\]
Let $d=D_\sigma(y)$ and $t=D_\rho(w)$.  The case $t=0$ is the
recursion base and gives
$W(\bar x,y,0)=g(\bar x)=f(\bar x,0)$.  If $1\le t\le d$, then
$D_\rho(z)=t-1<d$, so $I_\sigma(y,z)>0$, and
\[
 \operatorname{Tail}(y,z)=\sigma^{[d-(t-1)]}(y).
\]
Induction on $t$ therefore gives, for every $0\le t\le d$,
\[
 W(\bar x,y,w)=f(\bar x,\sigma^{[d-t]}(y)).
\]
If $t>d$, the selector returns $0$.  Hence, for every $w$,
\[
 W(\bar x,y,w)\le B(\bar x,y).
\]
For $t\le d$, source admissibility gives a value bounded by $b$ at an
address at most $y$, and the monotonicity of $B$ gives the displayed bound;
for $t>d$, it is immediate.  Thus $B(\bar x,y)$ is a global
target-recursion bound independent of $w$.
At the exact-depth input $\widehat r(y)$ the preceding identity proves
\[
 W(\bar x,y,\widehat r(y))=f(\bar x,y).
\]
The remaining finite dependency order
\[
 T\prec\operatorname{Tail}\prec W
\]
defines a valid construction in the closure after $\widehat r$.  No addition, pairing, or unbounded selector is used, so the
argument applies unchanged when $m=0$.
\end{proof}

\begin{lemma}[exact-depth transfer]
\label{lem:depth-transfer}
Let $\rho$ be the defining descent of the target class
$\mathcal A=A_\rho(B_m)$, in the sense of
Definition~\ref{def:descent-closure}.  Let $\sigma\in\mathcal A$ be any
source descent, not necessarily a generalized inverse.  If
$r\in\mathcal A$ satisfies
\[
 D_\rho(r(y))\ge D_\sigma(y),
\]
then every admissible step recursion along the descent $\sigma$ whose base, transition, and
bound belong to $\mathcal A$ also belongs to $\mathcal A$.
\end{lemma}

\begin{proof}
Lemma~\ref{lem:depth-retraction} constructs the descent controls on a fixed sequence generated by the target
descent $\rho$.  Lemma~\ref{lem:exact-depth-point} produces an aligned input
$\widehat r(y)$ of exact depth $D_\sigma(y)$.
Lemma~\ref{lem:aligned-simulation} simulates the source recursion at that
aligned point.  The full dependency order is
\[
 I_\rho\prec\operatorname{Rev},\operatorname{Ret}\prec\operatorname{Up}
 \prec\operatorname{Choose}\prec I_\sigma\prec Q\prec\widehat r\prec T
 \prec\operatorname{Tail}\prec W.
\]
\end{proof}

\begin{lemma}[vertical depth domination]
\label{lem:vertical-domination}
Let $2\le n<n'$ and $l,l'\ge1$.  Put
\[
 \alpha=\varphi_{n,l},\qquad
 \beta=\varphi_{n',l'},\qquad
 y_s=\beta^{[s]}(0),\qquad
 c_s=D_{n,l}(y_s).
\]
Then $D_{n',l'}(y_s)=s$, and for every fixed $C,r\ge1$,
\[
 c_s>\exp(Cs^r)
\]
for all sufficiently large $s$.
\end{lemma}

\begin{proof}
Put $F=g_n=e_n$, so $\alpha=F^{[l]}$.  Since $n'\ge n+1$ and the generating functions are pointwise increasing
with their indices,
\[
 \beta(x)\ge e_{n+1}(x)=F^{[x]}(2).
\]
The threshold law gives $D_{n',l'}(y_s)=s$ and, whenever $c_s>0$,
\[
 y_s>\alpha^{[c_s-1]}(0)=F^{[l(c_s-1)]}(0).
\]
Moreover,
\[
 y_{s+1}=\beta(y_s)\ge F^{[y_s]}(2).
\]
If $d\le y_s/l$, then
\[
 \alpha^{[d]}(0)=F^{[ld]}(0)
 <F^{[y_s]}(2)\le y_{s+1}.
\]
Another application of the threshold law yields
\[
 c_{s+1}>\left\lfloor\frac{y_s}{l}\right\rfloor.
\]
Since $F\ge e_2$ and
$e_2^{[t]}(0)\ge2^{2^{t-1}}$ for $t\ge1$, the preceding inequalities imply,
after discarding finitely many $s$,
\[
 c_{s+1}\ge2^{c_s}.
\]
Indeed,
\[
 y_s>
 F^{[l(c_s-1)]}(0)
 \ge2^{2^{l(c_s-1)-1}},
\]
and the last quantity eventually exceeds $l2^{c_s}$.  Starting at a
sufficiently large $s_0$, induction now gives, for example,
\[
 c_s\ge2^{2^{s-s_0-1}}\qquad(s>s_0).
\]
This lower bound eventually exceeds every $\exp(Cs^r)$.
\end{proof}

\begin{lemma}[vertical padding]
\label{lem:vertical-padding}
If $n'>n\ge2$, then for every $l,l'\ge1$ an internal $r$ exists in
$\Hclass{m}{n}{l}$ with
\[
 D_{n,l}(r(y))\ge D_{n',l'}(y).
\]
\end{lemma}

\begin{proof}
Put $\alpha=\varphi_{n,l}$, $\beta=\varphi_{n',l'}$,
\[
 y_s=\beta^{[s]}(0),\qquad c_s=D_{n,l}(y_s).
\]
Lemma~\ref{lem:vertical-domination} gives $c_s/s\to\infty$.  Choose $s_0$ with
$c_s\ge s+1$ for $s\ge s_0$.  By generalized-inverse identities,
\[
 y_s>\alpha^{[c_s-1]}(0)\ge\alpha^{[s]}(0).
\]
Thus $D_{n,l}(y)=s\ge s_0$ implies
\[
 y\le\alpha^{[s]}(0)<\beta^{[s]}(0),
\]
and hence $D_{n',l'}(y)\le D_{n,l}(y)$.  There are only finitely many
remaining $y$.  Choose $Y\ge1$ exceeding all of them and put
\[
 M=\max_{y<Y}D_{n',l'}(y),
 \qquad
 R=\alpha^{[M]}(0).
\]
The threshold law gives $D_{n,l}(R)=M$, and therefore
\[
 D_{n,l}(R)\ge D_{n',l'}(y)\qquad(y<Y).
\]

It remains to verify that the finite patch is internal even at initial-basis
level $0$.  Choose $p$ with
$C_p=\varphi_{n,l}^{[p]}(0)\ge Y-1$ and put $a=C_p-Y+1$.  By the threshold
law,
\[
 \rho_{n,l}^{[p]}(a+y)=0
 \quad\Longleftrightarrow\quad y<Y.
\]
The shift $a+y$ means the fixed successor iterate $S^{[a]}(y)$, so no
addition is required.  Hence
\[
 \operatorname{GE}_Y(y)=
 \NZ(\rho_{n,l}^{[p]}(S^{[a]}(y)))
\]
belongs to every row.  Define a step recursion along the target descent $\rho$ on a control input
$c$ by
\[
 J(y,0)=R,\qquad J(y,c)=y\quad(c>0).
\]
The fixed shift $S^{[R]}(y)$ bounds both branches.  Consequently the map
\[
 r(y)=J(y,\operatorname{GE}_Y(y))
 =\begin{cases}R,&y<Y,\\y,&y\ge Y\end{cases}
\]
is internal and satisfies the required depth inequality.
\end{proof}

\begin{lemma}[finite controls and shifted tables]
\label{lem:finite-controls}
Fix a class $\Hclass{m}{n}{l}$.
\begin{enumerate}
\item Every map from a fixed finite interval
$\{0,\ldots,C\}$ to $\Nat$ extends, with any fixed default value, to a
member of $\Hclass{m}{n}{l}$.
\item If $a_0,\ldots,a_C$ are fixed nonnegative integers, then the
parameterized table
\[
 T(u,i)=S^{[a_i]}(u)\quad(0\le i\le C)
\]
with any fixed shifted default is in $\Hclass{m}{n}{l}$.
\end{enumerate}
Both assertions hold at $m=0$.
\end{lemma}

\begin{proof}
We first construct the fixed thresholds directly.  For $i\ge1$, choose
$p$ with $C_p=\varphi_{n,l}^{[p]}(0)\ge i-1$ and put
$a=C_p-i+1$.  The generalized-inverse identities give
\[
 \rho_{n,l}^{[p]}(S^{[a]}(y))=0
 \quad\Longleftrightarrow\quad y<i.
\]
Hence
\[
 \operatorname{GE}_i(y)=
 \NZ\!\left(\rho_{n,l}^{[p]}(S^{[a]}(y))\right)
\]
is the predicate $y\ge i$; take $\operatorname{GE}_0=1$.
Lemma~\ref{lem:internal-controls} supplies the zero and nonzero tests.  A
one-step recursion on a Boolean control, bounded by $1$, supplies Boolean
conjunction, so
\[
 [y=i]=
 \operatorname{GE}_i(y)\wedge
 \ISZ(\operatorname{GE}_{i+1}(y)).
\]
Another one-step recursion selects between any two fixed constants, using
their maximum as a fixed bound.  A finite chain of these selectors proves
the first assertion.

For the second assertion, put $A=\max_i a_i$ including the default shift.
Every branch is bounded by the earlier shifted projection $S^{[A]}(u)$.
A conditional selecting two such branches is a step recursion on its
Boolean control, with the same shifted projection as its declared bound.
A fixed finite chain of these conditionals constructs $T$.  Only successor,
projections, fixed Boolean controls, and displayed shifted-projection bounds
were used, so the construction is valid when $m=0$.
\end{proof}

\begin{proposition}[row zero collapses to the ordinary hierarchy]
\label{prop:row-zero}
For every $m\in\Nat$ and every $l\ge1$,
\[
 \boxed{\Hclass{m}{0}{l}=\Eclass{m}.}
\]
\end{proposition}

\begin{proof}
The upper inclusion is Lemma~\ref{lem:calibration}; for $l=1$ the reverse
inclusion is \eqref{eq:ordinary-boundary}.  Fix $l\ge2$, write
$\rho(y)=y\dotminus l$, and show that $\Hclass{m}{0}{l}$ is closed under
ordinary bounded recursion.

First construct the alignment controls.  The step recursion
\[
 I(y,0)=y,\qquad I(y,z)=\rho(I(y,\rho(z)))\quad(z>0)
\]
has bound $y$ and value $\rho^{[D_{0,l}(z)]}(y)$.  Thus
$W(y)=I(y,\rho(y))$ is $0$ at $y=0$ and otherwise
$((y-1)\bmod l)+1$.  Lemma~\ref{lem:finite-controls} then gives the round-up
\[
 U(y)=l\left\lceil\frac yl\right\rceil,
 \qquad y\le U(y)\le y+l-1,
 \qquad U(\rho(y))+l=U(y)\ (y>0),
\]
using only a fixed shifted table.  For fixed $j\le l$, above the threshold
$y\ge l$ one has $y\dotminus j=S^{[l-j]}(\rho(y))$; below it a finite table
applies, and $\operatorname{GE}_l$ joins the two branches.  The residue
$y\bmod l$ is a fixed table of $W(y)$.  Hence all these controls are internal,
also at $m=0$.

Let an earlier bounded recursion be given by
\[
 f(\bar x,0)=g(\bar x),\qquad
 f(\bar x,y+1)=h(\bar x,y,f(\bar x,y)),\qquad f\le b.
\]
Choose by Lemma~\ref{lem:majorants} an earlier coordinatewise nondecreasing
composition majorant $B\ge b$.  Let $H(\bar x,z,v)$ be the fixed $l$-fold
composition of $h$ at the successive addresses
$U(z),U(z)+1,\ldots,U(z)+l-1$, and define
\[
 F(\bar x,0)=g(\bar x),\qquad
 F(\bar x,y)=H(\bar x,\rho(y),F(\bar x,\rho(y)))\quad(y>0).
\]
Induction on $D_{0,l}(y)$, using $U(\rho(y))+l=U(y)$, gives
$F(\bar x,y)=f(\bar x,U(y))$.  Every intermediate address is at most
$y+l-1$, so $B(\bar x,S^{[l-1]}(y))$ is an earlier bound.

Write
\[
 \operatorname{sub}_i(y):=y\dotminus i,
 \qquad
 \operatorname{res}(y):=y\bmod l.
\]
The preceding paragraph shows that these fixed controls are internal.  For
$0\le i<l$ put $u_i(y)=U(\operatorname{sub}_i(y))$ and let
$T_i(\bar x,y)$ apply $i$ further $h$-updates to
$F(\bar x,\operatorname{sub}_i(y))$ starting at address $u_i(y)$.
Then
\[
 T_i(\bar x,y)=f(\bar x,u_i(y)+i),
\]
and all branches have the same bound above.  If $i=\operatorname{res}(y)$,
then $u_i(y)=y-i$, hence $T_i(\bar x,y)=f(\bar x,y)$.  A fixed finite chain
of conditionals on $[\operatorname{res}(y)=i]$, supplied by
Lemma~\ref{lem:finite-controls}, therefore computes $f$.  Thus the class is
closed under bounded recursion, so construction induction gives
$\Eclass{m}\subseteq\Hclass{m}{0}{l}$.
\end{proof}

In view of Proposition~\ref{prop:row-zero}, we write
$\Hclass{m}{0}{*}$ for the common class
$\Hclass{m}{0}{l}=\Eclass{m}$, $l\ge1$.

\begin{lemma}[comparison of descent depths]
\label{lem:descent-comparison}
Let $\tau$ and $\rho$ be nondecreasing descent functions.  If
$\tau(y)\le\rho(y)$ for every $y$, then
\[
 \tau^{[j]}(y)\le\rho^{[j]}(y)
 \qquad(j,y\in\Nat),
\]
and consequently
\[
 D_\tau(y)\le D_\rho(y)
 \qquad(y\in\Nat).
\]
\end{lemma}

\begin{proof}
The iterate inequality is induction on $j$.  The case $j=0$ is equality.  If
it holds at $j$, then monotonicity of $\tau$ and the pointwise comparison give
\[
 \tau^{[j+1]}(y)
 =\tau(\tau^{[j]}(y))
 \le\tau(\rho^{[j]}(y))
 \le\rho(\rho^{[j]}(y))
 =\rho^{[j+1]}(y).
\]
At $j=D_\rho(y)$ the right-hand side is $0$, so the left-hand side is also
$0$.  Hence $D_\tau(y)\le D_\rho(y)$.
\end{proof}

\begin{lemma}[fixed-iterate descent-depth recovery]
\label{lem:fixed-power-recovery}
Let $\rho$ be the defining descent of $\Hclass{m}{n}{l}$.  Let $\sigma$ be
the generalized inverse of a fixed strictly increasing step function $\psi$
satisfying $\psi(x)\ge x+1$,
suppose $\sigma\in\Hclass{m}{n}{l}$, and suppose that for some fixed
$c\ge1$,
\[
 \tau:=\sigma^{[c]}\le\rho
\]
pointwise.  Then $D_\sigma\in\Hclass{m}{n}{l}$.
\end{lemma}

\begin{proof}
Put $\tau=\sigma^{[c]}$.  The step recursion
\[
 I_\tau(y,0)=y,\qquad I_\tau(y,z)=\tau(I_\tau(y,\rho(z)))\quad(z>0)
\]
has bound $y$ and value $\tau^{[D_\rho(z)]}(y)$.  Define
\[
 C_c(y,0)=0,
\]
\[
 C_c(y,z)=
 \begin{cases}
 S^{[c]}(C_c(y,\rho(z))),&I_\tau(y,\rho(z))>0,\\
 C_c(y,\rho(z)),&I_\tau(y,\rho(z))=0
 \end{cases}
 \qquad(z>0).
\]
The conditional is supplied by Lemma~\ref{lem:finite-controls}.  Since $\tau$
is the inverse of $\psi^{[c]}$,
\[
 D_\tau(y)=\left\lceil\frac{D_\sigma(y)}c\right\rceil,
\]
and $\tau\le\rho$ gives $D_\tau\le D_\rho$ by
Lemma~\ref{lem:descent-comparison}.  Thus $C_c(y,y)=cD_\tau(y)\le y+c$, so
the shifted projection $S^{[c]}(y)$ is an admissible bound.

Next let
\[
 K(v,0)=v,\qquad K(v,u)=\tau(v)\quad(u>0),
\]
and define
\[
 L(y,0)=y,\qquad
 L(y,z)=K\!\left(L(y,\rho(z)),
 \NZ\bigl(\tau(L(y,\rho(z)))\bigr)\right)\quad(z>0).
\]
Both recursions are bounded by $y$.  Since $D_\tau(y)\le D_\rho(y)$,
$L(y,y)$ is the last positive point of the $\tau$-descent from $y$.
The finite set $\{w:D_\tau(w)=1\}=\{1,\ldots,\psi^{[c]}(0)\}$ therefore
admits, by Lemma~\ref{lem:finite-controls}, a fixed table
$\theta(w)=D_\sigma(w)$ (with $\theta(0)=0$).  Put
$s(y)=\theta(L(y,y))$ and use the shifted-table map
$A_c(u,s)=S^{[s]}(u)$ for $0\le s\le c$.

Finally, suppress the positive branch at the origin by a one-step selector
$\operatorname{Orig}(z,0)=0$, $\operatorname{Orig}(z,y)=z$ for $y>0$, and
set
\[
 D_\sigma^\star(y)=\operatorname{Orig}\!\left(
 A_c(C_c(\tau(y),\tau(y)),s(y)),y\right).
\]
If $y>0$, write $d=D_\sigma(y)$ and
$q=D_\tau(y)=\lceil d/c\rceil$.  Then
$L(y,y)=\sigma^{[c(q-1)]}(y)$, so $s(y)=d-c(q-1)$, while
$C_c(\tau(y),\tau(y))=c(q-1)$.  Thus $D_\sigma^\star(y)=d$, and the value at
$0$ is $0$.  Every operation used is an earlier step recursion, fixed
successor shift, or finite table, so the construction is valid also at
$m=0$.
\end{proof}

\begin{lemma}[adjacent vertical step function factorization]
\label{lem:adjacent-factorization}
For every $s\ge2$,
\[
 \rho_{s+1,1}=D_{s,1}\circ\rho_{s,1}.
\]
Moreover, for every $n\ge2$, $m\in\Nat$, and $l\ge1$,
$\rho_{n+1,1}\in\Hclass{m}{n}{l}$.
\end{lemma}

\begin{proof}
We first prove the algebraic assertion with the variable $s$ only.  The
factorization is
\[
 \rho_{s+1,1}=D_{s,1}\circ\rho_{s,1}.
\]
Put $F=g_s=e_s$ and $z=\rho_{s,1}(y)$.  For every $k$, the threshold
law gives, using $F(0)=2$ and the generator identity
$e_{s+1}(k)=e_s^{[k]}(2)$ of Lemma~\ref{lem:generator-identities},
\[
 \begin{aligned}
 D_{s,1}(z)\le k
 &\Longleftrightarrow z\le F^{[k]}(0)\\
 &\Longleftrightarrow y\le F^{[k+1]}(0)
 =F^{[k]}(2)=e_{s+1}(k)=\varphi_{s+1,1}(k).
 \end{aligned}
\]
The least such $k$ is exactly $\rho_{s+1,1}(y)$.

For the membership assertion, now fix $n,m,l$ as in the statement and write
$\delta=\rho_{n,1}$, $\rho=\rho_{n,l}=\delta^{[l]}$, and
$D=D_{n,l}$.  For $y>0$, let $d=D_{n,1}(y)$ and
$q=D(y)=\lceil d/l\rceil$.  The last positive node of the descent sequence generated by $\rho$,
\[
 W(y)=\rho^{[q-1]}(y),
\]
lies in the fixed interval
$1\le W(y)\le\varphi_{n,l}(0)$.  The iterator
\[
 I(y,0)=y,\qquad I(y,z)=\rho(I(y,\rho(z)))\quad(z>0)
\]
is bounded by $y$ and satisfies $I(y,z)=\rho^{[D(z)]}(y)$; hence
$W(y)=I(y,\rho(y))$ for $y>0$.

Lemma~\ref{lem:finite-controls} supplies a table $R_l$ satisfying
\[
 R_l(w)=D_{n,1}(w)
 \qquad(0\le w\le\varphi_{n,l}(0)).
\]
Thus
\[
 r(y)=R_l(W(y))=d-l(q-1)\in\{1,\ldots,l\}.
\]

The recursion
\[
 M_l(0)=0,\qquad M_l(y)=S^{[l]}(M_l(\rho(y)))\quad(y>0)
\]
computes $lD(y)$.  It has a shifted-projection bound: every step function
step raises its input by at least $2l$.  If $D(y)=q>0$, then
\[
 y>\varphi_{n,l}^{[q-1]}(0)\ge2l(q-1),
\]
so $lD(y)\le y+l$.  Thus $S^{[l]}(y)$ is an explicit admissible bound.
Define the fixed table
\[
 A_l(u,0)=0,\qquad
 A_l(u,s)=S^{[s-1]}(u)\quad(1\le s\le l),
\]
with bounded default $0$ outside $\{0,\ldots,l\}$.  It is assembled from
Lemma~\ref{lem:finite-controls} and is bounded by $S^{[l]}(u)$.  For $y>0$
it gives
\[
 A_l(M_l(\rho(y)),r(y))
 =l(q-1)+(r(y)-1),
\]
and at $y=0$ it gives $0$.  This reconstructs
\[
 \rho_{n+1,1}(y)
 =D_{n,1}(y)-1
 =l(q-1)+(r(y)-1)
\]
for $y>0$, with the already specified value $0$ at the origin.  Hence the
adjacent higher inverse belongs to the lower-index step-recursion class.
\end{proof}

\begin{lemma}[fixed-iterate domination across row indices]
\label{lem:fixed-iterate-domination}
If $s>n\ge2$ and $l\ge1$, a fixed $c\ge1$ satisfies
\[
 g_s^{[c]}(x)\ge\varphi_{n,l}(x)
 \qquad(x\in\Nat).
\]
Consequently
\[
 \rho_{s,1}^{[c]}\le\rho_{n,l}.
\]
\end{lemma}

\begin{proof}
Put $F=g_n=e_n$.  Since $s>n$, the generator comparison and the generator
identity of Lemma~\ref{lem:generator-identities} give
\[
 g_s(x)=e_s(x)\ge e_{n+1}(x)=F^{[x]}(2).
\]
Since $F(u)\ge u+2$, one has
\[
 F^{[x-l]}(2)\ge2+2(x-l)>x
\]
for every sufficiently large $x$.  Hence
$F^{[x]}(2)>F^{[l]}(x)=\varphi_{n,l}(x)$ eventually.  Only finitely many
arguments remain.  For each such fixed argument, the iterates of $g_s$ are
unbounded; choosing one exponent $c$ as the maximum of the finitely many
required exponents gives
\[
 g_s^{[c]}(x)\ge\varphi_{n,l}(x)
 \qquad(x\in\Nat).
\]
Reversing this inequality for generalized inverses gives
\[
 \rho_{s,1}^{[c]}\le\rho_{n,l}.
\]
\end{proof}

\begin{lemma}[higher-index descents and depths are internal to lower rows]
\label{lem:higher-step-function-internal}
If $2\le n<n'$ and $l,l'\ge1$, then
\[
 \rho_{n',l'},D_{n',l'}\in\Hclass{m}{n}{l}
\]
for every $m$.
\end{lemma}

\begin{proof}
Write $\rho=\rho_{n,l}$.  Lemma~\ref{lem:adjacent-factorization} first
constructs $\rho_{n+1,1}$.  Suppose $s>n$ and $\rho_{s,1}$ has already been
constructed.  Lemma~\ref{lem:fixed-iterate-domination} supplies $c$ with
$\rho_{s,1}^{[c]}\le\rho$.  Lemma~\ref{lem:fixed-power-recovery} constructs
$D_{s,1}$, and Lemma~\ref{lem:adjacent-factorization} then constructs
$\rho_{s+1,1}$.  Repeating these two steps constructs $\rho_{n',1}$ in the
required stage order.

The generalized-inverse identities give
\[
 \rho_{n',l'}=\rho_{n',1}^{[l']},
\]
so the blocked inverse is internal.  Apply
Lemma~\ref{lem:fixed-iterate-domination} at $s=n'$ and enlarge its exponent
if necessary.  Since $\rho_{n',l'}\le\rho_{n',1}$,
\[
 \rho_{n',l'}^{[c]}
 \le \rho_{n',1}^{[c]}
 \le\rho.
\]
Applying Lemma~\ref{lem:fixed-power-recovery} to
$\rho_{n',l'}$ constructs $D_{n',l'}$.
\end{proof}

\begin{lemma}[horizontal collapse]
\label{lem:horizontal-collapse}
If $n\ge1$ and $m\ge n$, then for all $l,l'\ge1$,
\[
 \Hclass{m}{n}{l}=\Hclass{m}{n}{l'}.
\]
\end{lemma}

\begin{proof}
It suffices first to compare $l$ with $kl$.  Put
$G=g_n$, $\rho=\rho_{n,l}$ and $\sigma=\rho_{n,kl}=\rho^{[k]}$.
Exact-depth transfer with the identity padding gives
$\Hclass{m}{n}{kl}\subseteq\Hclass{m}{n}{l}$.

For the converse work in $\mathcal A=\Hclass{m}{n}{kl}$.  Since $m\ge n$,
$G$ (and addition) is available, and exact cancellation gives the internal
source descent
\begin{equation}
\label{eq:horizontal-rho-recovery}
 \rho(y)=\sigma(G^{[l(k-1)]}(y)).
\end{equation}
Addition supplies the bounded selector
$\operatorname{Sel}^{+}(u,v,0)=u$, $\operatorname{Sel}^{+}(u,v,c)=v$ for
$c>0$, with bound $u+v$.

Consider an earlier $\rho$-recursion $f$ with base $g$, transition $h$ and
bound $b$, and choose an earlier coordinatewise nondecreasing composition
majorant $B\ge b$.  For a current source node $c$ and
$v=f(\bar x,\sigma(c))$, define $v_k=v$ and, for $j=k-1,\ldots,0$,
\begin{equation}
\label{eq:horizontal-block}
 v_j=\operatorname{Sel}^{+}\!\left(
 v_{j+1},h(\bar x,\rho^{[j+1]}(c),v_{j+1}),\rho^{[j]}(c)\right).
\end{equation}
Let $\operatorname{Block}(\bar x,c,v)=v_0$.  Downward induction gives
\begin{equation}
\label{eq:horizontal-block-correct}
 \operatorname{Block}(\bar x,c,f(\bar x,\sigma(c)))=f(\bar x,c),
\end{equation}
including the final incomplete block because zero controls skip the nominal
updates.

To make this block work for arbitrary target inputs, use the retraction maps
of Lemma~\ref{lem:depth-retraction} for $\sigma$.  Put
\[
 \operatorname{Cap}^{\sigma}_a(z)=
 \operatorname{Rev}^{\sigma}_a(\operatorname{Ret}^{\sigma}_a(z)).
\]
If $A=D_\sigma(a)$ and $Z=D_\sigma(z)$, then
$\operatorname{Cap}^{\sigma}_a(z)=0$ exactly when $Z\ge A$; when $Z<A$,
\[
 \sigma(\operatorname{Up}^{\sigma}_a(z))=
 \operatorname{Ret}^{\sigma}_a(z),
\]
and for $w>0$, $z=\sigma(w)$,
$\operatorname{Up}^{\sigma}_a(z)=\operatorname{Ret}^{\sigma}_a(w)$.

Now define a single $\sigma$-recursion, retaining the original source address
$y$ as a parameter:
\[
 V(\bar x,y,0)=g(\bar x),
\]
\begin{equation}
\label{eq:horizontal-coarse-recursion}
 V(\bar x,y,w)=\operatorname{Sel}^{+}\!\left(
 V(\bar x,y,z),
 \operatorname{Block}(\bar x,\operatorname{Up}^{\sigma}_y(z),V(\bar x,y,z)),
 \operatorname{Cap}^{\sigma}_y(z)\right),
 \qquad z=\sigma(w).
\end{equation}
Induction on $D_\sigma(w)$ gives the all-input invariant
\begin{equation}
\label{eq:horizontal-total-invariant}
 V(\bar x,y,w)=f(\bar x,\operatorname{Ret}^{\sigma}_y(w)).
\end{equation}
Indeed, put $z=\sigma(w)$, $A=D_\sigma(y)$, and $Z=D_\sigma(z)$.  If
$Z\ge A$, then $\operatorname{Cap}^{\sigma}_y(z)=0$, the selector retains
the previous value, and both relevant retractions equal the retained point on
the descent from $y$.  If $Z<A$, the control is nonzero;
$\sigma(\operatorname{Up}^{\sigma}_y(z))=\operatorname{Ret}^{\sigma}_y(z)$
and $\operatorname{Up}^{\sigma}_y(z)=\operatorname{Ret}^{\sigma}_y(w)$, so
the induction hypothesis and \eqref{eq:horizontal-block-correct} give the
required current value.  Hence $V(\bar x,y,y)=f(\bar x,y)$.  Every retracted
or intermediate block address lies on the descent from $y$, so it is at most
$y$; source admissibility and monotonicity of $B$ give the global target bound
$B(\bar x,y)$.  Thus $\Hclass{m}{n}{l}\subseteq\Hclass{m}{n}{kl}$.

For arbitrary $l,l'$, apply the result to the common multiple $L=ll'$:
\[
 \Hclass{m}{n}{l}=\Hclass{m}{n}{L}=\Hclass{m}{n}{l'}.
\]
\end{proof}

Whenever Lemma~\ref{lem:horizontal-collapse} applies, write
$\Hclass{m}{n}{*}$ for the common class $\Hclass{m}{n}{l}$, $l\ge1$.

\begin{lemma}[strictness of the initial-basis level]
\label{lem:base-separation}
If $a>b$, then for all step-function parameters,
\[
 \Hclass{a}{n}{l}\nsubseteq\Hclass{b}{n'}{l'}.
\]
\end{lemma}

\begin{proof}
The source initial basis contains $e_{b+1}$.  By Lemma~\ref{lem:majorants}, every
target function has a pointwise majorant in
$\operatorname{Cl}_\circ(B_b)$.  We show that no such majorant exists for
$e_{b+1}$.

If $b=0$, the witness $e_1(x,y)=x+y$ cannot be majorized by a constant or
shifted projection $x_i+c$: put $x=y>N>c$.  If $b=1$, the quadratic
$e_2(x)=x^2+2$ eventually exceeds every affine unary term.  If $b\ge2$, put
$F=e_b$.  Lemma~\ref{lem:majorants} bounds every unary composition term
by $F^{[c]}(x+c)$ for a fixed $c$.  But
\[
 e_{b+1}(x)=F^{[x]}(2).
\]
Since $F(u)\ge2u$, choose $x>c$ with
$F^{[x-c]}(2)>x+c$.  Strict monotonicity gives
\[
 e_{b+1}(x)=F^{[c]}(F^{[x-c]}(2))>F^{[c]}(x+c).
\]
Thus $e_{b+1}$ is in the source but not the target.
\end{proof}

\subsection*{Canonical inputs and zones}
Growth estimates alone cannot distinguish equal-row strides because their
descent depths have the same asymptotic scale.  The idea is instead to follow
one dependency from the output through a concrete evaluation.  On canonical
points $Y_t$, a $\rho_{n,q}$-step moves from layer $t$ to layer $t-q$; we
enlarge each layer to a polynomial-width zone so that lower-basis operations
may perturb values without changing the layer they represent.  Along a
selected dependency chain, basis operations only widen the current zone,
whereas each recursion-address step decreases its index by exactly $q$.
Since distinct high zones are separated, reaching $Y_{t-d}$ from $Y_t$ forces
$d$ to be a multiple of $q$.

For example, when $n=2$ and $Y_t=e_2^{[t]}(0)$,
\[
 \rho_{2,l}(Y_t)=Y_{t-l}\qquad(t\ge l).
\]
A stride-$3$ derivation can therefore accumulate only multiples of $3$ in
canonical-layer displacement, whereas $\rho_{2,2}$ has persistent displacement
$2$.  The purpose of the zones below is to make this elementary picture stable
under arbitrary composition and nested step recursion.

Fix $(m,n)$ in the trace band with $m<n$, and fix $q\ge1$.  Put
\[
 G=g_n=e_n,\qquad Y_t=G^{[t]}(0),
\]
and set
\[
 F(x)=2x+2\quad(n=2),\qquad F=e_{n-1}\quad(n\ge3).
\]
Throughout this subsection, a polynomial means a polynomial with
nonnegative integer coefficients, interpreted as a map $\Nat\to\Nat$.
First, for integers $s,p\in\Nat$, define the width-$p$ bounds
\[
 U_s[p]=F^{[p]}(Y_s+p),
 \qquad
 L_s[p]=\min\{u\in\Nat:F^{[p]}(u+p)\ge Y_s\},
\]
and the width-$p$ zone
\[
 \mathcal Z_s[p]
 =
 \{u\in\Nat:L_s[p]\le u\le U_s[p]\}.
\]
For a polynomial $P$ and $s,t\in\Nat$, abbreviate
\[
 U_{s,P}(t):=U_s[P(t)],\qquad
 L_{s,P}(t):=L_s[P(t)],\qquad
 \mathcal Z_{s,P}(t):=\mathcal Z_s[P(t)].
\]
Equivalently,
\[
 u\in\mathcal Z_s[p]
 \quad\Longleftrightarrow\quad
 u\le U_s[p]
 \ \text{and}\ 
 Y_s\le F^{[p]}(u+p).
\]
In particular, $Y_s\in\mathcal Z_{s,0}(t)=\mathcal Z_s[0]$.

\begin{lemma}[quantitative canonical gap]
\label{lem:canonical-gap}
For every polynomial $P$ and every fixed integer offset $a$, whenever the displayed indices are nonnegative,
\begin{equation}
\label{eq:canonical-gap}
 F^{[P(t)]}(Y_{t-a-1}+P(t))<Y_{t-a}
\end{equation}
for all sufficiently large $t$.  The conclusion remains valid after
replacing $P$ by any fixed sum, positive integer multiple, composition with a
fixed power, or larger polynomial.
\end{lemma}

\begin{proof}
Put $s=t-a-1$.  Since $G\ge e_2$,
\[
 Y_j=G^{[j]}(0)\ge e_2^{[j]}(0)\ge2^{2^{j-1}}
 \qquad(j\ge1).
\]
Consequently $Y_{s-1}>P(t)+2$ and $Y_s>P(t)+2$ for all sufficiently large
$t$.

Suppose first that $n\ge3$.  Then $F=e_{n-1}\ge e_2$ and the generator identity
gives
\[
 Y_s=G(Y_{s-1})=F^{[Y_{s-1}]}(2),
 \qquad
 Y_{s+1}=F^{[Y_s]}(2).
\]
Because $F(Y_s)\ge Y_s^2+2>Y_s+P(t)$, monotonicity yields
\[
 \begin{aligned}
 F^{[P(t)]}(Y_s+P(t))
 &<F^{[P(t)+1]}(Y_s)\\
 &=F^{[Y_{s-1}+P(t)+1]}(2).
 \end{aligned}
\]
Also
\[
 Y_s-Y_{s-1}
 \ge Y_{s-1}^2-Y_{s-1}+2>P(t)+1
\]
eventually, so $Y_{s-1}+P(t)+1<Y_s$.  Strict increase of $F$ therefore gives
\[
 F^{[Y_{s-1}+P(t)+1]}(2)
 <F^{[Y_s]}(2)=Y_{s+1},
\]
which is \eqref{eq:canonical-gap}.

If $n=2$, then $F(x)=2x+2$ and
\[
 F^{[r]}(x)=2^r(x+2)-2,
 \qquad
 Y_{s+1}=Y_s^2+2.
\]
The lower bound on $Y_s$ implies
$P(t)+2<\log_2Y_s$ eventually.  Since also $P(t)+2<Y_s$,
\[
 \log_2(Y_s+P(t)+2)<\log_2Y_s+1.
\]
Hence
\[
 P(t)+\log_2(Y_s+P(t)+2)<2\log_2Y_s,
\]
and exponentiation gives
\[
 F^{[P(t)]}(Y_s+P(t))
 =2^{P(t)}(Y_s+P(t)+2)-2
 <Y_s^2+2=Y_{s+1}.
\]
The listed modifications of $P$ are still fixed polynomials, so the same
argument applies.
\end{proof}

\begin{definition}[selected dependency chains]
\label{def:selected-evaluation-chain}
Let $\delta$ be a derivation over $B_m$ using $\rho_{n,q}$, evaluated at
$\bar x$.  A \emph{selected dependency chain} is a finite sequence
\[
 u_0,u_1,\ldots,u_r=\sem{\delta}(\bar x)
\]
together with concrete value occurrences in that evaluation.  The occurrence
of $u_r$ is the root output, the origin $u_0$ is either one of the formal
input occurrences $x_i$ or a fixed-numeral occurrence, and the intervening
occurrences form one inductively selected dependency path.  At a composition
node, when the selected head chain begins at its $i$-th formal input, it is
prefixed by a selected chain through the $i$-th argument subevaluation.  At a
step-recursion node, a selected transition chain is attached through exactly
one of its parameter, current-address, or previous-state inputs: a parameter
uses the corresponding outer input occurrence, an address uses the appropriate
segment of the outer descent schedule, and a previous state uses a selected
chain for the preceding recursive stage.  No concrete occurrence is reused.
We suppress the occurrence labels and display only the numerical values.
Every step is of one of the following types:
\begin{enumerate}
\item an equality step, $u_{j+1}=u_j$, contributed by a projection;
\item a \emph{basis step} contributed by one concrete occurrence of an
initial function, of one of the following forms:
\[
 u_{j+1}=u_j+1,
 \qquad
 u_{j+1}=x+y\ \text{with }u_j=\max\{x,y\},
 \qquad
 u_{j+1}=e_k(u_j)\quad(2\le k\le m);
\]
In the addition clause, $x$ and $y$ are the numerical values at the two
input occurrences of that concrete addition node.  Only the forms whose
generators occur in $B_m$ are available;
\item a \emph{descent step},
$u_{j+1}=\rho_{n,q}(u_j)$.
\end{enumerate}
A zero occurrence may supply the numeral origin $0$; fixed numeral
subderivations may equivalently be treated as fixed numeral origins.  The
chain is therefore an occurrence-respecting, existentially selected numerical
dependency path, not the full evaluation tree.
\end{definition}

\begin{lemma}[selected-chain existence and polynomial length]
\label{lem:selected-chain-interface}
Let $\delta$ be a fixed derivation over $B_m$ using $\rho_{n,q}$.  There are constants
$A_\delta,C_\delta^{\mathrm{bas}},r_\delta\ge1$ and a finite set
$\mathcal N_\delta\subseteq\Nat$ such that every concrete evaluation in
which every recursion address has depth at most $K$ admits a selected
evaluation chain of length at most
\[
 L_\delta(K):=A_\delta(K+1)^{r_\delta}.
\]
The chain may be chosen with the occurrence labels specified in
Definition~\ref{def:selected-evaluation-chain}.  If it has a numeral origin,
that numeral belongs to $\mathcal N_\delta$.  Every basis step from value $u$ to value $v$ satisfies
\begin{equation}
\label{eq:selected-basis-envelope}
 u\le v\le
 F^{[C_\delta^{\mathrm{bas}}]}(u+C_\delta^{\mathrm{bas}}).
\end{equation}
At input $Y_t=G^{[t]}(0)$, choose $B_\delta\ge1$ so that every
recursion address has depth at most $B_\delta(t+1)$, as supplied by
Lemma~\ref{lem:derivation-internal-envelope}.  Define
\[
 P_\delta(t):=L_\delta\bigl(B_\delta(t+1)\bigr).
\]
Then, for every such evaluation at input $Y_t$, a selected chain may be
chosen with length at most $P_\delta(t)$.
\end{lemma}

\begin{proof}
We prove existence and a recursive length bound by induction on construction
rank.  Since a fixed derivation contains only finitely many initial-function
nodes, one constant $C_\delta^{\mathrm{bas}}$ and one finite set of possible
numeral origins $\mathcal N_\delta$ can be chosen for all of them.
At an initial-function occurrence, use exactly the rule listed in
Definition~\ref{def:selected-evaluation-chain}: select the unique input of
successor or of a unary generator, and select a maximum input $u$ of
addition, for which $u\le x+y\le2u$.  A projection gives an equality step,
while zero starts a chain at the fixed numeral $0$.  The required envelope
can be checked directly.  If $n=2$, then $m\le1$ and
$F(x)=2x+2$; thus only successor and possibly addition occur, and
\[
 u+1\le F(u+1),\qquad 2u\le F(u+1).
\]
If $n\ge3$, then $F=e_{n-1}$.  Every unary generator in $B_m$ is bounded
by $F$ because $m\le n-1$ and
Lemma~\ref{lem:generator-identities} compares the generator indices; also
$u+1\le F(u)$ and $2u\le F(u)$ by
Lemma~\ref{lem:generator-identities}.  Enlarging one fixed constant absorbs
all initial nodes of the fixed derivation.  This proves the initial case with
constant length.

Suppose
\[
 \delta=\operatorname{Comp}(\eta;\delta_1,\ldots,\delta_r).
\]
Construct a selected chain in the concrete head evaluation.  If its origin is
a fixed numeral, use that chain.  If its origin is the $i$-th formal input of
the head, that value is the output of $\delta_i$; concatenate a selected
chain for $\delta_i$ with the head chain.  Hence one may take
\[
 L_\delta(K)
 \le L_\eta(K)+\max_iL_{\delta_i}(K).
\]
Only one argument branch is followed.

Now let
\[
 \delta=\operatorname{SR}_{\rho_{n,q}}(\gamma,\eta,\beta)
\]
and evaluate it at $(\bar p,z)$.  Put $d=D_{n,q}(z)$,
\[
 z_j=\rho_{n,q}^{[d-j]}(z)\quad(0\le j\le d),
\]
and
\[
 s_0=\sem{\gamma}(\bar p),\qquad
 s_{j+1}=\sem{\eta}(\bar p,z_j,s_j)\quad(0\le j<d).
\]
We construct by induction on $j$ a selected chain ending at $s_j$.  For
$j=0$, use a selected chain for the base derivation.  Given the chain for
$s_j$, select a chain in the transition evaluation producing $s_{j+1}$.
If its origin is a parameter input, it already begins at the corresponding
formal parameter of the outer derivation.  If it is the address input $z_j$,
prepend the full segment
\[
 z=z_d,\ z_{d-1},\ldots,z_j,
\]
with one descent step between consecutive values.  If it is the previous
state input $s_j$, concatenate the already constructed chain for $s_j$.  If
it is a fixed numeral, retain that numeral-origin chain.  These alternatives
are exhaustive because the transition derivation has ordered inputs
$(\bar p,z_j,s_j)$.

Since $d\le K$, the resulting chain has length at most
\[
 L_\gamma(K)+K L_\eta(K)+K+1.
\]
The declared bound does not create a dependency step; it only certifies the
values $s_j$.  The displayed recurrences imply by induction that
$L_\delta(K)\le A_\delta(K+1)^{r_\delta}$ for fixed constants.  All basis
constants occurring in copied subderivations are absorbed by their finite
maximum.

The final assertion follows from
Lemma~\ref{lem:derivation-internal-envelope} and the definition of
$P_\delta$ above.
\end{proof}

\begin{lemma}[zone uniqueness]
\label{lem:zone-unique}
For every fixed $d$ and polynomial $P$, and for all large $t$,
\begin{equation}
\label{eq:zone-unique}
 Y_{t-d}\in\mathcal Z_{s,P}(t)
 \quad\Longrightarrow\quad
 s=t-d.
\end{equation}
\end{lemma}

\begin{proof}
Take $t$ large enough that all indices below are nonnegative.  If
$s\le t-d-1$, then
\[
 Y_{t-d}\le U_{s,P}(t)
 \le F^{[P(t)]}(Y_{t-d-1}+P(t))<Y_{t-d}
\]
by \eqref{eq:canonical-gap} with offset $a=d$.  If
$s\ge t-d+1$, the lower-zone condition gives
\[
 Y_{t-d+1}\le Y_s
 \le F^{[P(t)]}(Y_{t-d}+P(t))<Y_{t-d+1},
\]
where the last inequality is \eqref{eq:canonical-gap} with offset
$a=d-1$.  This includes $d=0$, because fixed integer offsets, including
$a=-1$, are allowed once $t$ is sufficiently large.
\end{proof}

\begin{lemma}[basis-step stability]
\label{lem:basis-function-edge}
If a basis step satisfies $u\le v\le F^{[C]}(u+C)$ and
$u\in\mathcal Z_{s,P}(t)$, then $v\in\mathcal Z_{s,P+C}(t)$.
\end{lemma}

\begin{proof}
Put $p=P(t)$.  Since every iterate of the strictly increasing
integer-valued map $F$ satisfies $F^{[p]}(x+C)\ge F^{[p]}(x)+C$, the upper
zone inequality for $u$ gives
\[
 u+C\le F^{[p]}(Y_s+p)+C
      \le F^{[p]}(Y_s+p+C).
\]
Therefore
\[
 v\le F^{[C]}(u+C)
   \le F^{[p+C]}(Y_s+p+C)=U_{s,P+C}(t).
\]
For the lower condition, $v\ge u$ implies
\[
 Y_s\le F^{[p]}(u+p)
     \le F^{[p+C]}(v+p+C).
\]
Thus $v\in\mathcal Z_{s,P+C}(t)$.
\end{proof}

\begin{lemma}[descent-step displacement]
\label{lem:descent-edge}
Write $H_p(x)=F^{[p]}(x+p)$.  For every $p\ge0$,
\begin{equation}
\label{eq:shifted-commutation}
 G^{[q]}(H_p(x))\ge H_p(G^{[q]}(x)).
\end{equation}
Consequently, for every $w$,
\begin{equation}
\label{eq:inverse-shifted-commutation}
 \rho_{n,q}(H_p(w))\le H_p(\rho_{n,q}(w)).
\end{equation}
In particular, for every $s\ge q$ and every
$u\in\mathcal Z_{s,P}(t)$,
\begin{equation}
\label{eq:zone-step-function}
 \rho_{n,q}(u)\in\mathcal Z_{s-q,P}(t).
\end{equation}
\end{lemma}

\begin{proof}
We first prove
\begin{equation}
\label{eq:GF-commutation}
 G(F(x))\ge F(G(x)).
\end{equation}
If $n=2$, then $F(x)=2x+2$ and $G(x)=x^2+2$, and direct expansion gives
\[
 G(F(x))=4x^2+8x+6\ge2x^2+6=F(G(x)).
\]
If $n\ge3$, then $G(x)=F^{[x]}(2)$ and $F(x)\ge x+1$, whence
\[
 G(F(x))=F^{[F(x)]}(2)
 \ge F^{[x+1]}(2)=F(G(x)).
\]
Induction on $p$ gives
$G(F^{[p]}(z))\ge F^{[p]}(G(z))$.  Since $G$ is a strictly increasing
integer-valued function, $G(x+p)\ge G(x)+p$.  Therefore
\[
\begin{aligned}
 G(H_p(x))
 &=G(F^{[p]}(x+p))\\
 &\ge F^{[p]}(G(x+p))\\
 &\ge F^{[p]}(G(x)+p)=H_p(G(x)).
\end{aligned}
\]
Iteration in $q$ proves \eqref{eq:shifted-commutation}.

Apply \eqref{eq:shifted-commutation} at $x=\rho_{n,q}(w)$.  By the
definition of the generalized inverse,
$G^{[q]}(\rho_{n,q}(w))\ge w$.  Since $H_p$ is nondecreasing,
\[
 G^{[q]}\bigl(H_p(\rho_{n,q}(w))\bigr)
 \ge H_p\bigl(G^{[q]}(\rho_{n,q}(w))\bigr)
 \ge H_p(w).
\]
Thus $H_p(\rho_{n,q}(w))$ is an admissible threshold witness for
$\rho_{n,q}(H_p(w))$, proving
\eqref{eq:inverse-shifted-commutation}.

Now suppose $u\in\mathcal Z_{s,P}(t)$ and $s\ge q$, and put $p=P(t)$.
The upper zone inequality and monotonicity of $\rho_{n,q}$ give
\[
\begin{aligned}
 \rho_{n,q}(u)
 &\le \rho_{n,q}(H_p(Y_s))\\
 &\le H_p(\rho_{n,q}(Y_s))
  =H_p(Y_{s-q}).
\end{aligned}
\]
For the lower zone inequality, $Y_s\le H_p(u)$ implies
\[
\begin{aligned}
 Y_{s-q}
 &=\rho_{n,q}(Y_s)\\
 &\le\rho_{n,q}(H_p(u))\\
 &\le H_p(\rho_{n,q}(u)).
\end{aligned}
\]
These are exactly the two defining inequalities for
$\rho_{n,q}(u)\in\mathcal Z_{s-q,P}(t)$.
\end{proof}

\begin{theorem}[zone-index invariant along selected dependency chains]
\label{thm:selected-chain-valuation}
Let $u_0,\ldots,u_r$ be a selected dependency chain chosen as in
Lemma~\ref{lem:selected-chain-interface}, and suppose
$u_i\in\mathcal Z_{s,P}(t)$.  For $h\ge i$, let $j_h$ be the number of
descent steps among the transitions from $u_i$ to $u_h$, and let $b_h$ be the
number of basis steps among them.  Put
\[
 p_h:=P(t)+b_hC_\delta^{\mathrm{bas}}.
\]
As long as every one of those descent steps starts with current zone index at
least $q$, the following invariant holds:
\[
 u_h\in\mathcal Z_{s-qj_h}[p_h].
\]
Thus equality steps preserve both the index and the integer width, basis steps
preserve the index and enlarge the width by at most
$C_\delta^{\mathrm{bas}}$, and descent steps change the index exactly by
$-q$.  Since the chain has polynomial length, one polynomial
$P^\star$, depending only on $P$ and the fixed derivation $\delta$, absorbs
all basis widenings.  If the current index falls below $q$, equality and basis
steps remain covered, but the invariant provides no further descent step.
\end{theorem}

\begin{proof}
We induct on $h-i$.  At $h=i$, one has $j_i=b_i=0$, so the assertion is the
hypothesis $u_i\in\mathcal Z_{s,P}(t)$.  Suppose it holds at $h$.

An equality step leaves the current value unchanged.  For a basis step,
Lemma~\ref{lem:basis-function-edge}, applied at the fixed input $t$, preserves
the current zone index and replaces the integer width $p_h$ by at most
$p_h+C_\delta^{\mathrm{bas}}=p_{h+1}$.  Enlarging the width preserves zone
membership.  For a descent step whose current index $s-qj_h$ is at least
$q$, Lemma~\ref{lem:descent-edge} changes the zone index to
$s-q(j_h+1)$ without enlarging the width.  These are all possible step
types, so the invariant follows.

By Lemma~\ref{lem:selected-chain-interface}, the whole chain has length at
most the fixed polynomial $P_\delta(t)$.  Hence the fixed polynomial
\[
 P^\star(X):=P(X)+C_\delta^{\mathrm{bas}}P_\delta(X)
\]
satisfies $p_h\le P^\star(t)$ for every point of the chosen chain.
Enlarging an integer zone width preserves membership by the defining
inequalities and monotonicity of $F$, so
$\mathcal Z_{s-qj_h}[p_h]\subseteq
\mathcal Z_{s-qj_h,P^\star}(t)$ uniformly at input $Y_t$.  The
descent clause of Lemma~\ref{lem:descent-edge} requires a starting index at
least $q$, which explains the final qualification.
\end{proof}

\begin{lemma}[constant-origin exclusion]
\label{lem:constant-source-exclusion}
Let $\delta$ be a fixed unary derivation, and let $P_\delta$ and
$\mathcal N_\delta$ be supplied by
Lemma~\ref{lem:selected-chain-interface}.  A selected chain beginning at a
fixed numeral and using at most $P_\delta(t)$ basis
steps satisfying \eqref{eq:selected-basis-envelope}, interleaved with
nonincreasing descent steps, has final value smaller than $Y_{t-d}$ for every
fixed $d$ and all sufficiently large $t$.
\end{lemma}

\begin{proof}
Delete descent steps and compose the basis-step upper envelopes.  From a fixed
numeral the result is bounded by
$F^{[C(P_\delta(t)+1)]}(C(P_\delta(t)+1))$ for a fixed $C$.  Choose one $s_0$ such that every numeral in $\mathcal N_\delta$ is at most
$Y_{s_0}$.  For large $t$ one has
$s_0\le t-d-1$, so monotonicity followed by the canonical gap estimate
\eqref{eq:canonical-gap} puts this bound below $Y_{t-d}$.
\end{proof}

\begin{lemma}[low-zone barrier]
\label{lem:low-zone}
Let $P,Q$ be polynomials and let $d\ge0$ be fixed.
There is $t_0$ such that for all $t\ge t_0$ the following holds.  If
$u\le U_{s,P}(t)$ for some $s<q$, then after any sequence of at most $Q(t)$
basis steps $v\mapsto w$ satisfying
$v\le w\le F^{[C]}(v+C)$ for a fixed $C$, interleaved in any order with any
number of descent steps $w\mapsto\rho_{n,q}(w)$, the resulting value is
strictly smaller than $Y_{t-d}$.
\end{lemma}

\begin{proof}
From $s<q$ one has
\[
 u\le U_{s,P}(t)\le F^{[P(t)]}(Y_{q-1}+P(t)).
\]
Delete every descent step for an upper bound.  Applying at most $Q(t)$ basis
envelopes bounds the final value by
\[
 F^{[Q^\sharp(t)]}\bigl(Y_{q-1}+Q^\sharp(t)\bigr),
\]
where $Q^\sharp$ is a fixed polynomial majorant.  For large $t$,
$q-1\le t-d-1$, so monotonicity and \eqref{eq:canonical-gap} give a value
strictly below $Y_{t-d}$.
\end{proof}

\begin{lemma}[divisibility from selected chains]
\label{lem:canonical-path}
Let $f$ be a unary member of $\Hclass{m}{n}{q}$.  If a fixed $d\ge0$ satisfies
\[
 f(Y_t)=Y_{t-d}
\]
for every sufficiently large $t$, then $q\mid d$.
\end{lemma}

\begin{proof}
We combine the preceding ingredients: a selected chain reaching $Y_{t-d}$
cannot start at a fixed numeral, lower-basis steps cannot change its canonical
zone index, and every admissible descent changes that index by exactly $q$.
Choose a finite unary derivation $\delta$ of $f$.  At $Y_t$,
Lemma~\ref{lem:selected-chain-interface} gives a selected chain of
length at most $P_\delta(t)$.  By
Lemma~\ref{lem:constant-source-exclusion} and the finiteness of
$\mathcal N_\delta$, there is $t_0$ such that, for every $t\ge t_0$, no
selected chain ending at $Y_{t-d}$ can have a numeral origin.  Since the
derivation is unary, every such chain starts at the unique formal input value
$Y_t\in\mathcal Z_{t,0}(t)$.

Read the chain forward.  Equality and basis steps preserve the current zone
index with polynomial widening, and every descent step at an index at least
$q$ subtracts exactly $q$, by
Theorem~\ref{thm:selected-chain-valuation}.  If a descent step is encountered
when the current index is below $q$, Lemma~\ref{lem:low-zone}, applied to the
remaining polynomially many basis steps, makes the final value smaller than
$Y_{t-d}$, a contradiction.  Thus every descent step occurs at an index at
least $q$.

If the chain contains $j$ descent steps, one fixed polynomial $P^\star$
absorbs all basis widenings and
\[
 Y_{t-d}\in\mathcal Z_{t-jq,P^\star}(t).
\]
Lemma~\ref{lem:zone-unique} forces $t-d=t-jq$.  Hence $d=jq$ and $q\mid d$.
\end{proof}

\begin{lemma}[horizontal divisibility below initial-basis level $n$]
\label{lem:horizontal-valuation}
Let $n\ge2$, $m<n$, and $p,q\ge1$.  Then
\[
 \rho_{n,p}\in\Hclass{m}{n}{q}
 \quad\Longleftrightarrow\quad q\mid p.
\]
\end{lemma}

\begin{proof}
If $p=kq$, Lemma~\ref{lem:step-function-arithmetic} gives
$\rho_{n,p}=\rho_{n,q}^{[k]}$, so membership follows by composition.
For necessity, put $Y_t=g_n^{[t]}(0)$.  For $t\ge p$,
\[
 \rho_{n,p}(Y_t)=Y_{t-p}.
\]
Lemma~\ref{lem:canonical-path}, applied with $d=p$, gives $q\mid p$.
\end{proof}

\section{Vertical separation}

\begin{lemma}[vertical bit exclusion]
\label{lem:vertical-exclusion}
If $2\le n<n'$, $l,l'\ge1$, and $m\le n'$, then
\[
 \operatorname{bit}_{n,l}\notin\Hclass{m}{n'}{l'}.
\]
\end{lemma}

\begin{proof}
Use the notation of Lemma~\ref{lem:vertical-domination}.  For every
$0\le d<c_s$, the threshold law gives
\[
 D_{n,l}(\alpha^{[d]}(0))=d,
 \qquad
 D_{n,l}(\alpha^{[d]}(0)+1)=d+1.
\]
Moreover, $D_{n,l}(y_s)=c_s$ implies
$\alpha^{[d]}(0)<y_s$ for all $d<c_s$.  Thus the least significant
lower-depth bit changes at each of the $c_s$ distinct positions
$\alpha^{[d]}(0)$ below $y_s$, and hence
\[
 \Delta_{\operatorname{bit}_{n,l}}(y_s)\ge c_s.
\]
Because $m\le n'$, the target pair lies in the trace band.  If the bit
belonged to the target class, choose a derivation of it.
Theorem~\ref{thm:trace} would then give constants $C,r$, fixed once
that construction is fixed and therefore independent of $s$, such that
\[
 \Delta_{\operatorname{bit}_{n,l}}(y_s)
 \le
 \exp\!\left(
   C\bigl(1+D_{n',l'}(y_s+C)\bigr)^r
 \right).
\]
Because $n'\ge3$, the generator comparison gives
$\beta(x)\ge e_3(x)$ and hence $\beta(x)-x\to\infty$.  Therefore, for all
sufficiently large $s$, one has $y_s+C\le y_{s+1}$ and hence
$D_{n',l'}(y_s+C)\le s+1$.  The resulting upper bound is
$\exp(C'(s+2)^r)$ for a fixed $C'$, contradicting
Lemma~\ref{lem:vertical-domination}, which gives
$c_s>\exp(C'(s+2)^r)$ eventually.
\end{proof}

\section{Classification for \texorpdfstring{$n,n'\ge2$}{n,n' >= 2}}
\label{sec:classification-proof}

\begin{theorem}[complete classification for $n,n'\ge2$]
\label{thm:global-classification}
Let $a,b\in\Nat$, $n,n'\ge2$, and $l,l'\ge1$.
\begin{enumerate}
\item If $n>n'$, then
\[
 \Hclass{a}{n}{l}\subseteq\Hclass{b}{n'}{l'}
 \quad\Longleftrightarrow\quad a\le b.
\]
\item If $n=n'$, then
\[
 \Hclass{a}{n}{l}\subseteq\Hclass{b}{n}{l'}
 \quad\Longleftrightarrow\quad
 a\le b\ \text{and}\ (b\ge n\ \text{or}\ l'\mid l).
\]
\item If $n<n'$, then
\[
 \Hclass{a}{n}{l}\subseteq\Hclass{b}{n'}{l'}
 \quad\Longleftrightarrow\quad
 a\le b\ \text{and}\ b\ge n'+1.
\]
\end{enumerate}
\end{theorem}

\begin{proof}
Lemma~\ref{lem:base-separation} gives necessity of $a\le b$ in every
case.  If $n>n'$,
Lemmas~\ref{lem:higher-step-function-internal},~\ref{lem:vertical-padding}, and~\ref{lem:depth-transfer} transfer every source recursion to the target.
This proves the first equivalence.

At equal row, divisibility $l'\mid l$ makes the source descent a fixed
iterate of the target descent, so Lemma~\ref{lem:depth-transfer} gives
inclusion.  If $b\ge n$, Lemma~\ref{lem:horizontal-collapse} gives
inclusion for all iteration counts.  Conversely,
when $b\le n-1$, the source contains $\rho_{n,l}$ and Lemma~\ref{lem:horizontal-valuation} forces $l'\mid l$.  This proves the
second equivalence.

If $n<n'$ and $a\le b$, $b\ge n'+1$, ordinary calibration and
Theorem~\ref{thm:saturation} give
\[
 \Hclass{a}{n}{l}\subseteq\Eclass{a}
 \subseteq\Eclass{b}=\Hclass{b}{n'}{l'}.
\]
Conversely, suppose inclusion holds.  Lemma~\ref{lem:base-separation} gives $a\le b$.  If
$b\le n'$, then the source contains
$\operatorname{bit}_{n,l}$ by Lemma~\ref{lem:internal-controls}, whereas
Lemma~\ref{lem:vertical-exclusion} excludes that function from the target.
Hence $b\ge n'+1$.
\end{proof}

\begin{corollary}[traversal order below collapse]
\label{cor:traversal-order}
Fix $n\ge2$ and $m<n$.  The stride sector
$\{\Hclass{m}{n}{l}:l\ge1\}$ is anti-isomorphic to
$(\Nat_{>0},\mid)$.  Consequently it contains infinite strict descending
chains, infinite antichains, and an order-embedded copy of every finite
partial order.
\end{corollary}

\begin{proof}
Theorem~\ref{thm:global-classification} gives reverse divisibility.  Powers
of $2$ give a descending chain, distinct primes an antichain, and square-free
products of finitely many primes realize Boolean lattices, into which every
finite poset embeds.
\end{proof}

\begin{corollary}[asymptotic depth is not a complete invariant]
\label{cor:depth-not-complete}
Fix $n\ge2$ and $m<n$.  For any $l,l'\ge1$ the descent-depth functions
$D_{n,l}$ and $D_{n,l'}$ are $\Theta$-equivalent.  Nevertheless,
\[
 \Hclass{m}{n}{l}=\Hclass{m}{n}{l'}
 \quad\Longleftrightarrow\quad l=l',
\]
and the two classes are incomparable whenever neither $l\mid l'$ nor
$l'\mid l$.
\end{corollary}

\begin{proof}
Lemma~\ref{lem:step-function-arithmetic} gives
$D_{n,l}(y)=\lceil D_{n,1}(y)/l\rceil$, so all fixed-stride depths differ only
by constant factors.  The class statements are immediate from
Corollary~\ref{cor:traversal-order}.
\end{proof}

\begin{remark}
Corollary~\ref{cor:depth-not-complete} is the conceptual content of the stride
parameter: growth scale and asymptotic recursion depth can be held fixed while
the exact alignment of visited canonical layers still changes expressive
strength.
\end{remark}

\begin{corollary}[row-zero targets and comparisons with $n\ge2$]
\label{cor:row-zero-comparisons}
For every $n\in\Nat$,
\[
 \Hclass{a}{n}{l}\subseteq\Hclass{b}{0}{l'}
 \quad\Longleftrightarrow\quad a\le b.
\]
For every $n'\ge2$,
\[
 \Hclass{a}{0}{l}\subseteq\Hclass{b}{n'}{l'}
 \quad\Longleftrightarrow\quad
 a\le b\ \text{and}\ b\ge n'+1.
\]
\end{corollary}

\begin{proof}
The target-row-zero equivalence follows from ordinary calibration and Lemma~\ref{lem:base-separation}.  For the positive direction of the second equivalence, use
$\Hclass{a}{0}{l}=\Eclass{a}\subseteq\Eclass{b}$ and the saturation identity
$\Eclass{b}=\Hclass{b}{n'}{l'}$ when $b\ge n'+1$.

Conversely, Lemma~\ref{lem:base-separation} gives $a\le b$.  The parity function belongs to
$\Eclass{0}\subseteq\Eclass{a}=\Hclass{a}{0}{l}$.  If $b\le n'$, then $(b,n')$ lies in the trace band, so
Corollary~\ref{cor:finite-range-sparsity} excludes parity from
$\Hclass{b}{n'}{l'}$.  Hence inclusion forces $b\ge n'+1$.
\end{proof}

\section{The doubling row and the base-two bridge}
\label{sec:doubling-row}

The trace-band argument used for $n\ge2$ does not directly classify the
row generated by $g_1(x)=2x+1$.  We first settle its basis-zero stride sector
by a dyadic selected-chain argument.  The result is exact reverse divisibility,
so the exceptional row retains the same arithmetic order at its weakest basis
even though its canonical gaps are only exponential.

At initial basis $2$, doubling depth interacts unusually well with
$e_2(x)=x^2+2$.  The algebraic part of the section follows the chain
\[
 \text{normalized squaring}\ \Longrightarrow\ \text{comparison}
 \ \Longrightarrow\ \rho_{2,1}\text{ internal}
 \ \Longrightarrow\ \Hclass{2}{2}{*}\subsetneq\Hclass{2}{1}{*}.
\]
At basis $m\ge3$, an exact linear-depth ladder instead yields
\[
 \Hclass{m}{1}{*}=\Eclass{m}.
\]
We then turn separately to the binary-time size of the exceptional class
$\Hclass{2}{1}{*}$.

\subsection{Basis-zero stride classification}
\label{subsec:basis-zero-doubling}

Fix a target stride $q\ge1$, put $B=2^q$, and write
\[
 \rho(y)=\rho_{1,q}(y)=\left\lfloor\frac yB\right\rfloor,
 \qquad Y_t=2^t-1.
\]
For $t\ge q$ one has
\begin{equation}
\label{eq:canonical-dyadic-shift-main}
 \rho(Y_t)=Y_{t-q}.
\end{equation}

\begin{lemma}[basis-zero internal-value envelope]
\label{lem:basis-zero-internal-envelope}
Let $\delta$ be a fixed derivation over $B_0$ using $\rho_{1,q}$.  There is a constant $c_\delta$ such that, whenever every
input coordinate is at most $N$,
\[
 v\le N+c_\delta
 \qquad
 \bigl(v\in\operatorname{Val}_{\delta}(\bar x)\bigr).
\]
In particular, every recursion address occurring in the evaluation is at most
$N+c_\delta$.
\end{lemma}

\begin{proof}
Induct on the fixed derivation.  For an initial function the claim is
immediate.  At basis zero, Lemma~\ref{lem:majorants} says that every
composition majorant is a constant or a shifted projection, hence is at most
$N+c$ on inputs bounded by $N$.

For a composition, apply the induction hypotheses to the argument
derivations.  Their outputs, and all values created while computing them, are
at most $N+c$.  Apply the induction hypothesis for the head derivation with
input cap $N+c$ and absorb the finitely many constants.

Suppose that $\delta$ is a bounded step-recursion derivation.  Every outer
address is at most the recursion input and hence at most $N$.  The declared
bound has, by Lemma~\ref{lem:majorants}, an earlier basis-zero composition
majorant, so every recursive state is at most $N+c_0$ for a fixed $c_0$.
The base derivation is evaluated on inputs bounded by $N$, and every transition
derivation is evaluated on parameters and an address bounded by $N$ together
with a previous state bounded by $N+c_0$.  Applying the induction hypotheses to the base and transition derivations, and
then taking the maximum of the finitely many resulting constants, gives the
required $N+c_\delta$ envelope uniformly over all recursion stages.  The bound
derivation is used only through its earlier composition majorant above; it is
not part of the operational evaluation.
\end{proof}

\begin{lemma}[dyadic selected-chain extraction]
\label{lem:dyadic-chain-main}
Let $\delta$ be a fixed unary derivation over $B_0$ using $\rho_{1,q}$.
There are a polynomial $P_\delta$ and a finite set of numerals $C_\delta$
such that, for every $t$, the output occurrence in the evaluation of $\delta$
at $Y_t$ admits a dependency chain $u_0,\ldots,u_s$ with
$u_s=\sem{\delta}(Y_t)$, where $u_0$ is either the input occurrence $Y_t$ or
a numeral in $C_\delta$, every edge has one of the forms
\[
 u_{i+1}=u_i,\qquad u_{i+1}=u_i+1,\qquad
 u_{i+1}=\left\lfloor\frac{u_i}{2^q}\right\rfloor,
\]
and $s\le P_\delta(t)$.
\end{lemma}

\begin{proof}
Induct on a finite derivation of $\delta$.  Initial functions are immediate.
At composition, follow one selected ancestry through the head evaluation and,
when it reaches a head argument, splice in the selected chain of the
corresponding argument evaluation.  At a step-recursion node, follow the
selected ancestry through the last transition.  If it reaches a parameter or
a lower-address occurrence, splice in the corresponding parameter chain or
the finite descent segment; if it reaches the previous-state occurrence,
repeat one stage lower.  The process terminates at a parameter, an address, or
the base evaluation.

For the length bound, Lemma~\ref{lem:basis-zero-internal-envelope} bounds
every value and recursion address created at input $Y_t$ by
$Y_t+c_\delta$.  Hence
$D_{1,q}(a)=\lceil\log_2(a+1)/q\rceil=O_\delta(t+1)$ uniformly for every
address $a$ occurring in the evaluation, so every descent schedule has
$O_\delta(t+1)$ stages.  If $L_\eta(K)$ bounds a selected chain in a transition
subderivation when all descent depths are at most $K$, a recursion node obeys
\[
 L_\delta(K)\le L_\gamma(K)+K L_\eta(K)+O(K).
\]
Composition takes only finite sums and maxima.  Induction on the fixed
derivation gives the required polynomial.
\end{proof}

For a polynomial $P$ put
\[
 \mathcal Z_{s,P}(t)=\{u\in\Nat:|u-Y_s|\le P(t)\}.
\]
Adjacent canonical values satisfy $Y_{r+1}-Y_r=2^r$, so a polynomial-width
neighborhood of $Y_{t-d}$ contains no other canonical value for all large
$t$.  Moreover, if $u\in\mathcal Z_{s,P}(t)$, then successor keeps $u$ in a
polynomially wider neighborhood of $Y_s$, while for $s\ge q$,
\[
 \left\lfloor\frac u{2^q}\right\rfloor
 \in\mathcal Z_{s-q,Q}(t)
\]
for a polynomial $Q$ depending only on $P$ and $q$.  A selected chain that
starts at a fixed numeral, or whose propagated layer index falls below $q$,
can reach only polynomially bounded values after polynomially many remaining
successor steps and therefore cannot end at $Y_{t-d}$ for all large $t$.

\begin{lemma}[dyadic divisibility obstruction]
\label{lem:dyadic-divisibility-main}
Let $f\in\Hclass{0}{1}{q}$ be unary.  If a fixed $d\ge0$ satisfies
\[
 f(Y_t)=Y_{t-d}
\]
for all sufficiently large $t$, then $q\mid d$.
\end{lemma}

\begin{proof}
Choose the chain from Lemma~\ref{lem:dyadic-chain-main}.  The preceding
barrier excludes a fixed numeral as its origin, so the chain starts at
$Y_t\in\mathcal Z_{t,0}(t)$.  Equality edges preserve the layer index,
successor edges enlarge only the polynomial zone width, and every descent
edge lowers the layer index by exactly $q$ by
\eqref{eq:canonical-dyadic-shift-main}.  If there are $j$ descent edges,
zone propagation gives
\[
 Y_{t-d}\in\mathcal Z_{t-jq,Q}(t)
\]
for a fixed polynomial $Q$.  Exponential separation of adjacent canonical
values forces $t-d=t-jq$ for all sufficiently large $t$, hence $d=jq$.
\end{proof}

\begin{theorem}[complete basis-zero doubling classification]
\label{thm:basis-zero-doubling-main}
For all $p,q\ge1$,
\[
 \boxed{\Hclass{0}{1}{p}\subseteq\Hclass{0}{1}{q}
 \quad\Longleftrightarrow\quad q\mid p.}
\]
\end{theorem}

\begin{proof}
If $p=kq$, Lemma~\ref{lem:step-function-arithmetic} gives
$\rho_{1,p}=\rho_{1,q}^{[k]}$.  The source descent is internal to the target
and $D_{1,q}\ge D_{1,p}$, so exact-depth transfer,
Lemma~\ref{lem:depth-transfer}, simulates every source recursion in the
target.

Conversely, suppose the class inclusion holds.  The source contains its
defining descent $\rho_{1,p}$ by Lemma~\ref{lem:internal-controls}; hence the
target contains it as well.  On canonical inputs,
$\rho_{1,p}(Y_t)=Y_{t-p}$.  Lemma~\ref{lem:dyadic-divisibility-main} yields
$q\mid p$.
\end{proof}

\begin{remark}
Theorem~\ref{thm:basis-zero-doubling-main} removes the only unresolved
equal-row stride family.  Thus the stride order in the doubling row is dual
divisibility at $m=0$ and collapses completely for every $m\ge1$.
\end{remark}

\subsection{A logspace interpretation at basis zero}
\label{subsec:basis-zero-logspace}

The preceding divisibility order is not only an internal feature of the
hierarchy.  It already occurs inside a standard complexity class.  Let
$\mathsf{FL}$ denote the class of total numerical functions computed by
deterministic logspace transducers on standard binary inputs and outputs.
Fix $l\ge1$, put $B=2^l$, and write
$\rho=\rho_{1,l}$, so $\rho(y)=\lfloor y/B\rfloor$.

For a root input tuple $\bar x=(x_1,\ldots,x_k)$, use two kinds of symbolic
descriptors,
\[
 \mathbf C(c),\qquad \mathbf S(i,a,d),
\]
with values
\[
 \operatorname{val}_{\bar x}(\mathbf C(c))=c,\qquad
 \operatorname{val}_{\bar x}(\mathbf S(i,a,d))
 =\left\lfloor\frac{x_i}{B^a}\right\rfloor+d.
\]
They are closed under the basis-zero operations.  Successor increments the
constant or offset field, and descent is exact: if
\[
 r_{i,a}=\left\lfloor\frac{x_i}{B^a}\right\rfloor\bmod B,
\]
then
\begin{equation}
\label{eq:descriptor-descent-main}
 \rho(\mathbf C(c))=\mathbf C(\lfloor c/B\rfloor),\qquad
 \rho(\mathbf S(i,a,d))
 =\mathbf S\!\left(i,a+1,\left\lfloor\frac{r_{i,a}+d}{B}\right\rfloor\right).
\end{equation}
Because $l$ is fixed, $r_{i,a}$ is read from one fixed-size block of the
read-only input.  Write
\[
 \lambda(x):=\lceil\log_2(x+1)\rceil
\]
for binary length, so in particular $\lambda(0)=0$.  A source descriptor is
zero exactly when
\[
 d=0\quad\text{and}\quad la\ge\lambda(x_i),
\]
since $\lfloor x_i/2^{la}\rfloor=0$ exactly under this condition.

\begin{theorem}[symbolic logspace evaluation]
\label{thm:basis-zero-symbolic-logspace}
For every fixed derivation $\delta$ over $B_0$ using $\rho_{1,l}$, the
function $\sem{\delta}$ is computable by a deterministic logspace transducer
in polynomial time.  More precisely, on root inputs of total binary length
$N$, the evaluation can be represented by a descriptor whose numerical fields
are bounded by $N^{c_\delta}$ for a fixed constant $c_\delta$.
\end{theorem}

\begin{proof}
Assign each root input $x_i$ the descriptor $\mathbf S(i,0,0)$.  More generally,
we prove by induction on a subderivation $\eta$ the following invariant: for
every descriptor environment representing the numerical inputs of $\eta$, the
evaluator returns a descriptor representing $\sem{\eta}$ on those inputs, uses
$O_\eta(\log N)$ work space and polynomial time, and keeps every numerical
descriptor field polynomially bounded in $N$.

Evaluate the fixed derivation relative to such an environment.  Zero,
projections, successor, and composition use the descriptor rules above.
Consider a step-recursion node with parameter
descriptors $\bar P$ and recursion-coordinate descriptor $Y$.  By the
induction hypotheses for the argument derivations, the fields of $Y$ are
polynomially bounded in $N$; since a descriptor denotes either such a constant
or a shifted input plus such an offset, its represented value has
$O_\delta(N)$ bits.  Repeatedly apply
\eqref{eq:descriptor-descent-main} until $Y$ becomes zero and count the
number $d$ of descent steps.  Thus $d=O_\delta(N)$.  Evaluate the base derivation on
$\bar P$.  To reconstruct the value from the base toward the root, for
$j=d-1,\ldots,0$ recompute $\rho^{[j+1]}(Y)$ from the original descriptor
$Y$ and evaluate the transition derivation on the parameter descriptors, that
lower-address descriptor, and the current state descriptor.  The declared
bound is a static admissibility certificate and is not evaluated.

Only a fixed number of descriptors and $O(\log N)$-bit counters are live at
any time.  Simultaneous induction on the fixed derivation shows that both the
number of symbolic primitive steps and every descriptor field are bounded by a
polynomial in $N$: composition takes fixed finite sums, while a recursion node
performs $O_\delta(N)$ transition evaluations and at most $O_\delta(N^2)$
descriptor-descent steps when addresses are recomputed.  Offset fields grow
only through successor and are divided, up to a fixed remainder, by descent.
Thus every field uses $O_\delta(\log N)$ work bits.

Finally, a constant descriptor has $O(\log N)$ output bits, while a source
descriptor is the input coordinate shifted right by $la$ bits and then increased
by an $O(\log N)$-bit integer.  Its output length and each output bit can be
produced by repeated scans of the input while storing the shift, bit position,
and carry in logarithmic space.  Hence the denoted numerical function is in
$\mathsf{FL}$.
\end{proof}

\begin{corollary}[divisibility inside functional logspace]
\label{cor:basis-zero-fl}
For every $l\ge1$,
\[
 \boxed{\Hclass{0}{1}{l}\subsetneq\mathsf{FL}.}
\]
Consequently the stride sector $\{\Hclass{0}{1}{l}:l\ge1\}$ realizes the
dual-divisibility partial order entirely among proper subclasses of
$\mathsf{FL}$.
\end{corollary}

\begin{proof}
The inclusion is Theorem~\ref{thm:basis-zero-symbolic-logspace}.  It is strict
because binary addition belongs to $\mathsf{FL}$, whereas
Lemma~\ref{lem:majorants} gives every binary member of
$\Hclass{0}{1}{l}$ a basis-zero composition majorant, hence a constant or a
shifted projection.  No such function majorizes $x+y$.  The final statement
combines strictness with Theorem~\ref{thm:basis-zero-doubling-main}.
\end{proof}

Throughout the remainder of this section write
\[
 \rho=\rho_{1,1},\qquad D=D_{1,1},
 \qquad \rho(x)=\left\lfloor\frac x2\right\rfloor.
\]

\begin{lemma}[normalized squaring]
\label{lem:normalized-squaring}
The halving iterator
\[
 \operatorname{Half}(u,0)=u,\qquad \operatorname{Half}(u,z)=\rho(\operatorname{Half}(u,\rho(z)))\quad(z>0)
\]
belongs to $\Hclass{2}{1}{1}$ and satisfies
\[
 \operatorname{Half}(u,z)=\rho^{[D(z)]}(u)
       =\left\lfloor\frac{u}{2^{D(z)}}\right\rfloor.
\]
There is also a function $\Theta\in\Hclass{2}{1}{1}$ such that
\[
 \Theta(0)=0,\qquad
 \Theta(x)=\left\lfloor\frac{x^2+2}{2^{D(x)}}\right\rfloor
 \quad(x>0).
\]
For $x\ge2$ one has $1\le\Theta(x)<x$.  On every dyadic layer
\[
 I_L=\{2^{L-1},\ldots,2^L-1\}\qquad(L\ge1)
\]
the map $\Theta$ is strictly increasing, and
\[
 D(\Theta(x))\in\{L-1,L\}\qquad(x\in I_L,\ L\ge2).
\]
\end{lemma}

\begin{proof}
The displayed recursion for $\operatorname{Half}$ is bounded by the projection $u$.
Induction on $D(z)$ gives
\[
 \operatorname{Half}(u,z)=\rho^{[D(z)]}(u)
       =\left\lfloor\frac{u}{2^{D(z)}}\right\rfloor.
\]
Since $e_2(x)=x^2+2$ belongs to $B_2$, put
\[
 U(x)=\operatorname{Half}(e_2(x),x),\qquad
 \Theta(x)=\operatorname{ZV}(U(x),x).
\]
Lemma~\ref{lem:internal-zero-value-selector} makes this construction internal
and gives the stated values at $x=0$ and $x>0$.

Fix $x\in I_L$.  Then $2^{L-1}\le x<2^L$, and
\[
 \Theta(x)=\left\lfloor\frac{x^2+2}{2^L}\right\rfloor.
\]
The inequality $x^2+2<(x+1)2^L$ gives $\Theta(x)\le x$.  If $x=2$,
then $L=2$ and $x(2^L-x)=4>2$.  If $x\ge3$, then $2^L-x\ge1$ and
$x(2^L-x)\ge x>2$.  Thus in every case $x\ge2$ one has
\[
 x(2^L-x)>2,
\]
so $x^2+2<x2^L$ and therefore $\Theta(x)<x$.  For $L\ge2$,
\[
 \Theta(x)\ge
 \left\lfloor\frac{2^{2L-2}+2}{2^L}\right\rfloor
 \ge2^{L-2},
\]
while $\Theta(x)<2^L$, proving the depth assertion.  Finally,
\[
 \frac{(x+1)^2+2}{2^L}-\frac{x^2+2}{2^L}
 =\frac{2x+1}{2^L}>1,
\]
so consecutive arguments in $I_L$ have distinct increasing images.
\end{proof}

\begin{lemma}[quadratic iterate separation]
\label{lem:theta-separation}
Let $x\ne y$ and put $L=\max\{D(x),D(y)\}$.  There is
$k<4(L+1)^2$ such that
\[
 D(\Theta^{[k]}(x))\ne D(\Theta^{[k]}(y)).
\]
Moreover, at the least such $k$,
\[
 x<y
 \quad\Longleftrightarrow\quad
 D(\Theta^{[k]}(x))<D(\Theta^{[k]}(y)).
\]
\end{lemma}

\begin{proof}
If $D(x)\ne D(y)$, take $k=0$.  Suppose both values lie in $I_L$.
As long as their depths agree, Lemma~\ref{lem:normalized-squaring} preserves
their strict numerical order.

It remains to bound the time spent in one layer.  Write
$x=2^L-a$ with $1\le a\le2^{L-1}$.  While $\Theta(x)$ remains in $I_L$,
its new deficit is
\[
 a'=2a-\left\lfloor\frac{a^2+2}{2^L}\right\rfloor.
\]
Since two distinct values lie in one common layer, $L\ge2$.  If $a=1$, then
$a'=2\ge\frac43a$.  If $a\ge2$, then
\[
 \frac{a^2+2}{2^L}\le\frac a2+\frac12,
 \qquad
 a'\ge2a-\left\lceil\frac a2\right\rceil\ge\frac43a.
\]
Thus the deficit grows by a factor at least $4/3$ while the iterate sequence remains in
$I_L$.  Since $(4/3)^{3L}>2^{L-1}$, the depth drops after fewer than $3L+2$
iterations.  Summing over the layers below $L$ shows that every positive
iterate sequence reaches $1$ in fewer than $4(L+1)^2$ iterations.

If two distinct values had equal depths throughout this period, strict
increase on every common layer would keep their iterates distinct, although
both would eventually equal the unique point $1$ of $I_1$.  Hence a first
depth difference occurs.  At that point the numerical order, preserved up to
the preceding step, and monotonicity of $D$ give the stated direction.
\end{proof}

\begin{lemma}[block-scan correctness]
\label{lem:block-scan}
For parameters $x,y,c$ with $c>0$, put $C=D(c)$.  There are functions
$J,K,\operatorname{Scan},\operatorname{Cmp}\in\Hclass{2}{1}{1}$ with
\[
 J(u,z)=\Theta^{[D(z)]}(u),
 \qquad
 K(u,c,z)=\Theta^{[CD(z)]}(u),
\]
and the following property.  The value
$\operatorname{Scan}(s,x,y,c,z,c)$ retains a nonzero initial state
$s\in\{1,2\}$, and otherwise returns the first nonzero depth-comparison code
for the indices
\[
 CD(z),CD(z)+1,\ldots,CD(z)+C-1,
\]
where code $1$ means that the $x$-iterate has smaller depth and code $2$
means that the $y$-iterate has smaller depth.  Moreover,
$\operatorname{Cmp}(x,y,c,z)$ returns the first nonzero code among the
indices
\[
 0,1,\ldots,CD(z)-1.
\]
The functions $J$ and $K$ are bounded by $u$, while
$\operatorname{Scan}$ and $\operatorname{Cmp}$ are bounded by $2$.
\end{lemma}

\begin{proof}
Define
\[
 J(u,0)=u,
 \qquad
 J(u,z)=\Theta(J(u,\rho(z)))\quad(z>0).
\]
Since $\Theta(u)\le u$, the projection $u$ is a bound, and induction on
$D(z)$ gives $J(u,z)=\Theta^{[D(z)]}(u)$.  Next define
\[
 K(u,c,0)=u,
 \qquad
 K(u,c,z)=J(K(u,c,\rho(z)),c)\quad(z>0).
\]
The same projection is a bound and induction gives
$K(u,c,z)=\Theta^{[CD(z)]}(u)$.

Depth comparison is internal before the scan is defined.  Indeed,
\[
 \operatorname{DLE}(u,v):=\ISZ(\operatorname{Half}(u,v))
\]
satisfies
\[
 \operatorname{Half}(u,v)=0
 \quad\Longleftrightarrow\quad
 u<2^{D(v)}
 \quad\Longleftrightarrow\quad
 D(u)\le D(v).
\]
Using the Boolean controls of Lemma~\ref{lem:finite-controls}, put
\[
 \operatorname{DLT}(u,v)
 :=\operatorname{DLE}(u,v)\wedge
   \ISZ(\operatorname{DLE}(v,u)).
\]
Thus $\operatorname{DLT}(u,v)=1$ exactly when $D(u)<D(v)$.  A fixed table
constructs
\[
 \operatorname{Upd}(s,u,v)=
 \begin{cases}
 s,&s\in\{1,2\},\\
 1,&s=0\text{ and }\operatorname{DLT}(u,v)=1,\\
 2,&s=0\text{ and }\operatorname{DLT}(v,u)=1,\\
 0,&\text{otherwise}.
 \end{cases}
\]
This map is bounded by $2$.  Define
\[
\begin{aligned}
 \operatorname{Scan}(s,x,y,c,z,0)&=s,\\
 \operatorname{Scan}(s,x,y,c,z,w)
 &=\operatorname{Upd}\Bigl(
 \operatorname{Scan}(s,x,y,c,z,\rho(w)),\\
 &\hspace{19mm}J(K(x,c,z),\rho(w)),
 J(K(y,c,z),\rho(w))\Bigr)
 \qquad(w>0).
\end{aligned}
\]
At a node with $D(w)=r>0$, the newly inspected pair is
\[
 \Theta^{[CD(z)+r-1]}(x),
 \qquad
 \Theta^{[CD(z)+r-1]}(y).
\]
Induction on $r$ proves the asserted first-nonzero rule.  At $w=c$ this is the
displayed block of length $C$.

Finally define
\[
 \operatorname{Cmp}(x,y,c,0)=0,
\]
\[
 \operatorname{Cmp}(x,y,c,z)=
 \operatorname{Scan}\bigl(
 \operatorname{Cmp}(x,y,c,\rho(z)),x,y,c,\rho(z),c\bigr)
 \qquad(z>0).
\]
If $D(z)=r>0$, the preceding state has already scanned indices
$0,\ldots,C(r-1)-1$, while the new inner call scans the next block
$C(r-1),\ldots,Cr-1$.  Induction on $r$ proves the final assertion.  The
states of both recursions lie in $\{0,1,2\}$, so the constant bound $2$ is
valid.  The displayed dependency order
\[
 \Theta\prec J\prec K,\operatorname{DLE},\operatorname{DLT}
 \prec\operatorname{Upd}\prec
 \operatorname{Scan}\prec\operatorname{Cmp}
\]
is compatible with the finite-stage closure.
\end{proof}

\begin{theorem}[comparison and parity at initial basis two]
\label{thm:base-two-comparison}
The predicates
\[
 \operatorname{LT}(x,y)=[x<y],
 \qquad
 \operatorname{LE}(x,y)=[x\le y]
\]
belong to $\Hclass{2}{1}{1}$.  Consequently ordinary parity belongs to
$\Hclass{2}{1}{1}$.
\end{theorem}

\begin{proof}
Put
\[
 u=S^{[16]}(x+y),
 \qquad
 c=e_2^{[2]}(u),
 \qquad
 R=D(u),
 \qquad
 C=D(c),
\]
and let $L=\max\{D(x),D(y)\}$.  Monotonicity gives $R\ge L$, while
$u\ge16$ gives $R\ge5$.  For every positive $v$,
\[
 D(e_2(v))\ge2D(v)-1:
\]
if $r=D(v)$, then $v\ge2^{r-1}$ and $v^2+3>2^{2r-2}$.
Applying the estimate twice yields $C\ge4R-3$.  Hence in all cases
\[
 C^2\ge4(L+1)^2.
\]
By Lemma~\ref{lem:block-scan},
$\operatorname{Cmp}(x,y,c,c)$ inspects the depth pairs for all indices
$0,\ldots,C^2-1$ in increasing order and retains the first unequal one.
Lemma~\ref{lem:theta-separation} guarantees such an index below
$4(L+1)^2$ whenever $x\ne y$, and identifies the numerical order from the
first unequal depth.  Therefore
\[
 \operatorname{Cmp}(x,y,c,c)=
 \begin{cases}
 1,&x<y,\\
 2,&x>y,\\
 0,&x=y.
 \end{cases}
\]
A fixed table extracts $\operatorname{LT}$ and $\operatorname{LE}$.

Finally put $q=\rho(x)$ and $d=q+q$.  Then $d=x$ for even $x$ and
$d=x-1$ for odd $x$, so
\[
 x\bmod2=\operatorname{LT}(d,x).
\]
Thus parity is internal.
\end{proof}

\begin{theorem}[the base-two vertical bridge]
\label{thm:base-two-first-bridge}
For every $l'\ge1$,
\[
 \rho_{2,1}\in\Hclass{2}{1}{l'}.
\]
\end{theorem}

\begin{proof}
By horizontal collapse it is enough to work in $\Hclass{2}{1}{1}$.  Define
\[
 A(0)=0,\qquad
 A(z)=S(A(\rho(z))+A(\rho(z)))\quad(z>0).
\]
The recursion has bound $2z+1$ and satisfies
\[
 A(z)=2^{D(z)}-1.
\]
Fix $y$ and put $a=A(y)$.  At a lower recursion node $u$, define the
descending binary weight
\[
 p(a,u)=S(\rho(\operatorname{Half}(a,u))).
\]
As $D(u)$ runs from $0$ to $D(a)-1$, these weights are
\[
 2^{D(y)-1},\ 2^{D(y)-2},\ldots,1
\]
in the evaluation order from the recursion base to its root.

Define a recursion with state $s$ by
\[
 B(y,a,0)=0
\]
and, for $z>0$, with $u=\rho(z)$,
\[
 \begin{aligned}
 p&=p(a,u),\\
 w&=B(y,a,u)+p,\\
 B(y,a,z)&=
 \begin{cases}
 w,&e_2(w)<y,\\
 B(y,a,u),&e_2(w)\ge y.
 \end{cases}
 \end{aligned}
\]
The comparison is supplied by Theorem~\ref{thm:base-two-comparison}.  The
selector has an affine earlier bound, while the recursion itself is bounded
by $y$: every accepted value satisfies $w^2+2<y$, hence $w<y$.

Put $D_0=D(y)=D(a)$ and let $s_t$ be the state after the first $t$
transitions from the base toward the root.  Induction on $t$ shows that
$s_t$ is the largest processed multiple of $2^{D_0-t}$ satisfying
$e_2(s_t)<y$, or $0$ if no such multiple exists.  At the last stage,
\[
 B(y,a,a)=\max\bigl(\{w<2^{D_0}:w^2+2<y\}\cup\{0\}\bigr).
\]
Since $a=2^{D(y)}-1\ge y$, every strict subthreshold witness is included.

Put $V(y)=S(B(y,A(y),A(y)))$.  The fixed-threshold predicate
$\operatorname{GE}_3(y)$ belongs to the class by
Lemma~\ref{lem:finite-controls}.  Define
\[
 \operatorname{Patch}_3(y,0)=0,
 \qquad
 \operatorname{Patch}_3(y,c)=V(y)\quad(c>0).
\]
This is one doubling-step recursion on $c$, bounded by $S(y)$.  Hence
\[
 \widehat\rho(y)
 =\operatorname{Patch}_3\bigl(y,\operatorname{GE}_3(y)\bigr)
 =
 \begin{cases}
 0,&y\le2,\\
 S(B(y,A(y),A(y))),&y\ge3.
 \end{cases}
\]
For $y\ge3$, the successor of the largest strict subthreshold integer is the
least $z$ satisfying $z^2+2\ge y$; for $y\le2$, both functions are $0$.
Thus $\widehat\rho=\rho_{2,1}$.
\end{proof}

\begin{lemma}[an exact linear-depth ladder in the doubling row]
\label{lem:row-one-exact-depth}
For every $m\ge3$ there is a function $Q\in\Hclass{m}{1}{1}$ satisfying
\[
 D(Q(y))=y\qquad(y\in\Nat).
\]
One may take
\[
 Q(0)=0,
 \qquad
 Q(y)=\rho\bigl(D(e_3(y))\bigr)
 \quad(y>0).
\]
\end{lemma}

\begin{proof}
The generator $e_3$ belongs to $B_m$ when $m\ge3$, while $D$ and $\rho$
are internal by Lemma~\ref{lem:internal-controls}.  Put
\[
 U(y)=\rho(D(e_3(y))),\qquad
 Q(y)=\operatorname{ZV}(U(y),y).
\]
Lemma~\ref{lem:internal-zero-value-selector} gives $Q(0)=0$ and
$Q(y)=U(y)$ for $y>0$.

Put $a_y=e_3(y)$.  Then $a_1=6$ and $a_{y+1}=a_y^2+2$.  For every $y\ge1$,
\begin{equation}
\label{eq:e3-binary-window}
 2^{2^y}<a_y,
 \qquad
 a_y+1\le2^{2^{y+1}-1}.
\end{equation}
The inequalities hold at $y=1$.  If they hold at $y$, then
\[
 a_{y+1}=a_y^2+2>2^{2^{y+1}},
\]
and
\[
 a_{y+1}+1=a_y^2+3<(a_y+1)^2
 \le2^{2(2^{y+1}-1)}
 <2^{2^{y+2}-1}.
\]
Thus \eqref{eq:e3-binary-window} follows by induction.

Let $L_y=D(a_y)=\lceil\log_2(a_y+1)\rceil$.  Then
\[
 2^y+1\le L_y\le2^{y+1}-1,
\]
so
\[
 2^{y-1}\le\left\lfloor\frac{L_y}{2}\right\rfloor\le2^y-1.
\]
Since $D(z)=\lceil\log_2(z+1)\rceil$, this is exactly
$D(\rho(L_y))=y$.  The origin was patched to $Q(0)=0$.
\end{proof}

\begin{theorem}[ordinary saturation in the doubling row]
\label{thm:row-one-saturation}
For every $l\ge1$,
\[
 \boxed{\Hclass{m}{1}{l}=\Eclass{m}\qquad(m\ge3).}
\]
\end{theorem}

\begin{proof}
The upper inclusion is Lemma~\ref{lem:calibration}.  Let $m\ge3$ and first
take $l=1$.  By Corollary~\ref{cor:bound-presentation-equivalence}, it is
enough to consider ordinary bounded recursion with a bound nondecreasing in
the recursion coordinate.  Let
\[
 f(\bar x,0)=g(\bar x),
 \qquad
 f(\bar x,t+1)=h(\bar x,t,f(\bar x,t)),
 \qquad
 f(\bar x,t)\le b(\bar x,t),
\]
where $g,h,b$ have already been constructed in $\Hclass{m}{1}{1}$ and $b$
is nondecreasing in its final coordinate.  Define
\[
 W(\bar x,0)=g(\bar x),
\]
\[
 W(\bar x,z)=
 h\bigl(\bar x,D(\rho(z)),W(\bar x,\rho(z))\bigr)
 \qquad(z>0).
\]
Because $D(\rho(z))=D(z)-1$ for $z>0$, induction on $D(z)$ gives
\[
 W(\bar x,z)=f(\bar x,D(z)).
\]
The earlier function
\[
 \widehat b(\bar x,z)=b(\bar x,D(z))
\]
is nondecreasing in $z$ and bounds $W$.  If $Q$ is supplied by
Lemma~\ref{lem:row-one-exact-depth}, then
\[
 W(\bar x,Q(y))=f(\bar x,D(Q(y)))=f(\bar x,y).
\]
Thus $\Hclass{m}{1}{1}$ is closed under ordinary bounded recursion, and
construction induction yields $\Eclass{m}\subseteq\Hclass{m}{1}{1}$.
Finally, Lemma~\ref{lem:horizontal-collapse} gives
$\Hclass{m}{1}{l}=\Hclass{m}{1}{1}$ for every $l\ge1$.
\end{proof}

\subsection*{Binary complexity at initial basis two}

The preceding results compare recursion strength inside the hierarchy.  We now
ask a different question: how large is the exceptional class
$\Hclass{2}{1}{*}$ under standard binary computation?

For a fixed arity, write
\[
 \|\bar x\|_{\rm bin}
 :=1+\sum_i\left\lceil\log_2(x_i+1)\right\rceil
\]
for the total binary input length.  Let $\mathsf{FP}$ denote the total
numerical functions computed by deterministic multitape Turing machines in
time polynomial in this length, with inputs and outputs in standard binary
notation.  The arity, $l$, and the derivation are fixed.  Let $\mathsf P$ and
$\mathsf{NP}$ have their usual binary-language meanings.  A numerical
predicate decides the language obtained from the standard binary encoding of
its input tuple.

\begin{theorem}[strict binary polynomial-time inclusion]
\label{thm:base-two-fp}
For every fixed $l\ge1$,
\[
 \boxed{\Hclass{2}{1}{l}\subsetneq\mathsf{FP}.}
\]
In particular, every Boolean-valued member of $\Hclass{2}{1}{l}$ decides a
language in $\mathsf P$.
\end{theorem}

\begin{proof}
Fix a derivation and let $N$ be the binary input length.  Induction on the
derivation shows simultaneously that evaluation is polynomial-time and that
all intermediate values have $O(N)$ bits.  This is immediate for $B_2$ and is
preserved by composition.  For a step recursion along
$\rho_{1,l}(y)=\lfloor y/2^l\rfloor$, Lemma~\ref{lem:majorants} gives, for
fixed $r,c$,
\[
 b(\bar x,z)\le e_2^{[r]}(M+c),\qquad M=\max(\bar x,z).
\]
Since fixed iterates of $e_2(t)=t^2+2$ are polynomials in $M$, admissibility
bounds every state by $O(N)$ bits.  The descent has
\[
 D_{1,l}(y)=\left\lceil\frac{\log_2(y+1)}l\right\rceil=O(N)
\]
stages, obtained by fixed right shifts.  Evaluating the base and then the
transition back along these $O(N)$ addresses makes only polynomially many
calls on $O(N)$-bit arguments, so the recursion remains polynomial-time.
Thus $\Hclass{2}{1}{l}\subseteq\mathsf{FP}$.

For strictness, recall $\lambda(x)=\lceil\log_2(x+1)\rceil$ and
$\operatorname{smash}(x,y)=2^{\lambda(x)\lambda(y)}$.  This function is in
$\mathsf{FP}$ by writing one followed by $\lambda(x)\lambda(y)$ zeroes.  On
the other hand, ordinary calibration and Lemma~\ref{lem:majorants} imply that
every $f\in\Hclass{2}{1}{l}\subseteq\Eclass{2}$ has a polynomial numerical
majorant in $M=\max_i x_i$.  On the diagonal,
$\operatorname{smash}(x,x)=2^{L^2}$ with $L=\lambda(x)$ and $x<2^L$, which
eventually exceeds every polynomial in $x$.  Hence
$\operatorname{smash}\notin\Eclass{2}$ and the inclusion is proper.
\end{proof}

\begin{corollary}[conditional strictness at initial basis two]
\label{cor:base-two-conditional-strictness}
For every fixed $l\ge1$,
\[
 \Hclass{2}{1}{l}\subseteq\Eclass{2}.
\]
Moreover,
\[
 \Hclass{2}{1}{l}=\Eclass{2}
 \quad\Longrightarrow\quad
 \mathsf P=\mathsf{NP}.
\]
Consequently, under the hypothesis $\mathsf P\ne\mathsf{NP}$,
\[
 \boxed{\Hclass{2}{1}{l}\subsetneq\Eclass{2}.}
\]
More generally, the inclusion is strict whenever
$\Eclass{2}\nsubseteq\mathsf{FP}$.
\end{corollary}

\begin{proof}
Ordinary calibration gives the inclusion, while
Theorem~\ref{thm:base-two-fp} puts its source in $\mathsf{FP}$.
Grozea~\cite[Definition~1, p.~271; Theorem~3, p.~273]{Grozea2004} defines
his $E_0$ from zero, successor, composition, and limited recursion, which is the
low-level closure needed for the inclusion used here.  In that framework he
defines the packed-array representation for CNF formulas and constructs a
Boolean SAT predicate in the corresponding lowest class, hence in
$\Eclass{0}\subseteq\Eclass{2}$ under the present notation.  In his postfix encoding, symbols are stored
in fixed-width binary fields.  With
$k=\lceil\log_2(v+4)\rceil$, a formula of $s$ symbols and $v$ variables uses
$O(sk+k+v)$ bits, hence polynomial size under the standard CNF encoding.
Converting a standard CNF instance to this packed postfix representation is a
polynomial-time symbol-by-symbol encoding and produces polynomially many bits.
Therefore equality $\Hclass{2}{1}{l}=\Eclass{2}$ would place this fixed SAT
predicate in $\mathsf{FP}$ after polynomial-time preprocessing,
so $\mathsf P=\mathsf{NP}$.  The remaining statements are immediate.
\end{proof}

\begin{corollary}[the base-two vertical bridge]
\label{cor:first-row-one-bridge}
For every $m\ge2$ and every $l,l'\ge1$,
\[
 \Hclass{m}{2}{l}\subseteq\Hclass{m}{1}{l'}.
\]
At $m=2$ the inclusion is strict.  For every $m\ge3$ both sides equal
$\Eclass{m}$.
\end{corollary}

\begin{proof}
Horizontal collapse reduces the source to $\Hclass{m}{2}{1}$ and the target
to $\Hclass{m}{1}{1}$.  Theorem~\ref{thm:base-two-first-bridge} and
\eqref{eq:basis-monotonicity} put the source descent $\rho_{2,1}$ in the target for every
$m\ge2$.  Moreover,
\[
 \rho_{2,1}(y)\le\left\lfloor\frac y2\right\rfloor=\rho_{1,1}(y).
\]
Indeed, for $y=2k$ one has $k^2+2\ge2k$, while for $y=2k+1$ one has
$k^2+2-(2k+1)=(k-1)^2\ge0$.  Lemma~\ref{lem:descent-comparison} gives
$D_{2,1}(y)\le D_{1,1}(y)$.  Exact-depth transfer,
Lemma~\ref{lem:depth-transfer}, with padding map $r(y)=y$, now simulates every
$\rho_{2,1}$-recursion in the doubling row.

At $m=2$, parity belongs to the target by
Theorem~\ref{thm:base-two-comparison}, whereas the finite-range trace
consequence excludes parity from $\Hclass{2}{2}{l}$.  Hence the inclusion is
strict.  For $m\ge3$, Theorems~\ref{thm:saturation} and
\ref{thm:row-one-saturation} identify both sides with $\Eclass{m}$.
\end{proof}

\begin{proposition}[proved comparisons involving $n=1$]
\label{prop:partial-row-one}
The following statements hold.
\begin{enumerate}
\item For all $l,l'\ge1$,
\[
 \Hclass{0}{1}{l}\subseteq\Hclass{0}{1}{l'}
 \quad\Longleftrightarrow\quad l'\mid l.
\]
\item If $m\ge1$ and $l,l'\ge1$, then
\[
  \Hclass{m}{1}{l}=\Hclass{m}{1}{l'}.
\]
\item For every $a,b\in\Nat$ and $l,l'\ge1$,
\[
  \Hclass{a}{1}{l}\subseteq\Hclass{b}{0}{l'}
  \quad\Longleftrightarrow\quad a\le b.
\]
\item If $m\ge2$, then
\[
 \Hclass{m}{2}{l}\subseteq\Hclass{m}{1}{l'}.
\]
This inclusion is strict at $m=2$ and is equality for $m\ge3$.
\item Let $n\ge2$ and $b\ge2$.  Then
\[
 \Hclass{a}{n}{l}\subseteq\Hclass{b}{1}{l'}
 \quad\Longleftrightarrow\quad a\le b.
\]
\item Let $a\ge2$ and $n'\ge2$.  Then
\[
 \Hclass{a}{1}{l}\subseteq\Hclass{b}{n'}{l'}
 \quad\Longleftrightarrow\quad
 a\le b\ \text{and}\ b\ge n'+1.
\]
\item For every $m\ge3$ and all $l,l'\ge1$,
\[
 \Hclass{m}{0}{l}=\Hclass{m}{1}{l'}=\Eclass{m}.
\]
\end{enumerate}
\end{proposition}

\begin{proof}
The first statement is Theorem~\ref{thm:basis-zero-doubling-main}; the second
is Lemma~\ref{lem:horizontal-collapse}; the third is the $n=1$ instance of
Corollary~\ref{cor:row-zero-comparisons}; the fourth is
Corollary~\ref{cor:first-row-one-bridge}; and the seventh is
Theorem~\ref{thm:row-one-saturation} together with
Proposition~\ref{prop:row-zero}.

For the fifth statement, necessity follows from Lemma~\ref{lem:base-separation}.
If $b=2$ and $a\le2$, Theorem~\ref{thm:global-classification} first embeds
the source into $\Hclass{2}{2}{1}$, and
Corollary~\ref{cor:first-row-one-bridge} embeds that class into
$\Hclass{2}{1}{l'}$.  If $b\ge3$, then
$\Hclass{b}{1}{l'}=\Eclass{b}$, so ordinary calibration and $a\le b$ give
the inclusion.

For the sixth statement, Lemma~\ref{lem:base-separation} gives $a\le b$.  If
$b\le n'$, the target lies in the trace band and therefore excludes parity,
whereas parity belongs to $\Hclass{2}{1}{1}$ by
Theorem~\ref{thm:base-two-comparison} and hence to every
$\Hclass{a}{1}{l}$ with $a\ge2$.  Thus $b\ge n'+1$ is necessary.  Under the
two displayed conditions, ordinary calibration gives
$\Hclass{a}{1}{l}\subseteq\Eclass{a}\subseteq\Eclass{b}$, and Theorem~\ref{thm:saturation} gives $\Eclass{b}=\Hclass{b}{n'}{l'}$.
\end{proof}

\begin{remark}[scope boundary for the doubling row]
\label{rem:row-one-open-families}
The complete classification theorem of Section~\ref{sec:classification-proof}
concerns $n,n'\ge2$ and is unaffected by the exceptional doubling row.  Within
$n=1$, Proposition~\ref{prop:partial-row-one} settles every equal-row stride
comparison and all cases needed for the structural transitions proved above.
What is not claimed here is a complete classification of the residual
cross-row cases in which the source or target has initial basis at most $1$.
Separately, at initial basis $2$ there remains a saturation question internal
to the doubling row,
\[
 \Hclass{2}{1}{*}\stackrel{?}{=}\Eclass{2}.
\]
This basis-two equality problem is distinct from those residual basis-$0/1$
cross-row comparisons.  We have
\[
 \Hclass{2}{2}{*}\subsetneq\Hclass{2}{1}{*}\subseteq\Eclass{2},
 \qquad \Hclass{2}{1}{*}\subsetneq\mathsf{FP}.
\]
Corollary~\ref{cor:base-two-conditional-strictness} makes the middle inclusion
strict under $\mathsf P\ne\mathsf{NP}$; the question mark records only the
absence of an unconditional separation.
\end{remark}

\section{Conclusion}

This paper shows that asymptotic recursion depth is not a complete invariant
of bounded recursive expressiveness.  Step recursion isolates a traversal
parameter while keeping the growth row and recursion format fixed.  For every
fixed $n\ge2$ and $m<n$, all fixed-stride descent depths are
$\Theta$-equivalent, but the corresponding classes form exactly the dual
divisibility order.  Thus one fixed Grzegorczyk sector already contains
infinite descending chains, infinite antichains, and every finite
partial-order pattern without any change of growth scale.  At $m=n$ this
traversal information disappears completely, and at $m=n+1$ the common class
becomes ordinary bounded recursion.  The classification therefore identifies
both a traversal-sensitive invariant and the exact threshold at which stronger
initial functions erase it.

The proof separates positive and negative mechanisms.  Exact-depth transfer
aligns descents even at the bottom basis; direct derivation induction gives
piecewise-monotone trace sparsity, while canonical zones and selected
numerical dependency chains yield reverse divisibility and vertical
separation.

Theorem~\ref{thm:global-classification} gives the complete $n,n'\ge2$
criterion and Corollary~\ref{cor:row-zero-comparisons} integrates row zero.
The exceptional doubling row is also sharp at its weakest basis:
$\Hclass{0}{1}{p}\subseteq\Hclass{0}{1}{q}$ holds exactly when $q\mid p$,
and every class in this dual-divisibility sector is a proper subclass of
$\mathsf{FL}$.  Thus the traversal order already appears inside deterministic
functional logspace, whereas all strides collapse from basis one onward.  It has the strict bridge
$\Hclass{2}{2}{*}\subsetneq\Hclass{2}{1}{*}$ and satisfies
$\Hclass{m}{1}{*}=\Eclass{m}$ for $m\ge3$.  At basis $2$ it lies properly in
$\mathsf{FP}$, while equality with $\Eclass{2}$ would imply
$\mathsf P=\mathsf{NP}$.  These low-basis doubling-row questions lie outside
the complete $n,n'\ge2$ classification and do not affect the traversal-order
theorem.

\section*{Acknowledgements}
The author is especially grateful to his wife, Stephanie Nekrasova, and his
father, Vladimir Osipov, for their unwavering support and encouragement.
Without their support, this work would not have been completed.

\section*{Data availability}
No data was used for the research described in this article.

\section*{Declaration of generative AI and AI-assisted technologies in the manuscript preparation process}
During the preparation of this work, the author used ChatGPT (OpenAI) for
copyediting, grammatical correction, and language polishing of the manuscript.  The author reviewed and edited the resulting
text as needed and takes full responsibility for the content of the article.

\clearpage

\appendix
\section{Inclusion diagrams}
\label{sec:inclusion-diagrams}

For reference, the table below gives the structural regimes proved in the
main text.  Figures~\ref{fig:global-inclusion-grid} and
\ref{fig:iteration-inclusions} then collect the detailed class relations; they
are not used in any proof.

\begin{center}
\small
\renewcommand{\arraystretch}{1.12}
\begin{tabular}{@{}>{\raggedright\arraybackslash}p{0.24\linewidth}>{\raggedright\arraybackslash}p{0.62\linewidth}@{}}
\toprule
\textbf{Sector} & \textbf{Classification} \\
\midrule
$n\ge2$, $m<n$ & equal-row inclusion is reverse divisibility: $\Hclass{m}{n}{l}\subseteq\Hclass{m}{n}{l'}$ iff $l'\mid l$ \\
$n\ge2$, $m=n$ & all strides give the same class \\
$n\ge2$, $m\ge n+1$ & $\Hclass{m}{n}{l}=\Eclass{m}$ \\
$n=1$, $m=0$ & equal-row inclusion is reverse divisibility \\
$n=1$, $m\ge1$ & all strides give the same class \\
$n=1$, $m\ge3$ & $\Hclass{m}{1}{l}=\Eclass{m}$ \\
\bottomrule
\end{tabular}
\end{center}

\clearpage
\begin{landscape}
\begin{figure}[p]
\centering
\resizebox{0.96\linewidth}{!}{%
\begin{tikzpicture}[
  >=Latex,
  x=1.16cm,y=0.52cm,
  cls/.style={draw,rounded corners=2pt,align=center,
              inner xsep=1.5pt,inner ysep=3pt,font=\scriptsize,
              minimum height=0.66cm,text width=1.22cm},
  ecls/.style={cls,very thick,minimum height=0.70cm},
  eqcls/.style={cls,font=\tiny,minimum height=0.56cm},
  bcls/.style={cls,thick,font=\tiny,minimum height=0.58cm},
  hcls/.style={cls,font=\tiny,minimum height=0.58cm},
  rowone/.style={cls,dashed,font=\tiny,minimum height=0.58cm},
  fpbox/.style={draw,double,rounded corners=5pt,very thick},
  edge/.style={->,thick},
  incl/.style={->,thick},
  eqedge/.style={<->,thick},
  comp/.style={->,thick},
  cont/.style={densely dotted,thick},
  rightcont/.style={->,thick},
  lab/.style={midway,fill=white,inner sep=1pt,font=\tiny},
  rowlab/.style={font=\scriptsize,anchor=east},
  ellipsis/.style={font=\scriptsize,fill=white,inner sep=1.2pt}
]
\node[ecls] (E0) at (0,23.50) {$\Eclass{0}$};
\node[ecls] (E1) at (2,23.50) {$\Eclass{1}$};
\node[ecls] (E2) at (4,23.50) {$\Eclass{2}$};
\node[ecls] (E3) at (6,23.50) {$\Eclass{3}$};
\node[ecls] (E4) at (8,23.50) {$\Eclass{4}$};
\node[ellipsis] (Ed) at (10,23.50) {$\cdots$};
\node[ecls] (Eim1) at (12,23.50) {$\Eclass{i-1}$};
\node[ecls] (Ei) at (14,23.50) {$\Eclass{i}$};
\node[ecls] (Eip1) at (16,23.50) {$\Eclass{i+1}$};
\node[ellipsis] (Edp) at (18,23.50) {$\cdots$};
\coordinate (Eright) at (19.35,23.50);
\node[font=\small\bfseries] (Gtitle) at (9,24.82)
 {Grzegorczyk hierarchy};
\foreach \a/\b in {E0/E1,E1/E2,E2/E3,E3/E4,Eim1/Ei,Ei/Eip1}
  \draw[edge] (\a) -- node[lab,above] {$\subsetneq$} (\b);
\draw[cont] (E4) -- (Ed); \draw[cont] (Ed) -- (Eim1);
\draw[rightcont] (Eip1) -- (Eright);
\node[draw,rounded corners=4pt,very thick,
      fit=(Gtitle)(E0)(E1)(E2)(E3)(E4)(Ed)(Eim1)(Ei)(Eip1)(Edp)(Eright),
      inner xsep=7pt,inner ysep=6pt] {};

\node[eqcls] (Q00) at (0,20.50) {$\Hclass{0}{0}{*}$};
\node[eqcls] (Q10) at (2,20.50) {$\Hclass{1}{0}{*}$};
\node[eqcls] (Q20) at (4,20.50) {$\Hclass{2}{0}{*}$};
\node[eqcls] (Q30) at (6,20.50) {$\Hclass{3}{0}{*}$};
\node[eqcls] (Q40) at (8,20.50) {$\Hclass{4}{0}{*}$};
\node[ellipsis] (Q0d) at (10,20.50) {$\cdots$};
\node[eqcls] (Qim10) at (12,20.50) {$\Hclass{i-1}{0}{*}$};
\node[eqcls] (Qi0) at (14,20.50) {$\Hclass{i}{0}{*}$};
\node[eqcls] (Qip10) at (16,20.50) {$\Hclass{i+1}{0}{*}$};
\node[ellipsis] (Q0dp) at (18,20.50) {$\cdots$};
\coordinate (Q0right) at (19.35,20.50);
\node[rowlab] at (-1.05,20.50) {$n=0$};
\foreach \a/\b in {E0/Q00,E1/Q10,E2/Q20,E3/Q30,E4/Q40,Eim1/Qim10,Ei/Qi0,Eip1/Qip10}
  \draw[eqedge] (\a.south) -- node[lab,right,xshift=3pt] {$=$} (\b.north);
\foreach \a/\b in {Q00/Q10,Q10/Q20,Q20/Q30,Q30/Q40,Qim10/Qi0,Qi0/Qip10}
  \draw[edge] (\a) -- node[lab,above] {$\subsetneq$} (\b);
\draw[cont] (Q40) -- (Q0d); \draw[cont] (Q0d) -- (Qim10);
\draw[rightcont] (Qip10) -- (Q0right);

\node[rowone] (R01g) at (0,17.80) {$\Hclass{0}{1}{l}$};
\node[rowone] (R11g) at (2,17.80) {$\Hclass{1}{1}{*}$};
\node[bcls] (R21g) at (4,17.80) {$\Hclass{2}{1}{*}$};
\node[eqcls] (R31g) at (6,17.80) {$\Hclass{3}{1}{*}$};
\node[eqcls] (R41g) at (8,17.80) {$\Hclass{4}{1}{*}$};
\node[ellipsis] (R1d) at (10,17.80) {$\cdots$};
\node[eqcls] (Rim11g) at (12,17.80) {$\Hclass{i-1}{1}{*}$};
\node[eqcls] (Ri1g) at (14,17.80) {$\Hclass{i}{1}{*}$};
\node[eqcls] (Rip11g) at (16,17.80) {$\Hclass{i+1}{1}{*}$};
\node[ellipsis] (R1dp) at (18,17.80) {$\cdots$};
\coordinate (R1right) at (19.35,17.80);
\node[rowlab] at (-1.05,17.80) {$n=1$};
\foreach \a/\b in {R01g/R11g,R11g/R21g,R21g/R31g}
  \draw[comp] (\a) -- node[lab,above] {$\subsetneq$} (\b);
\foreach \a/\b in {R31g/R41g,Rim11g/Ri1g,Ri1g/Rip11g}
  \draw[edge] (\a) -- node[lab,above] {$\subsetneq$} (\b);
\draw[cont] (R41g) -- (R1d); \draw[cont] (R1d) -- (Rim11g);
\draw[rightcont] (Rip11g) -- (R1right);
\foreach \a/\b in {R01g/Q00,R11g/Q10}
  \draw[incl] (\a.north) -- node[lab,right,xshift=3pt] {$\subseteq$} (\b.south);
\draw[incl] (R21g.north) -- node[lab,right,xshift=3pt] {$\subseteq$} (Q20.south);
\foreach \a/\b in {R31g/Q30,R41g/Q40,Rim11g/Qim10,Rip11g/Qip10}
  \draw[eqedge] (\a.north) -- node[lab,right,xshift=3pt] {$=$} (\b.south);
\draw[eqedge] (Ri1g.north) --
  node[lab,right,xshift=3pt] {$=$} (Qi0.south);

\node[hcls] (H02) at (0,15.00) {$\Hclass{0}{2}{1}$};
\node[hcls] (H12) at (2,15.00) {$\Hclass{1}{2}{1}$};
\node[bcls] (H22) at (4,15.00) {$\Hclass{2}{2}{*}$};
\node[eqcls] (H32) at (6,15.00) {$\Hclass{3}{2}{*}$};
\node[eqcls] (H42) at (8,15.00) {$\Hclass{4}{2}{*}$};
\node[ellipsis] (H2d) at (10,15.00) {$\cdots$};
\node[eqcls] (Him12) at (12,15.00) {$\Hclass{i-1}{2}{*}$};
\node[eqcls] (Hi2) at (14,15.00) {$\Hclass{i}{2}{*}$};
\node[eqcls] (Hip12) at (16,15.00) {$\Hclass{i+1}{2}{*}$};
\node[ellipsis] (H2dp) at (18,15.00) {$\cdots$};
\coordinate (H2right) at (19.35,15.00);
\node[rowlab] at (-1.05,15.00) {$n=2$};
\draw[edge] (H02) -- node[lab,above] {$\subsetneq$} (H12);
\draw[edge] (H12) -- node[lab,above] {$\subsetneq$} (H22);
\draw[comp] (H22) -- node[lab,above] {$\subsetneq$} (H32);
\draw[edge] (H32) -- node[lab,above] {$\subsetneq$} (H42);
\foreach \a/\b in {Him12/Hi2,Hi2/Hip12}
  \draw[edge] (\a) -- node[lab,above] {$\subsetneq$} (\b);
\draw[rightcont] (Hip12) -- (H2right);
\begin{scope}[on background layer]
\draw[edge] (H02.east) .. controls +(0.85,1.05) and +(0.85,-1.05) .. (Q00.east);
\draw[edge] (H12.east) .. controls +(0.95,1.05) and +(0.95,-1.05) .. (Q10.east);
\end{scope}
\node[font=\tiny,fill=white,inner sep=0.8pt] at (0.88,20.00) {$\subsetneq$};
\node[font=\tiny,fill=white,inner sep=0.8pt] at (2.88,20.00) {$\subsetneq$};
\draw[cont] (H42) -- (H2d); \draw[cont] (H2d) -- (Him12);
\draw[comp] (H22.north) -- node[lab,left,xshift=-3pt] {$\subsetneq$} (R21g.south);
\foreach \a/\b in {H32/R31g,H42/R41g,Him12/Rim11g,Hip12/Rip11g}
  \draw[eqedge] (\a.north) -- node[lab,right,xshift=3pt] {$=$} (\b.south);
\draw[eqedge] (Hi2.north) --
  node[lab,right,xshift=3pt] {$=$} (Ri1g.south);
\node[hcls] (H03) at (0,12.80) {$\Hclass{0}{3}{1}$};
\node[hcls] (H13) at (2,12.80) {$\Hclass{1}{3}{1}$};
\node[hcls] (H23) at (4,12.80) {$\Hclass{2}{3}{1}$};
\node[bcls] (H33) at (6,12.80) {$\Hclass{3}{3}{*}$};
\node[eqcls] (H43) at (8,12.80) {$\Hclass{4}{3}{*}$};
\node[ellipsis] (H3d) at (10,12.80) {$\cdots$};
\node[eqcls] (Him13) at (12,12.80) {$\Hclass{i-1}{3}{*}$};
\node[eqcls] (Hi3) at (14,12.80) {$\Hclass{i}{3}{*}$};
\node[eqcls] (Hip13) at (16,12.80) {$\Hclass{i+1}{3}{*}$};
\node[ellipsis] (H3dp) at (18,12.80) {$\cdots$};
\coordinate (H3right) at (19.35,12.80);
\node[rowlab] at (-1.05,12.80) {$n=3$};
\foreach \a/\b in {H03/H13,H13/H23,H23/H33}
  \draw[edge] (\a) -- node[lab,above] {$\subsetneq$} (\b);
\draw[comp] (H33) -- node[lab,above] {$\subsetneq$} (H43);
\foreach \a/\b in {Him13/Hi3,Hi3/Hip13}
  \draw[edge] (\a) -- node[lab,above] {$\subsetneq$} (\b);
\foreach \a/\b in {H03/H02,H13/H12}
  \draw[comp] (\a.north) -- node[lab,left,xshift=-3pt] {$\subsetneq$} (\b.south);
\draw[edge] (H23.north) -- node[lab,left,xshift=-3pt] {$\subsetneq$} (H22.south);
\draw[edge] (H33.north) -- node[lab,right,xshift=3pt] {$\subsetneq$} (H32.south);
\foreach \a/\b in {H43/H42,Him13/Him12,Hi3/Hi2,Hip13/Hip12}
  \draw[eqedge] (\a.north) -- node[lab,right,xshift=3pt] {$=$} (\b.south);
\draw[cont] (H43) -- (H3d); \draw[cont] (H3d) -- (Him13);
\draw[rightcont] (Hip13) -- (H3right);

\node[hcls] (H04) at (0,10.60) {$\Hclass{0}{4}{1}$};
\node[hcls] (H14) at (2,10.60) {$\Hclass{1}{4}{1}$};
\node[hcls] (H24) at (4,10.60) {$\Hclass{2}{4}{1}$};
\node[hcls] (H34) at (6,10.60) {$\Hclass{3}{4}{1}$};
\node[bcls] (H44) at (8,10.60) {$\Hclass{4}{4}{*}$};
\node[ellipsis] (H4d) at (10,10.60) {$\cdots$};
\node[eqcls] (Him14) at (12,10.60) {$\Hclass{i-1}{4}{*}$};
\node[eqcls] (Hi4) at (14,10.60) {$\Hclass{i}{4}{*}$};
\node[eqcls] (Hip14) at (16,10.60) {$\Hclass{i+1}{4}{*}$};
\node[ellipsis] (H4dp) at (18,10.60) {$\cdots$};
\coordinate (H4right) at (19.35,10.60);
\node[rowlab] at (-1.05,10.60) {$n=4$};
\foreach \a/\b in {H04/H14,H14/H24,H24/H34,H34/H44}
  \draw[edge] (\a) -- node[lab,above] {$\subsetneq$} (\b);
\foreach \a/\b in {Him14/Hi4,Hi4/Hip14}
  \draw[edge] (\a) -- node[lab,above] {$\subsetneq$} (\b);
\foreach \a/\b in {H04/H03,H14/H13,H24/H23}
  \draw[comp] (\a.north) -- node[lab,left,xshift=-3pt] {$\subsetneq$} (\b.south);
\draw[edge] (H34.north) -- node[lab,left,xshift=-3pt] {$\subsetneq$} (H33.south);
\draw[edge] (H44.north) -- node[lab,right,xshift=3pt] {$\subsetneq$} (H43.south);
\foreach \a/\b in {Him14/Him13,Hi4/Hi3,Hip14/Hip13}
  \draw[eqedge] (\a.north) -- node[lab,right,xshift=3pt] {$=$} (\b.south);
\draw[cont] (H44) -- (H4d); \draw[cont] (H4d) -- (Him14);
\draw[rightcont] (Hip14) -- (H4right);

\node[ellipsis] at (0,8.50) {$\vdots$};
\node[ellipsis] at (2,8.50) {$\vdots$};
\node[ellipsis] at (4,8.50) {$\vdots$};
\node[ellipsis] at (6,8.50) {$\vdots$};
\node[ellipsis] at (8,8.50) {$\vdots$};
\node[ellipsis] at (12,8.50) {$\vdots$};
\node[ellipsis] at (14,8.50) {$\vdots$};
\node[ellipsis] at (16,8.50) {$\vdots$};
\node[ellipsis,anchor=east] at (-1.05,8.50) {$\vdots$};

\node[hcls] (H0im2) at (0,6.40) {$\Hclass{0}{i-2}{1}$};
\node[hcls] (H1im2) at (2,6.40) {$\Hclass{1}{i-2}{1}$};
\node[hcls] (H2im2) at (4,6.40) {$\Hclass{2}{i-2}{1}$};
\node[hcls] (H3im2) at (6,6.40) {$\Hclass{3}{i-2}{1}$};
\node[hcls] (H4im2) at (8,6.40) {$\Hclass{4}{i-2}{1}$};
\node[ellipsis] (Him2d) at (10,6.40) {$\cdots$};
\node[eqcls] (Him1im2) at (12,6.40) {$\Hclass{i-1}{i-2}{*}$};
\node[eqcls] (Hiim2) at (14,6.40) {$\Hclass{i}{i-2}{*}$};
\node[eqcls] (Hip1im2) at (16,6.40) {$\Hclass{i+1}{i-2}{*}$};
\node[ellipsis] (Him2dp) at (18,6.40) {$\cdots$};
\coordinate (Him2right) at (19.35,6.40);
\node[rowlab] at (-1.05,6.40) {$n=i-2$};
\foreach \a/\b in {H0im2/H1im2,H1im2/H2im2,H2im2/H3im2,H3im2/H4im2}
  \draw[edge] (\a) -- node[lab,above] {$\subsetneq$} (\b);
\foreach \a/\b in {Him1im2/Hiim2,Hiim2/Hip1im2}
  \draw[edge] (\a) -- node[lab,above] {$\subsetneq$} (\b);
\draw[cont] (H4im2) -- (Him2d); \draw[cont] (Him2d) -- (Him1im2);
\draw[rightcont] (Hip1im2) -- (Him2right);
\foreach \a/\b in {Him1im2/Him14,Hiim2/Hi4,Hip1im2/Hip14}
  \draw[cont] (\a.north) -- node[lab,right,xshift=3pt] {$=\cdots=$} (\b.south);

\node[hcls] (H0im1) at (0,4.30) {$\Hclass{0}{i-1}{1}$};
\node[hcls] (H1im1) at (2,4.30) {$\Hclass{1}{i-1}{1}$};
\node[hcls] (H2im1) at (4,4.30) {$\Hclass{2}{i-1}{1}$};
\node[hcls] (H3im1) at (6,4.30) {$\Hclass{3}{i-1}{1}$};
\node[hcls] (H4im1) at (8,4.30) {$\Hclass{4}{i-1}{1}$};
\node[ellipsis] (Him1d) at (10,4.30) {$\cdots$};
\node[bcls] (Him1im1) at (12,4.30) {$\Hclass{i-1}{i-1}{*}$};
\node[eqcls] (Hiim1) at (14,4.30) {$\Hclass{i}{i-1}{*}$};
\node[eqcls] (Hip1im1) at (16,4.30) {$\Hclass{i+1}{i-1}{*}$};
\node[ellipsis] (Him1dp) at (18,4.30) {$\cdots$};
\coordinate (Him1right) at (19.35,4.30);
\node[rowlab] at (-1.05,4.30) {$n=i-1$};
\foreach \a/\b in {H0im1/H1im1,H1im1/H2im1,H2im1/H3im1,H3im1/H4im1}
  \draw[edge] (\a) -- node[lab,above] {$\subsetneq$} (\b);
\draw[comp] (Him1im1) -- node[lab,above] {$\subsetneq$} (Hiim1);
\draw[edge] (Hiim1) -- node[lab,above] {$\subsetneq$} (Hip1im1);
\draw[cont] (H4im1) -- (Him1d); \draw[cont] (Him1d) -- (Him1im1);
\draw[rightcont] (Hip1im1) -- (Him1right);
\foreach \a/\b in {H0im1/H0im2,H1im1/H1im2,H2im1/H2im2,H3im1/H3im2,H4im1/H4im2}
  \draw[comp] (\a.north) -- node[lab,left,xshift=-3pt] {$\subsetneq$} (\b.south);
\draw[edge] (Him1im1.north) -- node[lab,right,xshift=3pt] {$\subsetneq$} (Him1im2.south);
\foreach \a/\b in {Hiim1/Hiim2,Hip1im1/Hip1im2}
  \draw[eqedge] (\a.north) -- node[lab,right,xshift=3pt] {$=$} (\b.south);

\node[hcls] (H0i) at (0,2.20) {$\Hclass{0}{i}{1}$};
\node[hcls] (H1i) at (2,2.20) {$\Hclass{1}{i}{1}$};
\node[hcls] (H2i) at (4,2.20) {$\Hclass{2}{i}{1}$};
\node[hcls] (H3i) at (6,2.20) {$\Hclass{3}{i}{1}$};
\node[hcls] (H4i) at (8,2.20) {$\Hclass{4}{i}{1}$};
\node[ellipsis] (Hid) at (10,2.20) {$\cdots$};
\node[hcls] (Him1i) at (12,2.20) {$\Hclass{i-1}{i}{1}$};
\node[bcls] (Hii) at (14,2.20) {$\Hclass{i}{i}{*}$};
\node[eqcls] (Hip1i) at (16,2.20) {$\Hclass{i+1}{i}{*}$};
\node[ellipsis] (Hidp) at (18,2.20) {$\cdots$};
\coordinate (Hiright) at (19.35,2.20);
\node[rowlab] at (-1.05,2.20) {$n=i$};
\foreach \a/\b in {H0i/H1i,H1i/H2i,H2i/H3i,H3i/H4i,Him1i/Hii}
  \draw[edge] (\a) -- node[lab,above] {$\subsetneq$} (\b);
\draw[comp] (Hii) -- node[lab,above] {$\subsetneq$} (Hip1i);
\draw[cont] (H4i) -- (Hid); \draw[cont] (Hid) -- (Him1i);
\draw[rightcont] (Hip1i) -- (Hiright);
\foreach \a/\b in {H0i/H0im1,H1i/H1im1,H2i/H2im1,H3i/H3im1,H4i/H4im1}
  \draw[comp] (\a.north) -- node[lab,left,xshift=-3pt] {$\subsetneq$} (\b.south);
\draw[edge] (Him1i.north) -- node[lab,left,xshift=-3pt] {$\subsetneq$} (Him1im1.south);
\draw[edge] (Hii.north) -- node[lab,right,xshift=3pt] {$\subsetneq$} (Hiim1.south);
\draw[eqedge] (Hip1i.north) -- node[lab,right,xshift=3pt] {$=$} (Hip1im1.south);

\node[hcls] (H0ip1) at (0,0.10) {$\Hclass{0}{i+1}{1}$};
\node[hcls] (H1ip1) at (2,0.10) {$\Hclass{1}{i+1}{1}$};
\node[hcls] (H2ip1) at (4,0.10) {$\Hclass{2}{i+1}{1}$};
\node[hcls] (H3ip1) at (6,0.10) {$\Hclass{3}{i+1}{1}$};
\node[hcls] (H4ip1) at (8,0.10) {$\Hclass{4}{i+1}{1}$};
\node[ellipsis] (Hip1d) at (10,0.10) {$\cdots$};
\node[hcls] (Him1ip1) at (12,0.10) {$\Hclass{i-1}{i+1}{1}$};
\node[hcls] (Hiip1) at (14,0.10) {$\Hclass{i}{i+1}{1}$};
\node[bcls] (Hip1ip1) at (16,0.10) {$\Hclass{i+1}{i+1}{*}$};
\node[ellipsis] (Hip1dp) at (18,0.10) {$\cdots$};
\coordinate (Hip1right) at (19.35,0.10);
\node[rowlab] at (-1.05,0.10) {$n=i+1$};
\foreach \a/\b in {H0ip1/H1ip1,H1ip1/H2ip1,H2ip1/H3ip1,H3ip1/H4ip1,Him1ip1/Hiip1,Hiip1/Hip1ip1}
  \draw[edge] (\a) -- node[lab,above] {$\subsetneq$} (\b);
\draw[cont] (H4ip1) -- (Hip1d); \draw[cont] (Hip1d) -- (Him1ip1);
\draw[rightcont] (Hip1ip1) -- (Hip1right);
\foreach \a/\b in {H0ip1/H0i,H1ip1/H1i,H2ip1/H2i,H3ip1/H3i,H4ip1/H4i,Him1ip1/Him1i}
  \draw[comp] (\a.north) -- node[lab,left,xshift=-3pt] {$\subsetneq$} (\b.south);
\draw[edge] (Hiip1.north) -- node[lab,left,xshift=-3pt] {$\subsetneq$} (Hii.south);
\draw[edge] (Hip1ip1.north) -- node[lab,right,xshift=3pt] {$\subsetneq$} (Hip1i.south);

\node[ellipsis] (V0)  at (0,-1.65)  {$\vdots$};
\node[ellipsis] (V1)  at (2,-1.65)  {$\vdots$};
\node[ellipsis] (V2)  at (4,-1.65)  {$\vdots$};
\node[ellipsis] (V3)  at (6,-1.65)  {$\vdots$};
\node[ellipsis] (V4)  at (8,-1.65)  {$\vdots$};
\node[ellipsis] (Vd)  at (10,-1.65) {$\vdots$};
\node[ellipsis] (Vim1) at (12,-1.65) {$\vdots$};
\node[ellipsis] (Vi)   at (14,-1.65) {$\vdots$};
\node[ellipsis] (Vip1) at (16,-1.65) {$\vdots$};
\node[rowlab] at (-1.05,-1.65) {$n\to\infty$};
\foreach \a/\b in {H0ip1/V0,H1ip1/V1,H2ip1/V2,H3ip1/V3,H4ip1/V4,Hip1d/Vd,Him1ip1/Vim1,Hiip1/Vi,Hip1ip1/Vip1}
  \draw[cont] (\a.south) -- (\b.north);

\begin{scope}[on background layer]
\node[fpbox,fit=(R01g)(R21g)(H0ip1)(H2ip1)(V0)(V1)(V2),
      inner xsep=7pt,inner ysep=6pt] (FPbox) {};
\end{scope}
\node[font=\scriptsize,fill=white,inner sep=1.5pt,anchor=west]
      at ([xshift=-1pt]FPbox.east)
      {$\subsetneq\mathsf{FP}$};

\end{tikzpicture}%
}
\begingroup
\sloppy
\setlength{\emergencystretch}{2em}
\caption{Proved inclusions and equalities among the displayed classes.
Here $i\ge7$ is a schematic index.  All labelled single arrows are solid;
their labels distinguish proved strict inclusions ($\subsetneq$) from
inclusions whose strictness is unresolved ($\subseteq$).  Double arrows denote
equality, with every equality sign placed to the right of its arrow.  Dotted
$=\cdots=$ segments denote omitted equality chains; other dotted segments and
ellipses suppress intermediate classes.  The
terminal right arrows and the bottom vertical ellipses show that the hierarchy
continues without bound as $m$ and $n$ increase.  Differences in node-border
style are used only for visual grouping and carry no additional
inclusion-theoretic meaning.  The double rectangle encloses classes contained
in $\Hclass{2}{1}{*}\subsetneq\mathsf{FP}$.}
\label{fig:global-inclusion-grid}
\endgroup
\fussy
\end{figure}
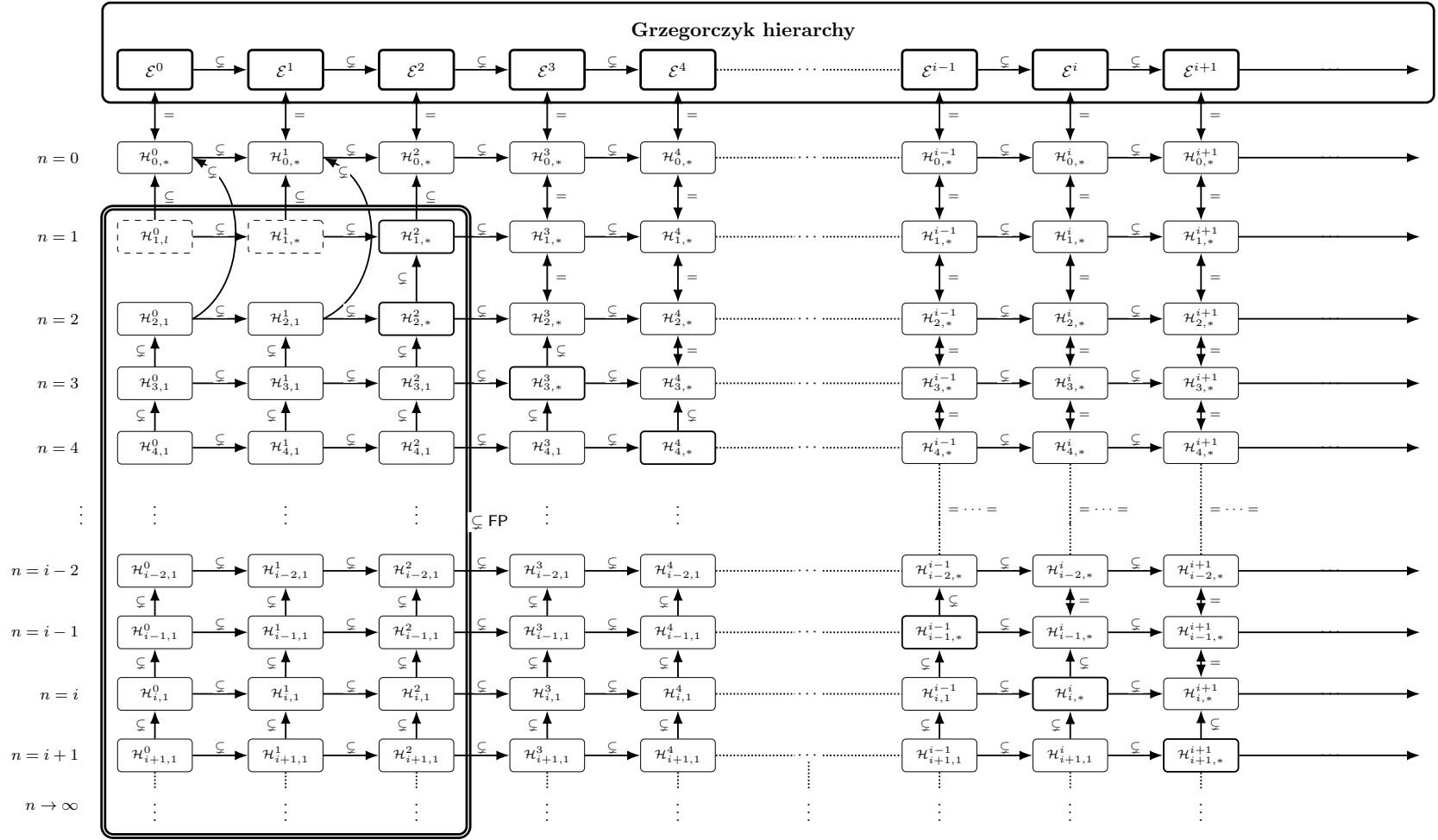
\end{landscape}
\clearpage

\begin{landscape}
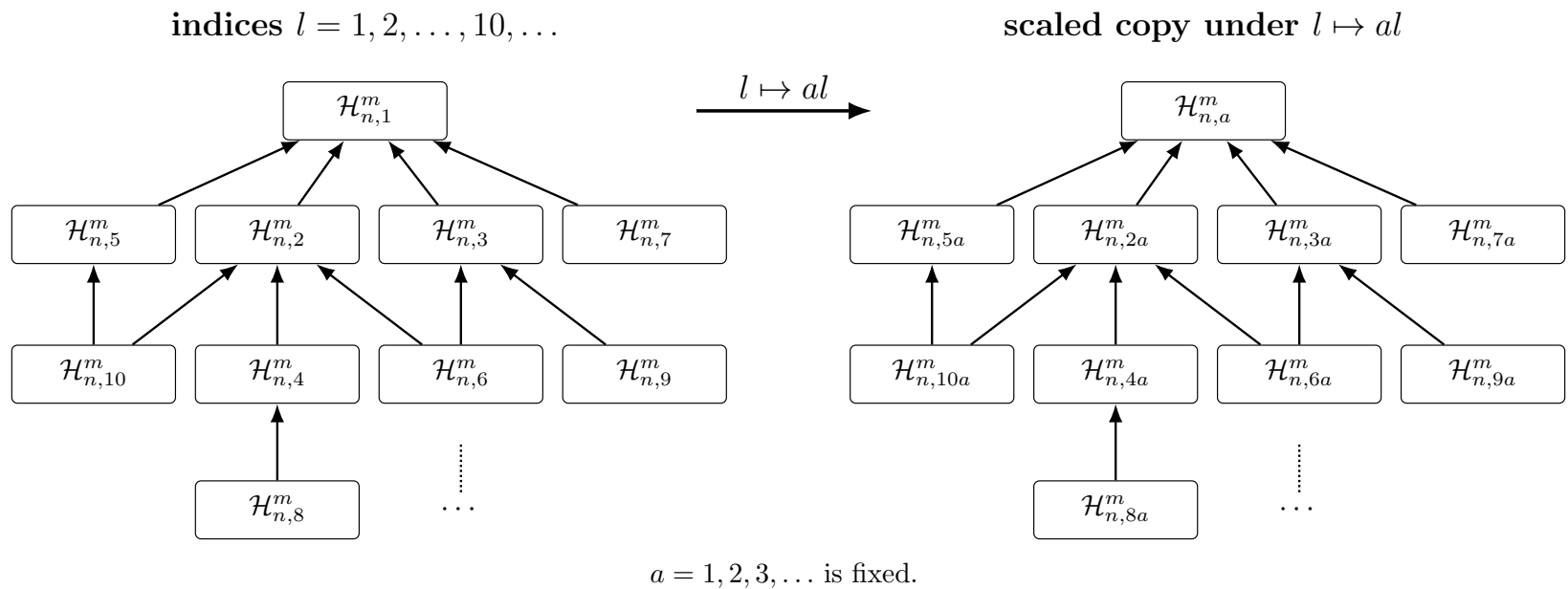
\begin{figure}[p]
\centering
\resizebox{0.94\linewidth}{!}{%
\begin{tikzpicture}[
 >=Latex,
 cls/.style={draw,rounded corners=2pt,align=center,
             inner xsep=6pt,inner ysep=5pt,font=\small,
             minimum width=2.05cm,minimum height=0.72cm},
 edge/.style={->,thick},
 mapedge/.style={->,very thick},
 cont/.style={densely dotted,thick},
 title/.style={font=\large\bfseries},
 note/.style={font=\small,align=center}
]
\node[title] at (-4.0,6.15) {indices $l=1,2,\ldots,10,\ldots$};
\node[cls] (L1)  at (-4.0,5.10) {$\Hclass{m}{n}{1}$};
\node[cls] (L5)  at (-7.4,3.55) {$\Hclass{m}{n}{5}$};
\node[cls] (L2)  at (-5.1,3.55) {$\Hclass{m}{n}{2}$};
\node[cls] (L3)  at (-2.8,3.55) {$\Hclass{m}{n}{3}$};
\node[cls] (L7)  at (-0.5,3.55) {$\Hclass{m}{n}{7}$};
\node[cls] (L10) at (-7.4,1.80) {$\Hclass{m}{n}{10}$};
\node[cls] (L4)  at (-5.1,1.80) {$\Hclass{m}{n}{4}$};
\node[cls] (L6)  at (-2.8,1.80) {$\Hclass{m}{n}{6}$};
\node[cls] (L9)  at (-0.5,1.80) {$\Hclass{m}{n}{9}$};
\node[cls] (L8)  at (-5.1,0.10) {$\Hclass{m}{n}{8}$};
\node (Ldot) at (-2.8,0.10) {$\cdots$};
\foreach \u/\v in {L2/L1,L3/L1,L5/L1,L7/L1,L4/L2,L6/L2,L6/L3,L10/L2,L10/L5,L9/L3,L8/L4}
  \draw[edge] (\u) -- (\v);
\draw[cont] (Ldot) -- (-2.8,0.95);

\node[title] at (6.5,6.15) {scaled copy under $l\mapsto al$};
\node[cls] (A1)  at (6.5,5.10) {$\Hclass{m}{n}{a}$};
\node[cls] (A5)  at (3.1,3.55) {$\Hclass{m}{n}{5a}$};
\node[cls] (A2)  at (5.4,3.55) {$\Hclass{m}{n}{2a}$};
\node[cls] (A3)  at (7.7,3.55) {$\Hclass{m}{n}{3a}$};
\node[cls] (A7)  at (10.0,3.55) {$\Hclass{m}{n}{7a}$};
\node[cls] (A10) at (3.1,1.80) {$\Hclass{m}{n}{10a}$};
\node[cls] (A4)  at (5.4,1.80) {$\Hclass{m}{n}{4a}$};
\node[cls] (A6)  at (7.7,1.80) {$\Hclass{m}{n}{6a}$};
\node[cls] (A9)  at (10.0,1.80) {$\Hclass{m}{n}{9a}$};
\node[cls] (A8)  at (5.4,0.10) {$\Hclass{m}{n}{8a}$};
\node (Adot) at (7.7,0.10) {$\cdots$};
\foreach \u/\v in {A2/A1,A3/A1,A5/A1,A7/A1,A4/A2,A6/A2,A6/A3,A10/A2,A10/A5,A9/A3,A8/A4}
  \draw[edge] (\u) -- (\v);
\draw[cont] (Adot) -- (7.7,0.95);

\draw[mapedge] (0.15,5.10) -- node[above,font=\large] {$l\mapsto al$} (2.35,5.10);
\node[note] at (1.25,-0.72) {$a=1,2,3,\ldots$ is fixed.};
\end{tikzpicture}%
}
\caption{Reverse divisibility for $n\ge2$ and $m<n$.  Each panel displays
a finite portion of the divisibility order; every arrow is a strict inclusion
toward the containing class.  In particular, both proper-divisor arrows from
$6$ to $2,3$ and from $10$ to $2,5$ are shown.  The right diagram is obtained
by the map $l\mapsto al$.  Further divisibility edges, such as $8\to2$ and
$10\to1$, are omitted.}
\label{fig:iteration-inclusions}
\end{figure}
\end{landscape}

\end{document}